\documentclass[11pt]{article}

\usepackage[margin=1in]{geometry}
\usepackage{amsmath,amssymb,amsthm}
\usepackage{mathrsfs}
\usepackage{enumitem}
\usepackage{microtype}
\usepackage{xcolor}
\usepackage{booktabs}
\usepackage{cite}

\usepackage[colorlinks=true,linkcolor=blue!50!black,citecolor=blue!50!black,urlcolor=blue!50!black]{hyperref}

\usepackage[capitalise,nameinlink]{cleveref}
\usepackage{mathtools}

\newtheorem{theorem}{Theorem}[section]

\newtheorem{definition}[theorem]{Definition}

\newtheorem{lemma}[theorem]{Lemma}
\newtheorem{remark}[theorem]{Remark}
\newtheorem{proposition}[theorem]{Proposition}
\newtheorem{corollary}[theorem]{Corollary}

\newtheorem{construction}[theorem]{Construction}

\newcommand{\F}{\mathbb F}
\newcommand{\N}{\mathbb N}
\newcommand{\R}{\mathbb R}
\newcommand{\C}{\mathbb C}

\newcommand{\cD}{\mathcal D}

\newcommand{\cF}{\mathcal F}

\newcommand{\cL}{\mathcal L}
\newcommand{\cM}{\mathcal M}
\newcommand{\cP}{\mathcal P}

\newcommand{\cV}{\mathcal V}

\newcommand{\cZ}{\mathcal Z}

\newcommand{\calL}{\mathcal L}

\newcommand{\eps}{\varepsilon}

\newcommand{\ip}[2]{\left\langle #1,#2\right\rangle}

\newcommand{\defeq}{\coloneqq}

\DeclareMathOperator{\rank}{rank}
\DeclareMathOperator{\row}{row}
\DeclareMathOperator{\im}{im}
\DeclareMathOperator{\codim}{codim}
\DeclareMathOperator{\wt}{wt}
\DeclareMathOperator{\Span}{span}

\DeclareMathOperator{\Prb}{Pr}
\DeclareMathOperator{\Exp}{\mathbb{E}}

\DeclareMathOperator{\Hom}{Hom}
\DeclareMathOperator{\Hamm}{Hamm}
\DeclareMathOperator{\poly}{poly}
\DeclareMathOperator{\Tr}{Tr}
\DeclareMathOperator{\dis}{dis}
\newcommand{\ket}[1]{\lvert #1\rangle}

\newcommand{\PCSS}{\operatorname{\mathsf{PCSS}}}
\newcommand{\Nest}{\operatorname{\mathsf{Nest}}}
\newcommand{\Flag}{\operatorname{Flag}}
\newcommand{\LDist}{\operatorname{LDist}}

\def\showauthornotes{1}

\ifnum\showauthornotes=0
\newcommand{\fnote}[1]{{\sf\small\color{blue}{[Fernando: #1]}}}
\newcommand{\nnote}[1]{{\sf\small\color{orange}{[Nikhil: #1]}}}
\newcommand{\xnote}[1]{\textcolor{green}{ {\textbf{(Xiaojuan: #1)}}}}
\else
\newcommand{\nnote}[1]{}
\newcommand{\fnote}[1]{}
\newcommand{\xnote}[1]{}
\fi

\title{From Random Quantum Codes to Explicit qLDPC Codes \\
        via Local Properties}

\author{
Fernando Granha Jeronimo\thanks{{\tt University of Illinois, Urbana-Champaign}. {\tt granha@illinois.edu}. }
\and
Xiaojuan Ma\thanks{{\tt University of Illinois, Urbana-Champaign}. {\tt xm20@illinois.edu}. }
\and
Nikhil Shagrithaya\thanks{{\tt nshagri@umich.edu}. }
}
\date{}

\begin{document}
\maketitle

\begin{abstract}
Constructing explicit codes matching the parameters of random codes has been a central and largely elusive question in coding theory. The quantum setting is even more challenging since it is highly desirable that the quantum code be an LDPC code.

Local coordinate-wise linear (LCL) [Levi, Mosheiff, and Shagrithaya, FOCS 2025] witnesses provide a unifying language for many coding-theoretic properties, from distance to list decoding and list recovery. In particular, it provides a framework to study properties of random linear codes, which achieve optimal parameters for many properties of linear codes.

For CSS quantum codes, however, a local witness has two distinct ranks: its physical rank before quotienting by stabilizers and its logical rank after quotienting. We develop a quantum version of the LCL framework for nested spaces $S \subseteq C$, in which local constraints are imposed on physical representatives while independence is measured in the logical quotient $C/S$. The resulting theory gives a threshold theorem for random CSS codes, and as a consequence shows that the per-sector rate threshold is equal to the classical rate threshold. We also define a quantum analogue of subspace design [Guruswami and Xing, STOC 2013] and show that they can be described in a natural manner within the quantum-LCL framework.

Finally, we give explicit constructions for arbitrary folded quantum-LCL properties, in a manner similar to the LCL derandomization of [Jeronimo and Shagrithaya, STOC 2026]. As a consequence, we obtain the first explicit constructions of quantum list-decodable codes and list-recoverable codes that have optimal list sizes, in addition to explicit quantum subspace design codes. We note that all our explicit constructions are qLDPC codes, an important property for quantum error-correcting codes.
\end{abstract}

\clearpage
\tableofcontents
\clearpage

\section{Introduction}\label{sec:introduction}

Random codes are known to achieve many optimal properties. However, constructing explicit such codes has been a central and largely elusive question in coding theory. The quantum setting is even more challenging since it is highly desirable that a 
quantum code be LDPC. Even the question of obtaining good quantum LDPC was only recently resolved in the breakthrough work of Panteleev and Kalachev~\cite{PK22}, whereas random quantum codes were long known to be good but not LDPC in the seminal work of Calderbank, Shor and Steane (CSS) \cite{CSS1,CSS2}.

LDPC codes feature very prominently in the field of quantum error-correction, where the locality of the parity checks is a key feature in allowing a code to be feasible for fault-tolerance \cite{Gottesman2014}. Therefore, considerable effort has been devoted towards developing quantum LDPC, or qLDPC codes. A long line of work gave various constructions starting from the seminal Kitaev toric code~\cite{Kitaev2003}, the hypergraph product code~\cite{TillichZemor2014}, constructions based on high-dimensional expanders \cite{EvraKaufmanZemor2020, KaufmanTessler2021}, fiber bundle construction \cite{HastingsHaahODonnell2021} and quasi-cyclic constructions \cite{PanteleevKalachev2022} and balanced product \cite{BreuckmannEberhardt2021}. The first asymptotically good qLDPC codes were obtained by Panteleev and Kalachev~\cite{PK22}, with subsequent unique decoding algorithms in later works~\cite{LZ22, DHLV23, GuPattisonTang2023}.

A central motivating question of this work is
\begin{center}
{\emph{Can we construct explicit quantum LDPC codes that possess many of the optimal properties of random quantum codes?}}
\end{center}

Before addressing the quantum challenges, we take a detour to through classical error correction literature which has recently developed an extensive toolkit to understand various properties of random codes and how to convert them into explicit codes. Classical coding theory has being an important source of ideas and inspiration for quantum codes, particularly through the CSS construction. Multiple code properties studied in the literature can be defined in terms of exclusion of small sets of vectors. For example, a code with minimum distance $\delta$ avoids containing pairs of vectors that have Hamming distance less than $\delta$. Other properties such as list decoding and list recovery can be defined in a similar manner.
Such code properties are commonly known as \emph{local properties}, and in a recent work, Levi, Mosheiff, and Shagrithaya~\cite{LMS25} created a framework to study them in the large alphabet regime. Known as the Local Coordinate-Wise Linear framework, it was used to prove the existence of rate thresholds for all LCL properties, in the context of random linear codes. In a later work, Jeronimo and Shagrithaya~\cite{JS26} used the Alon-Edmonds-Luby (AEL) construction procedure to obtain explicit constructions for all LCL properties. Owing to the use of AEL, these explicit constructions were LDPC codes.

A well-known object in the study of coding theory and pseudorandomness is that of subspace designs. In a seminal work, Guruswami and Xing introduced subspace designs in the study of algebraic list decoding, where one needs large collections of high-dimensional subspaces that have small total intersection with every low-dimensional test space~\cite{GX13}. Guruswami and Kopparty gave near-optimal explicit constructions over large fields~\cite{GK13}.  Subspace designs have been a central object in linear-algebraic pseudorandomness and in code constructions.  More recently, Brakensiek, Chen, Dhar, and Zhang showed that the subspace-design property simultaneously captures all LCL properties: a nearly optimal subspace-design code matches random linear codes with respect to every local property~\cite{BCDZ26a}.  The explicit subspace design codes obtained through their result, folded Reed--Solomon and univariate multiplicity codes, is over an alphabet that grows polynomially with the block length. Subsequently, Goyal, Guruswami, and Hsieh~\cite{GGH26} constructed explicit constant-alphabet subspace-design codes via the AEL framework.

Given the fundamental nature and broad applicability of LCL and subspace designs in the classical setting, it is natural to ask the following questions:
\begin{center}
\begin{itemize}
    \item[(i)] What are quantum generalizations of local properties and subspace designs? 
    \item[(ii)] Can we obtain explicit quantum CSS constructions for LCL properties and subspace designs via the local-to-global AEL paradigm?
\end{itemize}
\end{center}

In order to answer these questions, this paper first develops the quantum analogue of the LCL framework needed for CSS quantum codes. The guiding principle is:
\begin{center}
\emph{Impose local constraints on physical representatives, but measure distinctness in the logical quotient.}
\end{center}

The notion of quantum list decodable codes has also been studied in previous works, beginning with Leung and Smith~\cite{LS08}. These codes did not come with efficient list-decoding algorithms, and this was addressed in later works by Bergamaschi, Golowich, and Gunn~\cite{BGG24} with a non-LDPC construction, where a notion of quantum list recovery was introduced. Subsequently, Bergamaschi, Jeronimo, Mittal, Srivastava and Tulsiani~\cite{BJMST24} give efficient list docding algorithms for qLPDC up to the quantum Johnson bound. More recently, Gay, Jeronimo and Shukul \cite{GayJeronimoShukul27} give near-linear time list decoding at capacity for qLDPC codes. Quantum list-decodable codes were built by \cite{BGG24} as a tool to help build Approximate Quantum Error Correcting Codes, or AQECCS. These codes could correct errors all the way up to the quantum Singleton bound, with the caveat that the decoder would make an exponentially small error in the recovery process.

\paragraph{Stabilizer codes, normalizers, and syndromes.}
A \emph{stabilizer code} \(Q\) is the common \(+1\)-eigenspace of an
abelian Pauli subgroup \(\mathcal S(Q)\) containing no nontrivial
scalar operator.
Its \emph{normalizer} \(\mathcal N(Q)\) consists of the Pauli
operators commuting with every element of \(\mathcal S(Q)\), and
\(\mathcal N(Q)/\mathcal S(Q)\) describes the logical Pauli operators.

Two Pauli errors \(E,E'\) are \emph{stabilizer-equivalent} if
\(E^\dagger E'\in\mathcal S(Q)\), up to phase.
Stabilizer-equivalent errors have the same action on the code space up
to a global phase.

Fix stabilizer generators \(S_1,\ldots,S_r\).  The \emph{syndrome} of
a Pauli error \(E\) is
\(\operatorname{syn}_Q(E):=(\sigma_1,\ldots,\sigma_r)\in\F_p^r\),
where \(S_jE=\omega^{\sigma_j}ES_j\).
Syndromes are additive, and \(\operatorname{syn}_Q(E)=0\) exactly
when \(E\in\mathcal N(Q)\).  Consequently, two errors \(E,E'\) have
the same syndrome exactly when \(E^\dagger E'\in\mathcal N(Q)\), up
to phase.

\paragraph{CSS codes.}
Let \(C_1,C_2\subseteq(\F_q^s)^n\cong\F_q^{sn}\) be
\(\F_q\)-linear codes of dimensions \(k_1,k_2\), with duals taken
with respect to the standard \(\F_q\)-bilinear dot product on
\(\F_q^{sn}\), and suppose
\(C_2^\perp\subseteq C_1\) (equivalently
\(C_1^\perp\subseteq C_2\)).
For a CSS code \(Q=(C_1,C_2)\), ignoring phases, the stabilizer and normalizer
label spaces are
\(\mathcal S_Q=C_2^\perp\oplus C_1^\perp\) and
\(\mathcal N_Q=C_1\oplus C_2\), where the two components are the
\(X\)- and \(Z\)-Pauli labels.  Hence
\[
  \mathcal N_Q/\mathcal S_Q
  \cong
  (C_1/C_2^\perp)\oplus(C_2/C_1^\perp).
\]
The code encodes \(k_1+k_2-sn\) \(q\)-dimensional logical qudits and
has quantum rate \(R:=(k_1+k_2-sn)/(sn)\).
When \(k_1=k_2\), we call the CSS code \emph{balanced}; then
\(R_1:=k_1/(sn)=k_2/(sn)=(1+R)/2\).

Writing \(a=(a_1,\ldots,a_n)\) and \(b=(b_1,\ldots,b_n)\), with
\(a_i,b_i\in\F_q^s\), define the block weight of \(F=E_{a,b}\) by
\(\wt(F):=|\{i\in[n]:(a_i,b_i)\ne(0,0)\}|\).
Quantum list recovery allows a short list of possible Pauli labels at
each register.  For \(i\in[n]\), let
\(\mathcal E_i\subseteq\F_q^s\oplus\F_q^s\) satisfy
\(|\mathcal E_i|\le\ell\).
For a paired label
\(c=(a,b)\in(\F_q^s\oplus\F_q^s)^n\), define
\(\dis(c,\mathcal E):=
|\{i\in[n]:(a_i,b_i)\notin\mathcal E_i\}|\).
For \(F=E_{a,b}\), we also write
\(\dis(F,\mathcal E):=\dis((a,b),\mathcal E)\).

\paragraph{List Decoding, List Recovery.} Before stating the quantum versions of list decoding and list recovery, we state the classical versions. Let $\Hamm(x,y)$ denote the Hamming weight between two strings $x, y \in \Sigma^n$.
\begin{definition}[Classical List Decoding]
    A code $C \in \Sigma^n$ is $(\rho, L)$-list decodable if for every $z \in \Sigma^n$, the number of codewords in $c \in C$ satisfying $\Hamm(c, z) \le \rho n$ is at most $L$.
\end{definition}

\begin{definition}[Classical List recovery]
    A code $C \in \Sigma^n$ is $(\rho, \ell, L)$-list recoverable if for every input lists $S_i \subset \Sigma$, $|S_i| \le \ell, \forall i \in [n]$, the number of codewords in $c \in C$ satisfying 
    \[
        |\{ i \in [n] | c_i \not\in S_i\}| \le \rho n
    \]
    is at most $L$.
\end{definition}

The quantum versions of list decoding and list recovery are stated in terms of syndromes, rather than codewords.
\begin{definition}[Quantum list decoding;
{\cite[Definition~3.2]{BGG24}}]
\label{def:qld-intro}
A stabilizer code \(Q\) of block length \(n\) is
\((\rho,L)\)-\emph{quantum list decodable} (QLD) if, for every
possible syndrome \(\sigma\), there are at most \(L\) pairwise
stabilizer-distinct Pauli errors of weight at most \(\rho n\)
consistent with \(\sigma\).
\end{definition}

\begin{definition}[Quantum list recovery;
cf.~{\cite[Definition~B.2 and Claim~B.2]{BGG24}}]
\label{def:qlr-intro}
A stabilizer code \(Q\) of block length \(n\) is
\((\rho,\ell,L)\)-\emph{quantum list recoverable} (QLR) if, for every
syndrome \(\sigma\) and every collection of local Pauli lists
\(\mathcal E_1,\ldots,\mathcal E_n\) with
\(|\mathcal E_i|\le\ell\), there are at most \(L\) pairwise
stabilizer-distinct Pauli errors \(F\) with syndrome \(\sigma\)
satisfying \(\dis(F,\mathcal E)\le\rho n\).
\end{definition}

\paragraph{Reduction to the CSS quotient.}
The definitions above are stated in terms of Pauli errors.  In the CSS
setting, the relevant information is carried entirely by their labels.

Two Pauli errors \(E_{a,b}\) and \(E_{a',b'}\) have the same syndrome
if and only if
\((a-a',b-b')\in C_1\oplus C_2\), while they are
stabilizer-equivalent if and only if
\((a-a',b-b')\in C_2^\perp\oplus C_1^\perp\).
Thus, after translating a common-syndrome family by any one reference
error, its logically distinct candidates correspond to distinct cosets
in
\[
  \mathcal N_Q/\mathcal S_Q
  =
  (C_1\oplus C_2)/(C_2^\perp\oplus C_1^\perp).
\]

More explicitly, let
\(F^{(j)}=E_{a^{(j)},b^{(j)}}\), \(j\in[L+1]\), be pairwise
stabilizer-distinct errors with a common syndrome, and take
\(F^{(1)}\) as a reference.  Define
\(c_j:=(a^{(j)}-a^{(1)},b^{(j)}-b^{(1)})\).
Then \(c_j\in\mathcal N_Q\), and the cosets
\(c_j+\mathcal S_Q\) are pairwise distinct.

The local constraints are preserved by the same translation.
Given local Pauli lists \(\mathcal E_i\), set
\(T_i:=\mathcal E_i-(a_i^{(1)},b_i^{(1)})\).
Then \(|T_i|=|\mathcal E_i|\) and
\(\dis(c_j,T)=\dis(F^{(j)},\mathcal E)\) for every \(j\).
Hence both individual and aggregate disagreement are preserved under
the translation.
Conversely, representatives of distinct cosets in
\(\mathcal N_Q/\mathcal S_Q\) yield pairwise stabilizer-distinct
syndrome-zero Pauli errors.

Accordingly, QLR can be viewed as a classical list-recovery problem
over distinct cosets of the logical quotient.  QLD is the special case
\(\ell=1\), in which disagreement is simply block-Hamming distance
from the prescribed center.

\begin{definition}[Quantum list decoding, list recovery for CSS codes]
    For a CSS code $Q=(C_1,C_2)$, $Q$ is $(\rho, \ell, L)$-quantum list recoverable if for every collection of local Pauli lists
\(\mathcal E_1,\ldots,\mathcal E_n\) with
\(|\mathcal E_i|\le\ell\), if $c_1,\ldots,c_j \in N_Q$ have the same syndrome, satisfy \(\dis(c_i,\mathcal E)\le\rho n\) $\forall i \in [j]$, and $c_i + S_Q$ are all pairwise distinct cosets, then $j \le L$.

A CSS code $Q=(C_1,C_2)$ is $(\rho, L)$-quantum list decodable if it is $(\rho, 1, L)$-list recoverable.
\end{definition}

\subsection*{Quantum LCL}
We give a short introduction to the quantum Local Coordinate Wise-Linear (qLCL) framework.
A one-sided CSS sector is a nested pair
\[
      S\subseteq C\subseteq \F_q^n,
\]
where $S$ is the stabilizer space and $C/S$ is the logical space.  A witness matrix $A$ whose columns lie in $C$ has an ordinary row space $U$, but some nonzero linear combinations of its columns may land in $S$ and hence disappear in the quotient.  Thus a witness carries a flag $(U, T)$
\[
      T\le U,
\]
where $T$ is the logical row space and $U/T$ records stabilizer-valued directions.  If $\dim S=m_0$ and $\dim C=m_1$, the expected number of witnesses of a fixed local profile is controlled, up to constants depending only on the witness width, by
\begin{equation}\label{intro:eq:quotient-potential}
   \Phi_{m_0,m_1}(\cV,U,T)
   =\sum_{i=1}^n\dim(V_i\cap U)
     -(n-m_0)(\dim U-\dim T)-(n-m_1)\dim T .
\end{equation}
The two codimensions have different meanings.  Stabilizer-valued directions pay the codimension of $S$; logical directions pay the codimension of $C$.  When $S=0$ and $T=U$, this is exactly the LCL potential of Levi--Mosheiff--Shagrithaya~\cite{LMS25}.

As in the classical LCL theory, the full witness is not enough.  A lower-rank minor can be the first obstruction to appear.  In the quantum setting the correct minor of a flag $(U,T)$ along $W<U$ is $(W,T\cap W)$, so we use the quotient-minor gap
\[
   \Gamma_{m_0,m_1}(\cV,U,T)
   =\min_{W<U}\bigl(
      \Phi_{m_0,m_1}(\cV,U,T)-
      \Phi_{m_0,m_1}(\cV,W,T\cap W)
   \bigr).
\]
The intersection $T\cap W$ is not a convention: it is the logical row space that remains after restricting the witness to the minor.

\begin{theorem}[Informal qLCL threshold]\label{thm:informal-qlcl}
For fixed witness width and a fixed flagged local profile, a random nested pair $S\subseteq C\subseteq\F_q^n$ has the following threshold behavior.  If the quotient-minor gap is negative by $\epsilon$, then the flagged witness is absent with probability at least $1-q^{-\epsilon+O(1)}$.  If the quotient-minor gap is positive by $\epsilon$, then the flagged witness is present with probability at least $1-q^{-\epsilon+O(1)}$.  After a union bound over profiles and flags, the same criterion gives threshold theorems for LCL properties in both the parity-check and uniform nested-subspace ensembles.
\end{theorem}

The formal statements are Theorems~\ref{qlcl:thm:fixed-type}, \ref{qlcl:thm:qlcl-threshold}, and~\ref{qlcl:thm:uniform-threshold}. Section~\ref{sec:folded-prelims} extends the framework to folded CSS codes.  Logical distance is injective-logical ($T=U$), while a pure stabilizer event has $T=0$.  Thus the flag is not a technical entity; it is the mechanism that lets one distinguish nonzero logical witnesses from stabilizer degeneracies using a single potential.

The above statements deal with single sector witnesses only. However, for most applications of CSS codes, we need to observe the joint sector behavior. Subsection~\ref{sub:joint-sec} is concerned with extending the above statments to the joint sector setting.

\subsection*{Relative CSS subspace designs}

We give our definition for quantum subspace designs.

For a folded vector alphabet $W \cong \F_q^s$, we have our folded code $S\subseteq N\subseteq W^n$. A folded coordinate of a witness is now a linear map $B \rightarrow W$.
\begin{definition}[Relative subspace design]\label{rsd:def:relative}
Let \(S\subseteq N\subseteq W^n\), let \(r\in\N\), and let \(\alpha\in[0,1]\).  The pair \((N,S)\) is an \emph{\(\alpha\)-relative subspace design up to dimension \(r\)} if for every subspace \(A\le N\) satisfying
\[
  1\le \dim A\le r,
  \qquad A\cap S=\{0\},
\]
one has
\begin{equation}\label{rsd:eq:relative-design}
   \frac1n\sum_{i=1}^n \dim(A\cap K_i)
      \le \alpha\dim A .
\end{equation}
Equivalently,
\[
   \frac1n\sum_{i=1}^n \dim\pi_i(A)
      \ge (1-\alpha)\dim A .
\]
\end{definition}

The condition \(A\cap S=0\) says that \(A\) injects into the logical quotient \(N/S\).  Thus the test ignores pure stabilizer directions but measures zeros in the actual physical representatives.

\begin{definition}[CSS subspace-design code]\label{rsd:def:css-sd}
Let \(Q=Q(C_1,C_2)\) be the CSS code defined by \(C_2\subseteq C_1\subseteq W^n\).  We say that \(Q\) is an \emph{\((\alpha_X,\alpha_Z)\)-CSS subspace design up to dimension \(r\)} if
\[
   (C_1,C_2)\quad\text{is an }\alpha_X\text{-relative subspace design up to dimension }r,
\]
and
\[
   (C_2^\perp,C_1^\perp)\quad\text{is an }\alpha_Z\text{-relative subspace design up to dimension }r.
\]
\end{definition}
Thus, a folded CSS code is a CSS subspace design only if it satisfies the relative subspace design property for both sectors. We also define an LCL profile, known as kernel profile, and show that this is equivalent to relative subspace design codes, in the sense that the containment avoidance of the kernel profiles implies that the code is a relative subspace design code.

\subsection*{Construction: derandomizing LCL via AEL}

The construction follows the Alon--Edmonds--Luby local-to-global paradigm~\cite{AEL95}.  Classical AEL-based constructions use a constant-size inner code, an expander, and an outer code to lift local distance or list-recovery guarantees to long explicit codes~\cite{GI02,KMRS17,KRZSW23}.  Jeronimo, Mittal, Srivastava, and Tulsiani showed that AEL yields explicit constant-alphabet codes approaching the generalized Singleton bound for list decoding~\cite{JMST25}.  Most directly for the present paper, the LCL derandomization framework of Jeronimo--Shagrithaya~\cite{JS26} converts LCL guarantees of random linear codes into explicit constructions, and the subsequent work of Goyal, Guruswami, and Hsieh~\cite{GGH26} uses the paradigm to build subspace design codes with constant alphabet size.

On the quantum side, AEL transforms for CSS codes were developed in~\cite{BGG24,BJMST24}; our construction follows this CSS-AEL template, and the new point is the preservation of  qLCL guarantees.  Our CSS version first derandomizes all qLCL properties.  The relative subspace-design construction is then obtained by specializing the profile class to kernel profiles.  The outer CSS code constrains the sequence of inner logical symbols, and expansion converts local LCL robustness plus outer logical distance into a global absence theorem for bad LCL witnesses. Because we use qLDPC codes for the outer code, the final explicit construction also inherits this property.

\begin{theorem}[Informal robust LCL inner gadgets]\label{thm:informal-inner}
For every fixed witness dimension $r$, field $\F_q$, bounded-entropy folded profile class, and target inner CSS rate, there are constant-size folded CSS pairs whose $X$- and $Z$-normalizer sectors are AEL-robust folded qLCL gadgets at the rate-plus-slack thresholds supplied by the random qLCL calculation.  Since the size is constant, the gadgets can be found by exhaustive search.  Specializing to kernel profiles gives AEL-robust relative subspace-design gadgets.
\end{theorem}

The formal LCL existence theorem is Theorem~\ref{rsd:thm:random-robust-normalizer}; its kernel-profile corollary is Corollary~\ref{rsd:cor:random-robust-relative-normalizer}, and the two-sided CSS gadget statement is Theorem~\ref{rsd:thm:random-inner-css}.  The robust condition is stronger than the final design condition: it allows a local map from a global bad witness to have stabilizer kernel, exactly as happens when one inspects a single left vertex of the AEL construction.

\begin{theorem}[Informal main construction]\label{thm:main-informal}
Fix a target quantum rate $R\in(0,1)$, a witness dimension $r$, an accuracy $\eps>0$, and any pair of bounded-entropy folded profile classes for the two sectors.  There is an explicit folded CSS code family with folded block length $n$, constant fold size depending only on $(R,r,\eps)$ and the class entropy exponent, and quantum rate at least $R-O(\eps)$, containing no injective-logical LCL witness for the given classes up to dimension $r$ at thresholds
\[
   \alpha_X,\alpha_Z\le \frac{1+R}{2}+O(\eps).
\]
Specializing the classes to kernel profiles, the family is an $(\alpha_X,\alpha_Z)$-CSS subspace design up to dimension $r$.

\end{theorem}

The precise parameter statements are present in Theorem~\ref{rsd:thm:explicit}. The subspace-design consequence is recorded in Remark~\ref{rsd:rem:explicit-design}.
As a pure distance statement, the $r=1$ kernel-profile case of Theorem~\ref{thm:main-informal} is not new: folded CSS codes with Singleton-scale relative distance are already known via quantum AEL distance amplification~\cite{BGG24,BJMST24}. The new content is the qLCL guarantee for every bounded-entropy profile class, in particular the subspace-design property for every dimension up to $r$, of which the folded block-distance bound is the one-dimensional case.

This theorem allows one to obtain explicit constructions of quantum CSS codes by showing probabilistic existence of CSS codes having the desired qLCL property. Additionally, these constructions are qLDPC codes.

We also obtain explicit quantum list decodable and list recoverable codes qLDPC codes that have optimal list sizes:
\begin{theorem}[Explicit Quantum List Decodable Codes, Informal]\label{thm:inf-list-dec}
    There exists an infinite family of qLDPC codes with rate $R_Q$ that are $(\rho, L)$-quantum list decodable, with $\rho$ satisfying:
    \[
        \rho < \frac{L}{L+1}\cdot \frac{(1-R_Q-\eps)}{2}.
    \]
    The alphabet size is a constant, and depends only on $R, \eps, L$.
\end{theorem}
We compare this result to previous constructions. \cite{BGG24} obtained explicit list decodable codes of comparable radius, but their list size was polynomial in the block length, while ours is a constant.
The explicit codes in \cite{BJMST24} were list decodable for radius upto the Johnson bound $J((1-R)/2)$, which is strictly smaller than the bound we have on the radius. We mention that all three constructions yield qLDPC codes over constant alphabet sizes.

\begin{theorem}[Explicit Quantum List Recoverable Codes, Informal]\label{thm:inf-list-rec}
    There exists an infinite family of qLDPC codes that are $(\rho, \ell, L)$-quantum list recoverable such that
    \[
        \rho \le (1-R_Q-\eps)/2 \quad L \le (\ell/\eps)^{O(R/\eps)}.
    \]
    The alphabet size is a constant that only depends on $R, \eps, \ell, L$.
\end{theorem}

The formal versions of these theorems are Theorems~\ref{thm:exp-list-dec}, \ref{thm:exp-list-rec} respectively.

\subsection*{Main-result map}

The introduction has stated informal versions of all main theorems.  For reference, the following table records where each statement is proved and what role it plays in the merged argument.
\begin{center}
\begin{tabular}{@{}lll@{}}
\toprule
Informal statement & Formal result & Role in the paper \\
\midrule
Theorem~\ref{thm:informal-qlcl} & Theorems~\ref{qlcl:thm:fixed-type}, \ref{qlcl:thm:qlcl-threshold}, \ref{qlcl:thm:uniform-threshold} & qLCL thresholds \\
Theorem~\ref{thm:informal-inner} & Theorem~\ref{rsd:thm:random-robust-normalizer} and Theorem~\ref{rsd:thm:random-inner-css} & Constant-size LCL gadgets \\
Theorem~\ref{thm:main-informal} & Theorem~\ref{rsd:thm:explicit}, Remark~\ref{rsd:rem:explicit-design} & Explicit LCL CSS families, \\ & & Explicit quantum \\ & & subspace design codes \\
Theorem~\ref{thm:inf-list-dec}, Theorem~\ref{thm:inf-list-rec} & Theorems~\ref{thm:exp-list-dec}, \ref{thm:exp-list-rec} & Explicit list-recoverable, \\& & list-decodable codes \\
\bottomrule
\end{tabular}
\end{center}

\subsection*{Contributions and organization}

The paper's technical contributions are as follows.
\begin{enumerate}[leftmargin=2em,label=\textup{(\arabic*)}]
\item We formulate quotient LCL properties for nested pairs $S\subseteq C$, including logical flags, quotient-minor gaps, and a sharp fixed-type threshold theorem for random nested pairs.
\item We introduce relative CSS subspace designs, prove their deterministic profile characterization, and isolate minimal bad profiles suitable for expander amplification.
\item We derandomize folded qLCL profile classes via AEL-robust constant-size inner CSS gadgets and prove the corresponding  CSS-AEL lifting theorem.
\item We obtain explicit quantum list decodable and list recoverable codes with optimal list sizes, and also obtain explicit quantum subspace designable codes.
\end{enumerate}

\section{Preliminaries}\label{sec:qlcl-prelims}

We first detail the finite-dimensional objects used in the qLCL threshold argument. It follows the LCL viewpoint of Levi--Mosheiff--Shagrithaya, but adds exactly one new piece of data: a logical row space measuring the image of a witness in a quotient.  Folded vector alphabets are introduced only after the threshold theorem has been proved, so that the probabilistic argument remains visibly the same as in the scalar setting.

Throughout, $q$ is a prime power and all vector spaces are over $\F_q$.  For a vector space $X$, $\cL(X)$ denotes its lattice of $\F_q$-linear subspaces.  We write $[n]=\{1,\ldots,n\}$.  For a matrix $A\in\F_q^{n\times b}$, $A_{i,*}$ denotes its $i$th row and $A_{*,j}$ its $j$th column. The notation $U \le V$ is used to denote that $U$ is a subspace of $V$. We use the standard nondegenerate bilinear form on $\F_q^b$ and write $W^\perp$ for the orthogonal complement of a subspace $W\leq \F_q^b$. For two linear spaces $B, W$, we use $\Hom(B,W)$ to denote the set of all linear maps from $B$ to $W$.

\subsection{Local profiles and row spaces}

\begin{definition}[Local profile]
Fix $b\in\N$.  A \emph{$b$-local profile} of length $n$ is a sequence
\[
       \cV=(V_1,\ldots,V_n),\qquad V_i\leq \F_q^b .
\]
For $U\leq \F_q^b$, define its profile mass
\begin{equation}\label{qlcl:eq:profile-mass}
       a_\cV(U) \defeq \sum_{i=1}^n \dim(V_i\cap U).
\end{equation}
Let
\[
       \cM_\cV(U) \defeq \{A\in\F_q^{n\times b}: A_{i,*}\in V_i\cap U\text{ for every }i\}.
\]
Thus $|\cM_\cV(U)|=q^{a_\cV(U)}$.  Let
\[
       \cM_\cV^{=}(U) \defeq \{A\in\cM_\cV(U):\row(A)=U\}.
\]
\end{definition}

Equivalently, if
\[
       \lambda_A:\F_q^b\to\F_q^n,
       \qquad x\mapsto Ax,
\]
then $\ker\lambda_A=\row(A)^\perp$ and $\rank(A)=\dim\row(A)$.

\begin{lemma}[Exact row-span count]\label{qlcl:lem:exact-row-count}
Let $\cV$ be a $b$-local profile and let $U\leq \F_q^b$.  Then
\[
       |\cM_\cV^{=}(U)|\leq q^{a_\cV(U)}.
\]
Moreover, if
\[
       a_\cV(U)-a_\cV(W)\geq \epsilon
       \qquad\text{for every proper }W<U,
\]
then
\[
       |\cM_\cV^{=}(U)|
       \geq q^{a_\cV(U)}\left(1-q^{-\epsilon+b^2}\right).
\]
\end{lemma}

\begin{proof}
The upper bound is immediate.  A matrix in $\cM_\cV(U)$ whose row span is not $U$ lies in $\cM_\cV(W)$ for some proper subspace $W<U$.  The number of subspaces of $\F_q^b$ is at most $q^{b^2}$.  Hence
\[
   |\cM_\cV(U)\setminus \cM_\cV^{=}(U)|
   \leq \sum_{W<U}q^{a_\cV(W)}
   \leq q^{b^2}q^{a_\cV(U)-\epsilon}.
\]
Subtracting from $|\cM_\cV(U)|=q^{a_\cV(U)}$ proves the claim.
\end{proof}

\subsection{Nested pairs and logical row spaces}

\begin{definition}[Nested pair]
A \emph{nested pair} is an inclusion of linear codes
\[
       S\subseteq C\subseteq \F_q^n .
\]
The quotient $C/S$ is called the logical quotient.
\end{definition}

Let $A\in\F_q^{n\times b}$ have all columns in $C$.  The map $\lambda_A:\F_q^b\to C$ defined above induces a map
\[
       \overline\lambda_A:\F_q^b\to C/S.
\]
The kernel of $\overline\lambda_A$ is the space of linear combinations of the $b$ columns that are stabilizers.

\begin{definition}[Logical kernel and logical row space]\label{qlcl:def:logical-row-space}
For a nested pair $S\subseteq C$ and a matrix $A\in\F_q^{n\times b}$ whose columns lie in $C$, define
\[
       K_S(A) \defeq \lambda_A^{-1}(S) = \{x\in\F_q^b:Ax\in S\}.
\]
The \emph{logical row space} of $A$ modulo $S$ is
\[
       \Lambda_S(A)=K_S(A)^\perp\leq \F_q^b .
\]
\end{definition}

Since $\ker\lambda_A=\row(A)^\perp\subseteq K_S(A)$, one always has
\begin{equation}\label{qlcl:eq:lambda-contained-row}
       \Lambda_S(A)\leq \row(A).
\end{equation}
Also $\dim\Lambda_S(A)=\dim\im(\overline\lambda_A)$, so $\Lambda_S(A)$ records the number of independent logical directions represented by the columns of $A$.

\begin{definition}[Logical distinctness space]
Let
\[
       \LDist_b
       =\{T\leq\F_q^b:\forall j\neq k\text{ there exists }t\in T\text{ with }t_j\neq t_k\}.
\]
Equivalently, $T\in\LDist_b$ iff $e_j-e_k\notin T^\perp$ for all $j\neq k$.
\end{definition}

\begin{lemma}[Distinctness modulo stabilizers]\label{qlcl:lem:logical-distinctness}
Let $A\in\F_q^{n\times b}$ have columns in $C$.  The $b$ columns of $A$ are pairwise distinct as elements of $C/S$ iff
\[
       \Lambda_S(A)\in\LDist_b .
\]
\end{lemma}

\begin{proof}
Columns $j$ and $k$ represent the same coset modulo $S$ iff $A(e_j-e_k)\in S$, i.e. iff $e_j-e_k\in K_S(A)=\Lambda_S(A)^\perp$.  This is exactly the negation of the defining condition for $\LDist_b$.
\end{proof}

\section{Quotient LCL properties}\label{sec:quotient-lcl-properties}

We now define the quotient analogue of a local coordinate-wise linear property.  A profile still supplies coordinate-wise row constraints, and a realization is still a bounded-width matrix whose columns lie in the ambient code.  The only new data is the logical flag: the same physical matrix is measured by its ordinary row space and by its image in the quotient by stabilizers.  This is the minimal change that makes the classical LCL formalism compatible with CSS codes.

\subsection{Flags and realizations}

\begin{definition}[Flag and feasibility]
For fixed locality $b$, define
\[
       \Flag_b=\{(U,T):T\leq U\leq\F_q^b\}.
\]
For dimensions $m_0\leq m_1$, a flag $(U,T)$ is called \emph{$(m_0,m_1)$-feasible} if
\[
       \dim(U/T)\leq m_0,
       \qquad
       \dim T\leq m_1-m_0 .
\]
It is called \emph{parity-feasible} for $(m_0,m_1)$ if the second inequality holds; the first is not required.
\end{definition}

The first feasibility inequality says that the stabilizer-valued part of a witness must fit inside $S$; the second says that its nonzero logical image must fit inside $C/S$.  In the unconditioned parity-check ensemble of Section~\ref{qlcl:sec:random-nested}, the dimension of $S$ is random, so only parity-feasibility is needed for the exact second-moment argument.  In the uniform ensemble both inequalities are unavoidable.

\begin{definition}[Flagged realization]\label{qlcl:def:flagged-realization}
Let $S\subseteq C$ be a nested pair, let $\cV$ be a $b$-local profile, and let $(U,T)\in\Flag_b$.  We say that $S\subseteq C$ \emph{realizes} $(\cV,U,T)$ if there exists a matrix $A\in\F_q^{n\times b}$ such that
\begin{enumerate}[label=(\roman*)]
    \item $A_{i,*}\in V_i$ for every $i\in[n]$;
    \item $\row(A)=U$;
    \item every column of $A$ lies in $C$;
    \item $\Lambda_S(A)=T$.
\end{enumerate}
Such an $A$ is called a realization of $(\cV,U,T)$.
\end{definition}

\begin{lemma}[Automatic feasibility in a fixed nested pair]\label{qlcl:lem:auto-feasibility}
Let $S\subseteq C\subseteq\F_q^n$ have $\dim S=m_0$ and $\dim C=m_1$.  If $S\subseteq C$ realizes $(\cV,U,T)$, then $(U,T)$ is $(m_0,m_1)$-feasible.
\end{lemma}

\begin{proof}
Let $A$ be a realization and write $B=\im\lambda_A\leq C$.  Then $\dim B=\dim U$.  The image of $K_S(A)$ under $\lambda_A$ is exactly $B\cap S$, and its kernel is $\ker\lambda_A=U^\perp$.  Since $K_S(A)=T^\perp$,
\[
       \dim(B\cap S)
       =\dim T^\perp-\dim U^\perp
       =\dim U-\dim T.
\]
Thus $\dim(U/T)\leq\dim S=m_0$.  The induced image of $B$ in $C/S$ has dimension $\dim T$, so $\dim T\leq\dim(C/S)=m_1-m_0$.
\end{proof}

\begin{definition}[Quotient LCL property]\label{qlcl:def:qlcl-property}
A \emph{$b$-quotient LCL property} of nested pairs of length $n$ is specified by a finite family $\cF$ of $b$-local profiles and a set $\cD\subseteq\cL(\F_q^b)$ of allowed logical row spaces.  A nested pair $S\subseteq C$ satisfies $\cP(\cF,\cD)$ if there exist
\[
       \cV\in\cF,
       \qquad
       (U,T)\in\Flag_b,
       \qquad
       U\neq 0,
       \qquad
       T\in\cD,
\]
such that $S\subseteq C$ realizes $(\cV,U,T)$.  Thus the all-zero matrix is not a witness.
\end{definition}

\subsection{Quotient minors}

The next lemma is the reason the threshold parameter uses $T\cap W$.

\begin{lemma}[Quotienting a flagged witness]\label{qlcl:lem:quotient-witness}
Let $S\subseteq C$ realize $(\cV,U,T)$ via $A \in \F_q^{n \times b}$.  Let $W<U$ be a proper subspace, and let
\[
       \rho:U\to \overline U
\]
be a surjective linear map with kernel $W$.  Choose a matrix $B\in\F_q^{b\times d}$, where $d=\dim\overline U$, such that right multiplication by $B$ restricts to $\rho$ on $U$.  Put $D=AB$ and define the quotient profile $\overline \cV = (\overline V_1,\ldots, \overline V_n)$:
\[
       \overline V_i=\rho(V_i\cap U)\leq \overline U\leq\F_q^d.
\]
Then $D$ is a realization, for the same nested pair, of the quotient flag
\[
       (\overline\cV,\overline U,\rho(T)).
\]
Moreover
\[
       \dim\rho(T)=\dim T-\dim(T\cap W).
\]
\end{lemma}

\begin{proof}
The rows of $D$ are obtained by applying $\rho$ to the rows of $A$, hence $D_{i,*}\in\overline V_i$ and $\row(D)=\rho(\row(A))=\overline U$.  The columns of $D$ are linear combinations of the columns of $A$, so they lie in $C$.

It remains to identify the logical row space.  Let $K=K_S(A)=T^\perp$.  For $y\in\F_q^d$,
\[
       Dy\in S
       \quad\Longleftrightarrow\quad
       A(By)\in S
       \quad\Longleftrightarrow\quad
       By\in K.
\]
Thus $K_S(D)=B^{-1}K$.  Taking orthogonal complements gives
\[
       \Lambda_S(D)=(B^{-1}K)^\perp=T B=\rho(T),
\]
where the last equality uses that right multiplication by $B$ is $\rho$ on $U$ and $T\leq U$.  The kernel of $\rho|_T$ is $T\cap W$, proving the dimension formula.
\end{proof}

\section{Random nested linear pairs}\label{qlcl:sec:random-nested}

The random model used for the proof is a parity-check model.  It is slightly more flexible than the uniform nested-subspace ensemble because the parity checks are independent before conditioning on full rank; this independence is what makes the fixed-matrix and second-moment estimates exact.

Fix integers
\[
       0\leq m_0\leq m_1\leq n,
       \qquad
       a=n-m_1,
       \qquad
       c=m_1-m_0 .
\]
Thus $a$ is the codimension of $C$ and $c$ is the dimension of the logical quotient $C/S$.

\subsection{The parity-check ensemble}

\begin{definition}[Parity-check nested ensemble]\label{qlcl:def:pcss}
The ensemble $\PCSS(n,q;m_0,m_1)$ is sampled by choosing independent uniformly random matrices
\[
       H_C\in\F_q^{a\times n},
       \qquad
       H_S\in\F_q^{c\times n},
\]
and setting
\[
       C=\ker H_C,
       \qquad
       S=\ker H_C\cap\ker H_S.
\]
Then $S\subseteq C$.
\end{definition}

If $H_C$ and the stacked matrix $\begin{bmatrix}H_C\\H_S\end{bmatrix}$ have full row rank, then $\dim C=m_1$ and $\dim S=m_0$.  The unconditioned parity-check ensemble has the advantage that restrictions of the checks to fixed independent subspaces are exactly independent.  Conditioning on full rank recovers the uniform nested-subspace model in Section~\ref{qlcl:sec:uniform}.

\subsection{Potential and quotient-minor gap}

\begin{definition}[Flag cost and potential]\label{qlcl:def:potential}
For a flag $(U,T)$, write $u=\dim U$ and $t=\dim T$.  Define its random-nested-pair cost by
\begin{equation}\label{qlcl:eq:flag-cost}
       \kappa_{m_0,m_1}(U,T)
       \defeq a u+c(u-t)
       =(n-m_0)(u-t)+(n-m_1)t .
\end{equation}
For a profile $\cV$, define the qLCL potential
\begin{equation}\label{qlcl:eq:potential}
       \Phi_{m_0,m_1}(\cV,U,T)
       \defeq a_\cV(U)-\kappa_{m_0,m_1}(U,T).
\end{equation}
\end{definition}

The first expression in \eqref{qlcl:eq:flag-cost} comes from parity checks: every ordinary direction must satisfy the $a$ checks defining $C$, and each of the $u-t$ stabilizer directions must also satisfy the $c$ additional checks defining $S$ inside $C$.  The second expression says that stabilizer directions pay $\codim S=n-m_0$, whereas logical directions pay $\codim C=n-m_1$.

\begin{definition}[Quotient-minor gap]\label{qlcl:def:gamma}
For $U\neq 0$ and $T\leq U$, define
\begin{equation}\label{qlcl:eq:gamma}
       \Gamma_{m_0,m_1}(\cV,U,T)
        \defeq \min_{W<U}
       \left(
          \Phi_{m_0,m_1}(\cV,U,T)
          -\Phi_{m_0,m_1}(\cV,W,T\cap W)
       \right).
\end{equation}
The minimum is over proper subspaces $W$ of $U$.
\end{definition}

Because $0<U$ is allowed as a quotient minor, $\Gamma(\cV,U,T)\leq\Phi(\cV,U,T)$.  Hence a positive quotient-minor gap implies a positive raw potential.  The minor operation is the quotient analogue of passing from a witness to a lower-rank dependency among its columns: the ordinary row space drops from $U$ to $W$, and exactly the logical directions lying in $W$ survive, namely $T\cap W$.

\begin{lemma}[Mass separation from positive gap]\label{qlcl:lem:gamma-mass}
If $\Gamma_{m_0,m_1}(\cV,U,T)\geq \epsilon$, then for every proper $W<U$,
\[
       a_\cV(U)-a_\cV(W)\geq \epsilon .
\]
Consequently
\[
       |\cM_\cV^{=}(U)|
       \geq q^{a_\cV(U)}(1-q^{-\epsilon+b^2}).
\]
\end{lemma}

\begin{proof}
Let $T_W=T\cap W$.  The cost difference
\[
       \kappa(U,T)-\kappa(W,T_W)
       =(n-m_0)\bigl((u-t)-(\dim W-\dim T_W)\bigr)
        +(n-m_1)(t-\dim T_W)
\]
is nonnegative because $T_W=T\cap W$.  Therefore
\[
       a_\cV(U)-a_\cV(W)
       =\bigl(\Phi(U,T)-\Phi(W,T_W)\bigr)
        +\bigl(\kappa(U,T)-\kappa(W,T_W)\bigr)
       \geq \epsilon .
\]
Lemma~\ref{qlcl:lem:exact-row-count} gives the count.
\end{proof}

\subsection{A fixed matrix}

\begin{lemma}[Probability of a fixed flagged matrix]\label{qlcl:lem:fixed-matrix-prob}
Let $A\in\F_q^{n\times b}$ have row space $U$ of dimension $u$, and let $T\leq U$ have dimension $t$.  In the ensemble $\PCSS(n,q;m_0,m_1)$,
\[
 \Prb[ A_{*,1},\ldots,A_{*,b}\in C\text{ and }\Lambda_S(A)=T]
   = q^{-\kappa_{m_0,m_1}(U,T)}\,\iota(c,t),
\]
where
\[
       \iota(c,t)=
       \begin{cases}
       \displaystyle\prod_{j=0}^{t-1}(1-q^{j-c}), & t\leq c,\\[1ex]
       0, & t>c,
       \end{cases}
\]
with the convention $\iota(c,0)=1$.
\end{lemma}

\begin{proof}
The event that all columns of $A$ lie in $C$ is $H_C A=0$.  Since $\rank(A)=u$, this gives $a u$ independent linear equations on $H_C$, and has probability $q^{-au}$.

Condition on $H_C A=0$.  The matrix $H_SA$ is a uniformly random linear map
\[
       \F_q^b/\ker\lambda_A\cong \F_q^b/U^\perp\longrightarrow \F_q^c .
\]
The condition $\Lambda_S(A)=T$ is equivalent to
\[
       \ker(H_SA)=T^\perp.
\]
Modulo $U^\perp$, this says that a random map from the $u$-dimensional domain vanishes exactly on the fixed $(u-t)$-dimensional subspace $T^\perp/U^\perp$.  Vanishing on that subspace has probability $q^{-c(u-t)}$.  After quotienting by it, the induced map from a $t$-dimensional space to $\F_q^c$ must be injective, which has probability $\prod_{j=0}^{t-1}(1-q^{j-c})$ if $t\leq c$ and probability $0$ otherwise.  Multiplying by $q^{-au}$ gives the formula.
\end{proof}

\begin{corollary}[Expectation for one flag]\label{qlcl:cor:one-flag-expectation}
Let
\[
       X_{\cV,U,T}=\#\{A\in\cM_\cV^{=}(U): A_{*,j}\in C\text{ for all }j,\ \Lambda_S(A)=T\}.
\]
Then
\[
       \Exp [X_{\cV,U,T}]
       \leq q^{\Phi_{m_0,m_1}(\cV,U,T)}.
\]
If $\Gamma_{m_0,m_1}(\cV,U,T)\geq\epsilon$ and $t\leq c$, then
\[
       \Exp [X_{\cV,U,T}]
       \geq q^{\Phi_{m_0,m_1}(\cV,U,T)}
          (1-q^{-\epsilon+b^2})\,\iota(c,t).
\]
In particular, if $\epsilon\geq 2b^2+1$ and $t\leq c$, then
\[
       \Exp [X_{\cV,U,T}]\geq q^{\Phi_{m_0,m_1}(\cV,U,T)-2b^2}.
\]
All expectations are over the random ensemble $\PCSS(n,q;m_0,m_1)$.
\end{corollary}

\begin{proof}
The upper bound follows from $|\cM_\cV^{=}(U)|\le q^{a_\cV(U)}$ and Lemma~\ref{qlcl:lem:fixed-matrix-prob}.  The lower bound follows from Lemma~\ref{qlcl:lem:gamma-mass} and Lemma~\ref{qlcl:lem:fixed-matrix-prob}.  For the final displayed bound, note that $\iota(c,t)\geq q^{-t}\geq q^{-b}$ when $t\leq c$, and $1-q^{-\epsilon+b^2}\geq q^{-1}$ under the stated numerical assumption.  Weakening constants gives the stated $2b^2$ loss.
\end{proof}

\section{Fixed-type threshold theorem}\label{qlcl:sec:fixed-threshold}

We next prove the threshold for one fixed profile and one fixed flag.  The subcritical direction is a first-moment argument applied not necessarily to the original witness, but to the bad quotient minor that certifies the negative gap in potential.  The supercritical direction is a second-moment argument; where we need to keep track of the joint logical row space of two dependent witnesses.

\subsection{Subcritical direction}

\begin{lemma}[Potential of a quotient minor]\label{qlcl:lem:minor-potential}
Let $W<U$, set $T_W=T\cap W$, and let $\rho:U\to\overline U$ be a quotient map with kernel $W$.  Let $\overline\cV$ be the quotient profile $\overline V_i=\rho(V_i\cap U)$ and let $\overline T=\rho(T)$.  Then
\[
       \Phi_{m_0,m_1}(\overline\cV,\overline U,\overline T)
       =\Phi_{m_0,m_1}(\cV,U,T)-\Phi_{m_0,m_1}(\cV,W,T_W).
\]
\end{lemma}

\begin{proof}
For every coordinate $i$,
\[
       \dim\overline V_i
       =\dim(V_i\cap U)-\dim(V_i\cap W).
\]
Therefore $a_{\overline\cV}(\overline U)=a_\cV(U)-a_\cV(W)$.  Also
\[
       \dim\overline U=\dim U-\dim W,
       \qquad
       \dim\overline T=\dim T-\dim(T\cap W),
\]
so the flag cost subtracts in the same way.  Combining these identities proves the claim.
\end{proof}

\begin{proposition}[Subcritical fixed flag]\label{qlcl:prop:fixed-subcritical}
If $U\neq0$, $T\leq U$, and
\[
       \Gamma_{m_0,m_1}(\cV,U,T)\leq -\epsilon,
\]
then, for $(S,C)\sim\PCSS(n,q;m_0,m_1)$,
\[
       \Prb[(S,C)\text{ realizes }(\cV,U,T)]\leq q^{-\epsilon}.
\]
\end{proposition}

\begin{proof}
Choose $W<U$ such that
\[
       \Phi(\cV,U,T)-\Phi(\cV,W,T\cap W)\leq -\epsilon.
\]
By Lemma~\ref{qlcl:lem:quotient-witness}, any realization of $(\cV,U,T)$ gives a realization of the quotient profile and quotient flag associated with $W$.  By the first-moment upper bound in Corollary~\ref{qlcl:cor:one-flag-expectation} and Lemma~\ref{qlcl:lem:minor-potential}, the probability that this quotient realization exists is at most
\[
       q^{\Phi(\cV,U,T)-\Phi(\cV,W,T\cap W)}\leq q^{-\epsilon}.
\]
\end{proof}

\subsection{Joint flags and second moment}

For a fixed profile $\cV$ and flag $(U,T)$, let $X=X_{\cV,U,T}$ be the random variable from Corollary~\ref{qlcl:cor:one-flag-expectation}.  We bound the contribution to the second moment from dependent pairs. For subspaces $U, V \le \F_q^b$, let $U \oplus V$ be the subspace of $\F_q^{2b}$ defined as:
\[
    U \oplus V \defeq \{(u, v) \in \F_q^{2b} \mid u \in U, v \in V\}.
\]

For $\Omega\leq U\oplus U$, define
\[
       a_\cV^{(2)}(\Omega)
       =\sum_{i=1}^n\dim\bigl((V_i\oplus V_i)\cap\Omega\bigr).
\]
For a joint flag $\Theta\leq\Omega$, define
\[
       \Phi^{(2)}(\cV,\Omega,\Theta)
       =a_\cV^{(2)}(\Omega)-a\dim\Omega-c(\dim\Omega-\dim\Theta).
\]
This is the same potential as \eqref{qlcl:eq:potential}, but for the doubled witness width.

Let $\Omega\leq U\oplus U$ and suppose the projection onto the first copy $\pi_1:\Omega\to U$ is surjective.  Let $\ker \pi_1$ be isomorphic to some subspace $W \le U$ according to the linear isomorphism $\phi\colon \ker \pi_1 \rightarrow W$, where we identify $W$ as a subspace of the second copy of $U$. This is possible because $\pi_1$ is surjective, and therefore $\dim \ker \pi_1 \le \dim U$. Let $\Theta\leq\Omega$ satisfy $\pi_1(\Theta)=T$. Let 
\[
       T'=\phi(\ker(\pi_1|_\Theta))\leq W.
\]

\begin{lemma}[Pair submodularity]\label{qlcl:lem:pair-submodularity}
Let $U, \Omega, W, \Theta, T, T'$ be as defined above.
Then
\begin{equation}\label{qlcl:eq:pair-submod}
       \Phi^{(2)}(\cV,\Omega,\Theta)
       \leq
       \Phi_{m_0,m_1}(\cV,U,T)
       +\Phi_{m_0,m_1}(\cV,W,T').
\end{equation}
\end{lemma}

\begin{proof}
Denote $V' = V_i \oplus V_i$. By rank-nullity
\[
    \dim (\Omega \cap V') = \dim \pi_1(V') + \dim (V' \cap \ker \pi_1).
\]
Because $\pi_1$ is a projection, $\pi_1(V')$ is contained in $V_i \cap U$, hence $\dim \pi_1(V') \le \dim (V_i \cap U)$.
Moreover, because $\pi_1$ is a projection to the first copy, we have:
\[
    V' \cap \ker \pi_1 \cong 0 \oplus (V_i \cap W).
\]
Hence, 
\[
       \dim((V_i\oplus V_i)\cap\Omega)
       \leq \dim(V_i\cap U)+\dim(V_i\cap W).
\]

Summing over $i$ gives $a_\cV^{(2)}(\Omega)\leq a_\cV(U)+a_\cV(W)$.  Since $\pi_1$ is onto,
\[
       \dim\Omega=\dim U+\dim W,
       \qquad
       \dim\Theta=\dim T+\dim T'.
\]
Substitution into the doubled potential gives \eqref{qlcl:eq:pair-submod}.
\end{proof}

\begin{lemma}[Dependent-pair bound]\label{qlcl:lem:dependent-pairs}
Assume $U\neq0$, $T\leq U$, $t=\dim T\leq c$, and
\[
       \Gamma_{m_0,m_1}(\cV,U,T)\geq\epsilon .
\]
Let $Y$ be the contribution to $\Exp [X^2]$ from ordered pairs $(A,B)$ for which the joint row span of the concatenated matrix $[A\;B]$ is not $U\oplus U$.  Then
\[
       Y\leq q^{2\Phi_{m_0,m_1}(\cV,U,T)-\epsilon+8b^2}.
\]
\end{lemma}

\begin{proof}
Fix a joint row space $\Omega\leq U\oplus U$ whose two coordinate projections are $U$ and with $\Omega\neq U\oplus U$.  The number of ordered pairs $(A,B)$ satisfying the profile and having exact joint row span $\Omega$ is at most $q^{a_\cV^{(2)}(\Omega)}$.

If both $A$ and $B$ are counted by $X$, then for $J=[A\;B]$ the joint logical row space $\Theta=\Lambda_S(J)$ satisfies $\Theta\leq\Omega$ and its two coordinate projections are both $T$.  For a subspace $L\leq \F_q^b \oplus \F_q^b $ we have the elementary identity
\[
       \pi_1(L^\perp)=\bigl(L\cap (\F_q^b \oplus 0)\bigr)^\perp,
\]
where $(L\cap (\F_q^b \oplus 0)\bigr)$ is identified in the natural manner with a subspace of $\F_q^b $, upon which the orthogonal complement is taken inside $\F_q^b $.  Applying it with $L=K_S(J)=\Lambda_S(J)^\perp$ gives
\[
       \pi_1(\Theta)=\bigl(K_S(J)\cap(\F_q^b\oplus0)\bigr)^\perp=K_S(A)^\perp=T,
\]
and similarly $\pi_2(\Theta)=T$.

For fixed $\Omega$ and $\Theta$, Lemma~\ref{qlcl:lem:fixed-matrix-prob} applied to $J$ gives the upper bound
\[
       q^{-a\dim\Omega-c(\dim\Omega-\dim\Theta)}
\]
for the probability of the corresponding joint event; we drop the injectivity factor.  There are at most $q^{4b^2}$ choices of $\Omega$ and at most $q^{4b^2}$ choices of $\Theta$.

Since $\Omega\neq U\oplus U$ and the first projection is onto, $W=\phi(\ker(\pi_1|_\Omega))$ is a proper subspace of $U$.  Moreover, if $T'=\phi(\ker(\pi_1|_\Theta))$, then $T'\leq T\cap W$, because the second projection of $\Theta$ is $T$.  Since the potential is monotone nondecreasing in the logical row space for fixed ordinary row space,
\[
       \Phi(\cV,W,T')\leq \Phi(\cV,W,T\cap W).
\]
The gap assumption gives
\[
       \Phi(\cV,W,T')\leq \Phi(\cV,U,T)-\epsilon.
\]
Applying Lemma~\ref{qlcl:lem:pair-submodularity}, we get that every joint type contributes at most
\[
       q^{\Phi^{(2)}(\cV,\Omega,\Theta)}
       \leq q^{2\Phi(\cV,U,T)-\epsilon}.
\]
Multiplying by the number of joint types proves the claim.
\end{proof}

\begin{proposition}[Supercritical fixed flag]\label{qlcl:prop:fixed-supercritical}
Assume $U\neq0$, $T\leq U$, $t=\dim T\leq c$, and
\[
       \Gamma_{m_0,m_1}(\cV,U,T)\geq\epsilon.
\]
If $\epsilon\geq 2b^2+1$, then
\[
       \Prb[(S,C)\text{ realizes }(\cV,U,T)]
       \geq 1-q^{-\epsilon+12b^2}.
\]
\end{proposition}

\begin{proof}
Let $X=X_{\cV,U,T}$.  Because $0<U$ is one of the quotient minors, the gap assumption implies
\[
       \Phi(\cV,U,T)\geq \epsilon.
\]
Corollary~\ref{qlcl:cor:one-flag-expectation} gives
\[
       \Exp X\geq q^{\Phi(\cV,U,T)-2b^2}.
\]

For ordered pairs $(A,B)$ with joint row span $U\oplus U$, the column spaces generated by $A$ and $B$ are independent subspaces of $\F_q^n$.  The restrictions of the random parity checks to these two subspaces are independent, so the covariance of the two indicator variables is zero.  For all remaining ordered pairs, the covariance is at most the joint probability, and Lemma~\ref{qlcl:lem:dependent-pairs} bounds the sum of these joint probabilities.  Therefore
\[
       \operatorname{Var}(X)
       \leq q^{2\Phi(\cV,U,T)-\epsilon+8b^2}.
\]
Chebyshev's inequality yields
\[
       \Prb[X=0]
       \leq \frac{\operatorname{Var}(X)}{(\Exp X)^2}
       \leq q^{-\epsilon+12b^2}.
\]
\end{proof}

\subsection{Fixed-type theorem}

\begin{theorem}[Fixed qLCL type threshold]\label{qlcl:thm:fixed-type}
Let $(S,C)\sim\PCSS(n,q;m_0,m_1)$.  Fix a $b$-local profile $\cV$ and a flag $(U,T)$ with $U\neq0$.
\begin{enumerate}[label=(\alph*)]
\item If $\Gamma_{m_0,m_1}(\cV,U,T)\leq-\epsilon$, then
\[
       \Prb[(S,C)\text{ realizes }(\cV,U,T)]\leq q^{-\epsilon}.
\]
\item If $\dim T\leq m_1-m_0$, $\Gamma_{m_0,m_1}(\cV,U,T)\geq\epsilon$, and $\epsilon\geq 2b^2+1$, then
\[
       \Prb[(S,C)\text{ realizes }(\cV,U,T)]
       \geq 1-q^{-\epsilon+12b^2}.
\]
\end{enumerate}
\end{theorem}

\begin{proof}
Part (a) is Proposition~\ref{qlcl:prop:fixed-subcritical}.  Part (b) is Proposition~\ref{qlcl:prop:fixed-supercritical}.
\end{proof}

\section{Threshold theorem for quotient LCL properties}\label{qlcl:sec:qlcl-threshold}

A quotient LCL property is a finite collection of fixed types for the given block length.  In asymptotic applications the size of this finite collection may grow with $n$, but the threshold bounds keep the explicit factor $|\cF|$.  The sharp parameter for the property is therefore the largest quotient-minor gap among all admissible profiles and logical row spaces.

For a quotient LCL property $\cP(\cF,\cD)$ define the feasible maximum gap
\begin{equation}\label{qlcl:eq:Gmax}
       \Gamma^{\max}_{m_0,m_1}(\cF,\cD)
       =\max_{\cV\in\cF}
        \max_{\substack{T\in\cD,\, U\supseteq T,\, U\neq0\\ \dim T\leq m_1-m_0}}
        \Gamma_{m_0,m_1}(\cV,U,T).
\end{equation}
If the indexing set is empty, set $\Gamma^{\max}=-\infty$.

\begin{theorem}[Quotient LCL threshold, parity-check ensemble]\label{qlcl:thm:qlcl-threshold}
Let $(S,C)\sim\PCSS(n,q;m_0,m_1)$ and let $\cP(\cF,\cD)$ be a $b$-quotient LCL property with finite profile family $\cF$.
\begin{enumerate}[label=(\alph*)]
\item \textbf{Subcritical region.}  If
\[
       \Gamma^{\max}_{m_0,m_1}(\cF,\cD)
       \leq -\epsilon,
\]
then
\[
       \Prb[(S,C)\text{ satisfies }\cP(\cF,\cD)]
       \leq |\cF|\,q^{-\epsilon+2b^2}.
\]
\item \textbf{Supercritical region.}  If
\[
       \Gamma^{\max}_{m_0,m_1}(\cF,\cD)
       \geq \epsilon
       \qquad\text{and}\qquad
       \epsilon\geq 2b^2+1,
\]
then
\[
       \Prb[(S,C)\text{ satisfies }\cP(\cF,\cD)]
       \geq 1-q^{-\epsilon+12b^2}.
\]
\end{enumerate}
\end{theorem}

\begin{proof}
For (a), if the property holds then some flag $(U,T)$ and some profile $\cV\in\cF$ are realized, and any realized flag is parity-feasible: when $\dim T>m_1-m_0$, the realization probability vanishes by Lemma~\ref{qlcl:lem:fixed-matrix-prob}, since $\iota(c,t)=0$ for $t>c$.  The number of flags in $\F_q^b$ is at most $q^{2b^2}$.  Apply Theorem~\ref{qlcl:thm:fixed-type}(a) to each parity-feasible flag and union bound.

For (b), choose a profile and feasible flag attaining gap at least $\epsilon$.  Theorem~\ref{qlcl:thm:fixed-type}(b) implies that this one flagged profile is realized with probability at least $1-q^{-\epsilon+12b^2}$, and that realization witnesses $\cP(\cF,\cD)$.
\end{proof}

In applications $b$ is fixed and $|\cF|\leq q^{o(n)}$, so the displayed bounds give exponentially sharp high-probability statements whenever the gap is $\Omega(n)$.  For list-recovery and distance profiles, $|\cF|$ is typically exponential in $n$ but independent of $q$; the large-alphabet regime makes the union-bound factor negligible.

\section{Uniform nested subspaces}\label{qlcl:sec:uniform}

The parity-check ensemble was chosen above for simplicity of exposition, but the natural random object is a uniformly random nested pair of prescribed dimensions.  Conditioning on full rank transfers the threshold theorem with only a constant loss in probability.

\begin{definition}[Uniform nested ensemble]
Let $\Nest(n,q;m_0,m_1)$ be the uniform distribution over all nested pairs
\[
       S\subseteq C\subseteq\F_q^n,
       \qquad
       \dim S=m_0,
       \qquad
       \dim C=m_1 .
\]
\end{definition}

\begin{lemma}[Conditioning on full rank]\label{qlcl:lem:conditioning}
Let
\[
       H=\begin{bmatrix}H_C\\ H_S\end{bmatrix}.
\]
Conditioned on the event that $H_C$ and $H$ have full row rank, $\PCSS(n,q;m_0,m_1)$ is distributed as $\Nest(n,q;m_0,m_1)$.  Moreover this full-rank event has probability at least
\[
       \gamma_q=\prod_{j=1}^{\infty}(1-q^{-j})^2>0.
\]
\end{lemma}

\begin{proof}
A uniformly random full-rank parity-check matrix $H_C\in\F_q^{a\times n}$ has a uniformly random kernel of dimension $m_1$.  Conditional on $C=\ker H_C$, the additional checks $H_S$ induce a uniformly random linear map $C\to\F_q^c$.  Conditional on this induced map having rank $c$, its kernel is a uniformly random $m_0$-dimensional subspace of $C$.  This is exactly the uniform nested ensemble.

The probability that a random $r\times N$ matrix over $\F_q$ has full row rank is $\prod_{j=N-r+1}^{N}(1-q^{-j})$, which is at least $\prod_{j=1}^{\infty}(1-q^{-j})$.  Apply this once to $H_C$ and once to the induced map on $C$.
\end{proof}

For the uniform model we maximize only over flags that are feasible for the fixed dimensions.  Define
\begin{equation}\label{qlcl:eq:Gmax-unif}
\begin{aligned}
       \Gamma^{\max,\mathrm{unif}}_{m_0,m_1}(\cF,\cD)
       =\max_{\cV\in\cF}
        \max\Bigl\{
        \Gamma_{m_0,m_1}(\cV,U,T):\;&T\in\cD,\ U\supseteq T,\ U\neq0,\\
        &\dim(U/T)\le m_0,\ \dim T\le m_1-m_0
        \Bigr\}.
\end{aligned}
\end{equation}
with value $-\infty$ if the indexing set is empty.  Lemma~\ref{qlcl:lem:auto-feasibility} shows that no other flags can occur in a deterministic nested pair of dimensions $(m_0,m_1)$.

\begin{theorem}[Quotient LCL threshold, uniform nested ensemble]\label{qlcl:thm:uniform-threshold}
Let $(S,C)\sim\Nest(n,q;m_0,m_1)$ and let $\cP(\cF,\cD)$ be a $b$-quotient LCL property with $\cF$ finite.
\begin{enumerate}[label=(\alph*)]
\item \textbf{Subcritical region.}  If
\[
       \Gamma^{\max,\mathrm{unif}}_{m_0,m_1}(\cF,\cD)\leq -\epsilon,
\]
then
\[
       \Prb[(S,C)\text{ satisfies }\cP(\cF,\cD)]
       \leq \gamma_q^{-1}|\cF|\,q^{-\epsilon+2b^2}.
\]
\item \textbf{Supercritical region.}  If
\[
       \Gamma^{\max,\mathrm{unif}}_{m_0,m_1}(\cF,\cD)\geq \epsilon
       \qquad\text{and}\qquad
       \epsilon\geq 2b^2+1,
\]
then
\[
       \Prb[(S,C)\text{ satisfies }\cP(\cF,\cD)]
       \geq 1-\gamma_q^{-1}q^{-\epsilon+12b^2}.
\]
\end{enumerate}
\end{theorem}

\begin{proof}
Let $\mathcal R$ be the full-rank event from Lemma~\ref{qlcl:lem:conditioning}.  Conditional on $\mathcal R$, the parity-check ensemble is the uniform nested ensemble.

For (a), if the property holds in a deterministic nested pair of dimensions $(m_0,m_1)$, Lemma~\ref{qlcl:lem:auto-feasibility} implies that some feasible flag in the indexing set of \eqref{qlcl:eq:Gmax-unif} is realized.  The number of flags in $\F_q^b$ is at most $q^{2b^2}$.  Theorem~\ref{qlcl:thm:fixed-type}(a), applied in the parity-check ensemble and union bounded over feasible flags and profiles, gives unconditional probability at most $|\cF|q^{-\epsilon+2b^2}$.  Conditioning on $\mathcal R$ multiplies this by at most $\gamma_q^{-1}$.

For (b), choose a feasible profile and flag attaining gap at least $\epsilon$.  Theorem~\ref{qlcl:thm:fixed-type}(b) gives probability at most $q^{-\epsilon+12b^2}$ that this flag is not realized in the parity-check ensemble.  Conditioning on $\mathcal R$ again multiplies the failure probability by at most $\gamma_q^{-1}$.
\end{proof}

\section{CSS codes}\label{qlcl:sec:css}

We now translate the one-sector nested-pair language into CSS terminology. This involves applying the above framework to both X and Z sectors separately at first, and then designing a joint sector formulation that enables one to analyze qLCL properties in the joint sector by looking at behavior in the individual sectors.

We use the standard $q$-ary CSS construction~\cite{CSS1,CSS2}.  Write $q=p^e$ with $p$ prime, let $\omega_p=e^{2\pi i/p}$, and let $\Tr=\Tr_{\F_q/\F_p}$ denote the field trace.  On one qudit $\C^q$ with computational basis $\{\ket a:a\in\F_q\}$, the generalized Pauli operators are
\[
       X(x)\ket a=\ket{a+x},
       \qquad
       Z(z)\ket a=\omega_p^{\Tr(za)}\ket a,
\]
extended multiplicatively to $n$ qudits for $x,z\in\F_q^n$.  Their commutation is governed by the trace form of the $\F_q$-bilinear pairing:
\[
       Z(z)X(x)=\omega_p^{\Tr(\ip xz)}\,X(x)Z(z),
\]
so an individual pair $X(x),Z(z)$ commutes iff $\Tr(\ip xz)=0$.  For prime $q$ the trace is the identity map and this is the condition $\ip xz=0$.  For non-prime $q$ the two conditions differ for individual pairs: over $\F_4$ one has $\Tr_{\F_4/\F_2}(1)=1+1^2=0$, so $X(1)$ and $Z(1)$ commute even though $\ip 11=1\neq0$.  What the CSS construction needs, however, is commutation against entire $\F_q$-linear spaces of labels, and there the two conditions coincide.

\begin{lemma}[Subspace commutation via the $\F_q$-pairing]\label{qlcl:lem:trace-pairing}
Let $D\leq\F_q^n$ be an $\F_q$-linear subspace and let $x\in\F_q^n$.  Then $X(x)$ commutes with $Z(z)$ for every $z\in D$ iff $x\in D^\perp$, where $D^\perp$ is the orthogonal complement with respect to the $\F_q$-bilinear form $\ip\cdot\cdot$.
\end{lemma}

\begin{proof}
If $x\in D^\perp$, then $\ip xz=0$, hence $\Tr(\ip xz)=0$, for every $z\in D$.  Conversely, suppose $\Tr(\ip xz)=0$ for every $z\in D$, and fix $z\in D$.  Since $D$ is $\F_q$-linear, $cz\in D$ for every $c\in\F_q$, so $\Tr(c\ip xz)=0$ for all $c\in\F_q$.  The trace form $(a,b)\mapsto\Tr(ab)$ is nondegenerate on $\F_q$ because the extension $\F_q/\F_p$ is separable, so $\ip xz=0$.  As $z\in D$ was arbitrary, $x\in D^\perp$.
\end{proof}

Thus, although the commutation of an individual pair $X(x),Z(z)$ is a trace condition rather than an $\F_q$-bilinear one, every commutation requirement appearing in the CSS construction is of the subspace form of Lemma~\ref{qlcl:lem:trace-pairing}, since stabilizer groups are generated by $\F_q$-linear label spaces.  The $\F_q$-bilinear pairing is therefore the correct bookkeeping device, for prime and non-prime $q$ alike.

\begin{definition}[CSS code from a nested pair]\label{qlcl:def:css-from-pair}
Given $S\subseteq C\subseteq\F_q^n$, define the CSS code $Q(S,C)$ whose $X$-type stabilizer labels are $S$ and whose $Z$-type stabilizer labels are $C^\perp$:
\[
       \mathsf S_X=S,
       \qquad
       \mathsf S_Z=C^\perp.
\]
By Lemma~\ref{qlcl:lem:trace-pairing}, every $X$-type stabilizer $X(x)$, $x\in S$, commutes with every $Z$-type stabilizer $Z(z)$, $z\in C^\perp$, iff $S\subseteq(C^\perp)^\perp=C$, which holds by hypothesis.  The number of encoded qudits is
\[
       k=n-\dim S-\dim C^\perp=\dim C-\dim S.
\]
\end{definition}

For this CSS code, the $X$-type normalizer labels are $C_X = C$ and the $X$-type stabilizer labels are $S_X = S$, so the $X$-logical quotient is
\[
       C_X/S_X = C/S.
\]
Similarly, the $Z$-logical quotient, for $C_Z = S^\perp$ and $S_Z = C^\perp$ is
\[
       C_Z/S_Z = S^\perp/C^\perp.
\]

Thus, the joint logical quotient of the CSS code is
\[
    (C_X + C_Z)/(S_X + S_Z) = (C + S^\perp)/(S + C^\perp) \cong (C/S) + (S^\perp/C^\perp).
\]

\subsection{Joint Sector Formulation}\label{sub:joint-sec}
We require a few definitions before we can detail the interaction of the joint sector of a CSS code with local profiles.

We retain the notation of the one-sector theory.  For a sector $s \in \{X, Z\}$, and for subspaces $T \le U \le \F_q^b$, define
\[
 \kappa_s(U,T) \defeq (n-m_1^s)\dim U + (m_1^s-m_0^s)(\dim U - \dim T).
\]

\begin{lemma}[Positivity of $\kappa_s(U,T)-\kappa_s(W,T\cap W)$]\label{obs:pos-of-d}
    For a subspace $W<U$,
    \[
        \kappa_s(U,T)-\kappa_s(W,T\cap W) \ge (n-m_1)(\dim U - \dim W) > 0.
    \]
\end{lemma}
\begin{proof}
    Write $d = \dim U - \dim W$, and $\ell = \dim W - \dim (T \cup W)$.
    \begin{align*}
        \kappa_s(U,T)-\kappa_s(W,T\cap W) &= (n-m_1^s)d + (m_1-m_0)(d-\ell) \\
        &\ge (n-m_1^s)d\\
        &>0.
    \end{align*}
    The first inequality follows because $d> \ell$, and the second is true because $d>0$, as $W$ is a strict subspace of $U$.
\end{proof}

\subsubsection{The normalized slope}

\begin{definition}[Flag slope]\label{def:flagged-sop}
For $U\ne0$, $T\le U$, and $W<U$, define
\begin{equation}\label{eq:flagged-slope}
 \lambda_{s}(\cV;U,T)
 :=\min_{W<U}
 \frac{a_{\cV}(U)-a_{\cV}(W)}
      {\kappa_s(U,T)-\kappa_s(W,T\cap W)}.
\end{equation}
Set $\lambda_{s}(\cV;0,0):=+\infty$.
\end{definition}

The critical value is $1$. In the one-sector qLCL formulation, $\lambda_{s}(\cV;U,T) < 1$ means that some proper quotient minor gains profile mass more
slowly than it pays random containment cost, and so the probability of the flagged type $(\cV, U ,T)$ appearing is
very low. This is formalized in the following proposition.

\begin{proposition}[Realization probability in terms of $\lambda$]\label{prop:fixed-type}
Suppose a random nested pair has dimensions $(m_0^s,m_1^s)$. If a nonzero
flagged type $(\cV, U ,T)$ satisfies
\begin{equation}\label{eq:lambda-low}
        \lambda_{s}(\cV;U,T)\le1-\eta,
\end{equation}
then for $(S, C) \sim \PCSS(n, q; m_0^s, m_1^s)$,
\begin{equation}\label{eq:fixed-type-input}
 \Pr[(S,C)\text{ realizes }(\cV,U,T)]
 \le q^{-\eta(n-m_1^s)}.
\end{equation}
For the uniform exact-dimensional model, the same estimate holds up to the
standard constant full-rank conditioning factor.
\end{proposition}
\begin{proof}
    Observe that \cref{eq:lambda-low} implies the existence of a subspace $W < U$ such that
    \begin{align*}
        a_{\cV}(U)-a_{\cV}(W) -(\kappa_s(U,T)-\kappa_s(W,T\cap W)) &\le -\eta \cdot (\kappa_s(U,T)-\kappa_s(W,T\cap W)) \\
        &\le -\eta \cdot (n-m^s_1).
    \end{align*}
    The second inequality follows from Lemma~\ref{obs:pos-of-d} and the fact that $\dim U -\dim W \ge 1$.
    Further, we see that $a_{\cV}(U)-a_{\cV}(W) -(\kappa_s(U,T)-\kappa_s(W,T\cap W)) \le -\eta \cdot (n-m_1^s)$ implies
    \[
        \Gamma_{m_0^s,m_1^s}(\cV, U ,T) \le -\eta \cdot (n-m_1^s).
    \]
    Applying Proposition~\ref{qlcl:prop:fixed-subcritical} with the above bound, the proposition follows.
\end{proof}

Thus, $\lambda$ being bounded away from $1$ implies that a random CSS code avoids realizing $(\cV, U, T)$ with high probability.

We now define the single-sector rate threshold for a fixed flag:
\begin{definition}[Single Sector Rate Threshold]
Upon denoting $r_s=m_1^s/n$, define:
\begin{equation}
    \label{eq:flagged-rate-score}
    R^{\mathrm{flag}}_{s}(\cV;U,T)
    :=
    1-(1-r_s)\lambda_{s}(\cV;U,T).
\end{equation}
\end{definition}

We provide some intuition for this definition. Rearranging \cref{eq:flagged-rate-score} gives
\[
    R^{\mathrm{flag}}_{s}(\cV;U,T)-r_s
    =
    (1-r_s)\bigl(1-\lambda_{s}(\cV;U,T)\bigr).
\]
Hence
\begin{equation}\label{eq:bi-implies-rate-lambda}
    \lambda_{s}(\cV;U,T)<1
    \qquad\Longleftrightarrow\qquad
    r_s<R^{\mathrm{flag}}_{s}(\cV;U,T).
\end{equation}
Or, more precisely, for any $\delta > 0$,
\begin{equation}\label{eq:precise-bi-implies}
    R^{\mathrm{flag}}_{s}(\cV;U,T)
    \ge r_s+\delta
    \iff
    \lambda_{s}(\cV;U,T)
    \le
    1-\frac{\delta}{1-r_s}.
\end{equation}
Thus, $R^{\mathrm{flag}}_{s}(\cV;U,T)$ serves as a one-sided threshold: if the sector rate is below this quantity, then by Proposition~\ref{prop:fixed-type}, the probability of realizing $(\cV, U ,T)$ is low.

We also require the definition of the rate threshold from the LCL framework of \cite{LMS25}:
\begin{definition}[Rate Threshold, Classical Setting]
    \[
 R_{\rm cl}(\cV;U)
 :=1-\min_{W<U}
 \frac{a_{\cV}(U)-a_{\cV}(W)}{n(\dim U-\dim W)}.
\]
\end{definition}

\begin{proposition}[One-sector fixed-type comparison]
\label{prop:sub-one-sector-comparison}
Every nonzero flag $T\le U$ satisfies
$R_s^{\rm flag}(\cV;U,T)\ge R_{\rm cl}(\cV;U)$, with equality for $T=U$.
Moreover, for every $W<U$,
\begin{equation}
\label{eq:sub-minor-rate-bound}
 a_{\cV}(U)-a_{\cV}(W)
 \ge n\bigl(1-R_s^{\rm flag}(\cV;U,T)\bigr)(\dim U-\dim W).
\end{equation}
If $A$ realizes $(U,T)$, then $T=U$ iff
$\operatorname{im}(A)\cap S_s=\{0\}$.
\end{proposition}

\begin{proof}
For $d:=\dim U-\dim W$ and
$\ell:=\dim T-\dim(T\cap W)$,
\[
 \kappa_s(U,T)-\kappa_s(W,T\cap W)
 =(n-m_1^s)d+(m_1^s-m_0^s)(d-\ell)
 \ge n(1-r_s)d.
\]
Combining this with \eqref{eq:flagged-rate-score} gives
\eqref{eq:sub-minor-rate-bound} and hence the threshold comparison.  For
$T=U$, one has $\ell=d$ for every minor, giving equality.  The last statement
is equivalent to equality of the logical and ordinary kernels.
\end{proof}

We also use the profile-mass supermodularity
\begin{equation}
\label{eq:sub-mass-supermodular}
 a_{\cV}(U_1)+a_{\cV}(U_2)
 \le a_{\cV}(U_1\cap U_2)+a_{\cV}(U_1+U_2),
\end{equation}
which follows coordinatewise from the dimension formula.

The following theorem relates the classical rate threshold to the newly defined threshold. Let $T_X \le U_X$ and $T_Z \le U_Z$ be subspaces of $\F_q^b$.
\begin{theorem}
\label{thm:sub-asymmetric-sector-reduction}
For every common-profile joint witness,
\[
 \boxed{
 R_{\rm cl}(\cV;U_X+U_Z)
 \le
 \max\{R_X^{\rm flag}(\cV;U_X,T_X),
       R_Z^{\rm flag}(\cV;U_Z,T_Z)\}.}
\]
\end{theorem}

\begin{proof}
Let $R_*$ be the larger marginal fixed-type threshold and fix
$W<U_X+U_Z$. Two applications of
\eqref{eq:sub-mass-supermodular} (first with $U_X, W$, and the second with $U_Z, W+U_X$) gives
\begin{align*}
 a_{\cV}(U_X+U_Z)-a_{\cV}(W)
 \ge{}&a_{\cV}(U_X)-a_{\cV}(U_X\cap W)\\
 &+a_{\cV}(U_Z)-a_{\cV}(U_Z\cap(W+U_X)).
\end{align*}
By Proposition~\ref{prop:sub-one-sector-comparison}, the right side is at
least
\[
 n(1-R_*)\bigl[\dim U_X-\dim(U_X\cap W)
 +\dim U_Z-\dim(U_Z\cap(W+U_X))\bigr].
\]
The bracket equals $\dim(U_X+U_Z)-\dim W$.  The definition of
$R_{\rm cl}(\cV;U_X+U_Z)$ now gives the claim.
\end{proof}

For a finite profile family $\cF$, let
\[
 R_{\rm cl}(\cF):=\min_{\cV\in\cF,\,U\in\LDist_b}R_{\rm cl}(\cV;U).
\]
This is the classical rate threshold for the LCL property $\cF$.

\subsubsection{Joint Realization}
We now define joint type realization for a CSS code. Until this point, we have only discussed realizations for a particular sector (X or Z). That is, when we were describing the behavior of random CSS codes realizing $(\cV, U, T)$, we were either talking only about $C_X/S_X$ or $C_Z/S_Z$, the logical space of operators corresponding to either $X$ or $Z$. However, for most applications of CSS codes, we need to look at the joint logical space:
\[
    (C_X + C_Z)/(S_X + S_Z) \cong (C_X/S_X) \oplus (C_Z/S_Z).
\]
Fortunately, the structure of CSS codes allows to analyze the behavior in the joint logical space by analyzing behavior in the individual sectors.
More formally, let $A \in \F_Q^{n \times b}$ be a matrix whose columns consisting of coset representatives (belonging to $C_X + C_Z$) that satisfy a $b$-local profile $\cV=(V_1,\ldots,V_n)$ (that is, the $i$th row of the matrix lies in $V_i$, for all $i \in [n]$). Then, by the definition of a CSS code, we know that there exist matrices $A_X, A_Z \in \F_Q^{n \times b}$ such that the columns of $A_X$ (respectively, $A_Z$) belong to $C_X$ (respectively, $C_Z$). These columns are representatives of the cosets in $C_X/S_X$ (respectively, $C_Z/S_Z$).
Thus, we can seek to avoid containing all columns of matrix $A$ simultaneously in a CSS code by avoiding one of $A_X$, $A_Z$ within the individual sectors.

\begin{definition}[Joint Realization of a Local Profile]\label{def:join-realize}
    Let $C_X, C_Z$ and $S_X, S_Z$ be the normalizer and stabilizer labels for the two sectors of CSS code $C = Q(S_X, C_X)$.
    For subspaces $U_s, T_s$, and $\cV$ as defined above, we say that $C$ \emph{jointly realizes} $(\cV, U_X+U_Z, T_X+T_Z)$ if there exist matrices  $A_s \in \F_q^{n \times b}$ for $s \in \{X,Z\}$ such that:
    \begin{enumerate}
        \item  all $i \in [n]$, the $i$th row of $A_s$ lies in $V_i$;
        \item $\row(A_s)=U_s$;
        \item every column of $A_s$ lies in $C_s$;
        \item $\Lambda_{S_s}(A_s) = T_s$.
    \end{enumerate}
    Such matrices $(A_X, A_Z)$ are called realizations of $(\cV, U, T)$.
\end{definition}

As in the classical case, we are interested in the case where the coset representatives are pairwise distinct. Therefore, we want to avoid containing joint realizations $(\cV, U_X+U_Z, T_X+T_Z)$ where $T_X+T_Z \in \LDist_b$.
To this end, for profile family $\cF$, we define
\[
 R^{\rm joint}(\cF)
 :=\min_{\substack{\cV\in\cF,\ T_X\le U_X,\ T_Z\le U_Z\\
                   T_X+T_Z\in\LDist_b}}
 \max\{R_X^{\rm flag}(\cV;U_X,T_X),
       R_Z^{\rm flag}(\cV;U_Z,T_Z)\}.
\]
Note that if both $r_X, r_Z$ are strictly less than $R^{\rm joint}(\cF)$, then at least one
\[
    r_X < R_X^{\rm flag}(\cV;U_X,T_X) \quad r_Z < R_Z^{\rm flag}(\cV;U_Z,T_Z)
\]
is satisfied, and hence by combining \eqref{eq:bi-implies-rate-lambda} and Proposition~\ref{prop:fixed-type}, the probability of realizing at least one of $(\cV;U_X,T_X), (\cV;U_Z,T_Z)$ is very low. By the above discussion, this implies that the probability of a joint realization of $(\cV;U_X+U_Z,T_X+T_Z)$ is low too.

We now state and prove the central theorem of this section: that the ``joint'' rate threshold defined above is actually equal to the classical rate threshold.
\begin{theorem}[Threshold Preservation]
\label{thm:sub-joint-sector-preservation}
For arbitrary sector parameters,
\[
   \boxed{R^{\rm joint}(\cF)=R_{\rm cl}(\cF).}
\]
\end{theorem}

\begin{proof}
Every subspace $U_X+U_Z$ is in $\LDist_b$ because it
contains $T_X+T_Z\in\LDist_b$.  Theorem~\ref{thm:sub-asymmetric-sector-reduction}
therefore gives $R_{\rm cl}(\cF) \le R^{\rm joint}(\cF)$.
Conversely, take the $(\cV, U)$ that minimizes $R_{\rm cl}(\cF)$ and choose $(U_X,T_X)=(U,U)$ and $(U_Z,T_Z)=(0,0)$.
Thus, $T_X+T_Z=U \in \LDist_b$, and we have that $\max\{R_X^{\rm flag}(\cV;U_X,T_X), R_Z^{\rm flag}(\cV;U_Z,T_Z)\} = R_X^{\rm flag}(\cV;U,U)$, which is equal to $R_{\rm cl}(\cV, U)$ by Proposition~\ref{prop:sub-one-sector-comparison}.
The theorem follows.
\end{proof}

Thus if $\max\{r_X \defeq m_1^X/n,r_Z \defeq m_1^Z/n \} + \delta  \le R_{\rm cl}(\cF)$ for some $\delta>0$, then by \eqref{eq:precise-bi-implies}
\[
    \lambda_X(\cV; U_X, T_X) \le 1-\frac{\delta}{1-r_X} ;\quad \lambda_Z(\cV; U_Z, T_Z) \le 1-\frac{\delta}{1-r_Z}
\]
for every $\cV \in \cF$, and $U_s, T_s$ such that $T_X+T_Z \in \LDist_b$.

For a CSS code $C$ as defined in Definition~\ref{def:join-realize}, let $E(C, \cF)$ denote the event that $C$ realizes $(\cV, U_X+U_Z, T_X+T_Z)$ for some $\cV \in \cF$, $T_X \le U_X, T_Z \le U_Z$ satisfying $T_X+T_Z \in \LDist_b$.
\begin{corollary}[Probability of Joint Sector Satisfaction]\label{cor:joint-form}
For the event $E$ defined above, we have:
\begin{equation}\label{eq:prob-join-sat}
    Pr_C[E(C, \cF)] \le |\cF|\cdot q^{4b^2} \cdot q^{-\delta n}.
\end{equation}
\end{corollary}
\begin{proof}
    The corollary follows by union bounding over all $\cV \in \cF$, $U_X,U_Z,T_X,T_Z$ satisfying the appropriate conditions, and then applying Proposition~\ref{prop:fixed-type}.
\end{proof}

\subsection{The Balanced Case}

For a CSS code $C_Z^\perp\subseteq C_X$, let
$k_X:=\dim C_X$, $k_Z:=\dim C_Z$, and $k_Q:=k_X+k_Z-n$.  The two sector
pairs have parameters $(n-k_Z,k_X)$ and $(n-k_X,k_Z)$, so
$r_X=k_X/n$, $r_Z=k_Z/n$, and $r_X+r_Z=1+R_Q$ for $R_Q:=k_Q/n$.  Their
costs are
\[
 \kappa_X(U,T)=k_Z\dim U-k_Q\dim T,\qquad
 \kappa_Z(U,T)=k_X\dim U-k_Q\dim T.
\]

If $k_X=k_Z$, then $r_X=r_Z=:r=(1+R_Q)/2$.
The general condition $\max\{r_X,r_Z\} +\delta \le R_{\rm cl}(\cF)$ therefore
reduces to $(1+R_Q)/2 +\delta \le R_{\rm cl}(\cF)$.

\subsection{Consequences of \cref{thm:sub-joint-sector-preservation}}
The most important consequence of the threshold preservation theorem is that the threshold rates for random linear codes apply in the quantum setting. This allows one to prove upper/lower bounds on threshold rates in the classical setting, and port them to the quantum side directly. Indeed, several such bounds have already been established for various LCL properties. We discuss the bounds for list decoding and list recovery.

In \cite{AGL24}, the authors established a lower bound on $R_{\rm cl}$ for the case of list-decoding, and \cite{LMS25}  proved a tight upper bound for the same quantity. In the case of list recovery, the best known lower bound was established by \cite{BCDZ26b} (followed by a small improvement in \cite{GG26}), while the lower bound is from \cite{LS25}.

\begin{corollary}\label{cor:list-dec-joint}
    Let $\cF_{\rm LD}$ denote the family of local profiles associated with (the complement of) $(\rho, L)$-list decodability. Then,
    \begin{equation}\label{eq:list-dec-joint}
        R^{\rm joint}(\cF_{\rm LD}) = R_{\rm cl}(\cF_{\rm LD}) = 1-\frac{L+1}{L} \rho.
    \end{equation}
\end{corollary}
\begin{proof}
    This follows by combining \cref{thm:sub-joint-sector-preservation} with the list-decoding upper and lower bounds of \cite{AGL24, LMS25}.
\end{proof}

For list-recovery, the known bounds are not tight. We only state the lower bound here.
\begin{corollary}\label{cor:list-rec-joint}
    Let $\cF_{\rm LR}$ denote the family of local profiles associated with (the complement of) $(\rho, \ell, L)$-list recoverability. Then, if for $\eps>0$, $\rho, \ell, L$ satisfy
    \[
        L \le \left( \frac{\ell}{1-\rho}\right)^{(1-\rho)/\eps}
    \]
    then
    \begin{equation}\label{eq:list-rec-joint}
        R^{\rm joint}(\cF_{\rm LR}) = R_{\rm cl}(\cF_{\rm LR}) \ge  1-\rho-\eps.
    \end{equation}
\end{corollary}
\begin{proof}
    This follows by combining \cref{thm:sub-joint-sector-preservation} with the list-recovery bound in \cite{GG26}, where they proved that random linear codes of rate $R$ achieve $(1-R-\eps, \ell, (\ell/(R+\eps))^{(R+\eps)/\eps})$-list recoverability.
\end{proof}

In the balanced case, we see that random CSS codes exist when:
\[
    \frac{1+R_Q}{2} < 1-\frac{L+1}{L} \rho; \quad  \frac{1+R_Q}{2} < 1-\rho-\eps
\]
for list decoding and list recovery respectively.

We note that the list sizes are the same in the classical and quantum settings. This is in contrast to the result achieved in \cite{BGG24}, where the list sizes are $L, L^2$ in the classical and quantum settings, respectively.
Moreover, in Section~\ref{sec:css-ael}, we show how to derandomize these results and achieve explicit quantum codes whose list sizes match that of classical ones.

\section{qLCL for Folded Coordinates}\label{sec:folded-prelims}
This section is concerned with the qLCL framework for folded coordinates. The definitions and theorems are similar to their counterparts appearing in the sections above. Indeed, this section is a generalization, and the purpose of presenting the theorems and proofs in the earlier sections was for the sake of exposition only, as those proofs are clearer and require fewer technical details. We note that this formulation is inspired by the LCL formulation in \cite{GGH26}.

For a finite-dimensional space \(V\), \(\calL(V)\) denotes its lattice of linear subspaces.  If \(U\le V\), then \(U^\circ\le V^*\) is the annihilator.  We use
\[
  \dim U+\dim U^\circ=\dim V,
  \qquad (U^\circ)^\circ=U.
\]

\subsection{Folded coordinates}

Let \(W\cong\F_q^s\) be one edge alphabet.  A folded length-\(n\) ambient space is
\[
  W^n\defeq W^n=\bigoplus_{i=1}^n W_i .
\]
For \(A\le W^n\), define the \(i\)-th coordinate-zero subspace and coordinate rank by
\[
  K_i\defeq \{x\in W^n:x_i=0\},\qquad
  A_i\defeq A\cap K_i,
\]
\[
  \rho_i(A)\defeq \dim A-\dim A_i=\dim \pi_i(A),
\]
where \(\pi_i:W^n\to W_i\) is projection.  Thus
\[
   \frac1n\sum_i\dim A_i\le \alpha\dim A
   \quad\Longleftrightarrow\quad
   \frac1n\sum_i\rho_i(A)\ge (1-\alpha)\dim A .
\]
The block Hamming weight of \(x\in W^n\) is
\[
  \wt(x)=|\{i\in[n]:x_i\ne 0\}|.
\]

\subsection{Bilinear forms and CSS notation}

Fix a nondegenerate \(\F_q\)-bilinear pairing \(\ip{\cdot}{\cdot}_W\) on \(W\), and extend it coordinatewise to \(W^n\).  Orthogonal complements are with respect to this pairing.  For CSS codes this is the standard \(\F_q\)-linear abstraction: the \(X\)- and \(Z\)-label spaces are dual through a nondegenerate coordinatewise pairing.  Its relation to physical Pauli commutation, which over non-prime fields is mediated by the trace pairing rather than by the \(\F_q\)-valued pairing itself, is recorded in Remark~\ref{rsd:rem:trace-pairing-folded}.

\begin{remark}[Physical commutation and the trace pairing]\label{rsd:rem:trace-pairing-folded}
As in the scalar case (Lemma~\ref{qlcl:lem:trace-pairing}), the commutation phase of a pair of generalized Pauli operators over \(\F_q\), \(q=p^e\), is the additive character \(\omega_p^{\Tr_{\F_q/\F_p}(\cdot)}\) applied to the \(\F_q\)-valued pairing.  Concretely, choose an \(\F_q\)-basis of \(W\) identifying the fixed nondegenerate form \(\ip{\cdot}{\cdot}_W\) with the standard form on \(\F_q^s\); each folded coordinate is then a block of \(s\) qudits, and \(X(x)\) commutes with \(Z(z)\) iff \(\Tr_{\F_q/\F_p}(\ip xz)=0\), where \(\ip{\cdot}{\cdot}\) is the coordinatewise extension of \(\ip{\cdot}{\cdot}_W\).  For an individual pair of vectors this trace condition is strictly weaker than \(\ip xz=0\) when \(q\) is not prime.  For \(\F_q\)-linear subspaces the two conditions coincide: every code space in this paper is \(\F_q\)-linear, hence closed under scalar multiplication, so \(X(x)\) commutes with all of \(\{Z(z):z\in D\}\) iff \(x\in D^\perp\), by the argument of Lemma~\ref{qlcl:lem:trace-pairing} applied verbatim to the pairing on \(W^n\).  All CSS commutation requirements below are of this subspace form, so the \(\F_q\)-bilinear model used throughout is faithful to physical Pauli commutation; the \(\F_q\)-valued pairing should not be read as an individual-pair commutation criterion.
\end{remark}

A folded CSS code is specified by a nested pair
\[
   C_2\subseteq C_1\subseteq W^n.
\]
In the usual two-code notation, one may take \(C_Z=C_1\) and \(C_X=C_2^\perp\), so \(C_X^\perp=C_2\subseteq C_Z\).  The two logical sectors are
\[
   L_X=C_1/C_2,
   \qquad
   L_Z=C_2^\perp/C_1^\perp.
\]
They have the same \(\F_q\)-dimension
\[
   k_Q=\dim C_1-\dim C_2,
\]
which is the quantum dimension.  The block logical distances are
\[
   d_X=\min_{x\in C_1\setminus C_2}\wt(x),
   \qquad
   d_Z=\min_{z\in C_2^\perp\setminus C_1^\perp}\wt(z).
\]
The pairing on \(W^n\) descends to a nondegenerate pairing
\[
  (C_1/C_2)\times (C_2^\perp/C_1^\perp)\to \F_q.
\]
Indeed, the left radical of the restriction to \(C_1\times C_2^\perp\) is \(C_1\cap (C_2^\perp)^\perp=C_2\), and similarly on the right.

\begin{remark}[Folded metric versus scalar metric]\label{rsd:rem:folded-metric}
All subspace-design and distance statements in the folded part of the paper use the block metric on the coordinates of \(W^n\).  If \(W=\F_q^s\) is unfolded into scalar coordinates, the scalar length becomes \(ns\).  A block-weight lower bound \(\wt_{\rm block}(x)\ge \delta n\) implies scalar weight at least \(\delta n\), and hence scalar relative distance at least \(\delta/s\), but not necessarily \(\delta\).  Thus the construction obtained later is binary in the folded sense: it is a CSS code over a constant-dimensional binary vector alphabet, with Singleton-scale guarantees in the folded metric.
\end{remark}

\begin{remark}[Self-dual-containing specialization]
A symmetric CSS pair is obtained from \(S=N^\perp\subseteq N\).  Then \(C_2=S\) and \(C_1=N\), the two logical sectors are naturally dual copies of \(N/S\), and the quantum rate is \((2\dim N-ns)/(ns)\).  The general nested-pair formulation below includes this case but is more convenient for the probabilistic inner-code construction.
\end{remark}

\subsection{Folded qLCL witnesses}

Let \(B\) be a finite-dimensional coefficient space.  For a linear map
\[
      M:B\to W^n,
\]
write \(M_i=\pi_i\circ M\in\Hom(B,W)\) (recall that $\Hom(B,W)$ is the set of all linear maps from $B$ to $W$).  The ordinary row space of \(M\) is
\[
      \row(M)=(\ker M)^\circ\le B^*.
\]
If \(S\subseteq N\subseteq W^n\) and \(M(B)\subseteq N\), define the logical kernel and logical row space by
\[
      K_S(M)=M^{-1}(S),
      \qquad
      \Lambda_S(M)=K_S(M)^\circ\le B^*.
\]
Then \(\Lambda_S(M)\le \row(M)\), with equality if and only if \(M(B)\cap S=0\). 

For \(U\le B^*\), let
\[
      \Hom_U(B,W)=\{f\in\Hom(B,W):U^\circ\subseteq\ker f\}.
\]
Equivalently, \(\Hom_U(B,W)\cong\Hom(U^*,W)\cong W\otimes U\).  A folded local profile is a tuple
\[
      \mathfrak V=(\mathfrak V_1,\ldots,\mathfrak V_n),
      \qquad \mathfrak V_i\le\Hom(B,W),
\]
and its mass on \(U\) is
\begin{equation}\label{eq:folded-mass}
      a_{\mathfrak V}(U)
      =\sum_{i=1}^n\dim\bigl(\mathfrak V_i\cap\Hom_U(B,W)\bigr).
\end{equation}
A folded realization of \((\mathfrak V,U,T)\) in \(S\subseteq N\) is a map \(M:B\to N\) such that \(M_i\in\mathfrak V_i\) for all \(i\), \(\row(M)=U\), and \(\Lambda_S(M)=T\).  If \(\dim_{\F_q}W=s\) and the ambient dimension is \(ns\), the corresponding random-pair potential for \(\dim S=m_0\), \(\dim N=m_1\), is
\begin{equation}\label{eq:folded-qlcl-potential}
      \Phi^{W}_{m_0,m_1}(\mathfrak V,U,T)
      =a_{\mathfrak V}(U)
       -(ns-m_0)(\dim U-\dim T)
       -(ns-m_1)\dim T .
\end{equation}
The previously stated threshold theorem (\cref{qlcl:thm:fixed-type}) from  extends to this folded formulation.  A folded constraint \(\mathfrak V_i\le\Hom(B,W)\) can couple the \(s\) scalar rows inside one folded coordinate, so we do not simply cite the scalar theorem.  Instead we restate each ingredient of the scalar proof in folded form and verify it.  The statements and constants are identical to the scalar ones because every counting step happens in the fixed coefficient space \(B\); only the per-coordinate counting changes, and it changes in the same way in the first and in the second moment.  Throughout, \(b=\dim B\), and for a flag \(T\le U\le B^*\) we write \(u=\dim U\) and \(t=\dim T\).

Define the folded quotient-minor gap by
\begin{equation}\label{eq:folded-gamma}
   \Gamma^W_{m_0,m_1}(\mathfrak V,U,T)
      \defeq \min_{U'<U}\Bigl(
         \Phi^W_{m_0,m_1}(\mathfrak V,U,T)
        -\Phi^W_{m_0,m_1}(\mathfrak V,U',T\cap U')
       \Bigr),
\end{equation}
the minimum ranging over proper subspaces \(U'<U\).

\begin{definition}[Folded parity-check and uniform ensembles]\label{rsd:def:folded-ensembles}
Fix a basis of \(W\), identifying \(W^n\) with \(\F_q^{ns}\), and put
\[
   a=ns-m_1,
   \qquad
   c=m_1-m_0 .
\]
The folded parity-check ensemble \(\PCSS^W(n,q;m_0,m_1)\) samples independent uniform matrices \(H_C\in\F_q^{a\times ns}\) and \(H_S\in\F_q^{c\times ns}\) and sets \(N=\ker H_C\) and \(S=\ker H_C\cap\ker H_S\); this is exactly \(\PCSS(ns,q;m_0,m_1)\) of Definition~\ref{qlcl:def:pcss} on the scalar unfolding.  The folded uniform ensemble \(\Nest^W(n,q;m_0,m_1)\) is likewise \(\Nest(ns,q;m_0,m_1)\).  The folded flag cost is
\[
   \kappa^W_{m_0,m_1}(U,T)
      \defeq au+c(u-t)
      =(ns-m_0)(u-t)+(ns-m_1)t,
\]
so that \(\Phi^W_{m_0,m_1}(\mathfrak V,U,T)=a_{\mathfrak V}(U)-\kappa^W_{m_0,m_1}(U,T)\) in~\eqref{eq:folded-qlcl-potential}.
\end{definition}

For \(U\le B^*\), let
\[
   \cM_{\mathfrak V}(U)
      =\{M:B\to W^n:\ M_i\in\mathfrak V_i\cap\Hom_U(B,W)\text{ for every }i\},
\]
and let \(\cM^{=}_{\mathfrak V}(U)=\{M\in\cM_{\mathfrak V}(U):\row(M)=U\}\).  Since \(\row(M)\le U\) iff \(U^\circ\subseteq\ker M=\bigcap_i\ker M_i\), iff \(M_i\in\Hom_U(B,W)\) for every \(i\), the set \(\cM_{\mathfrak V}(U)\) consists exactly of the profile-obeying maps with row space contained in \(U\), and the coordinates of such a map range independently over the vector spaces \(\mathfrak V_i\cap\Hom_U(B,W)\).  Hence
\begin{equation}\label{eq:folded-mass-count}
   |\cM_{\mathfrak V}(U)|=q^{a_{\mathfrak V}(U)} .
\end{equation}

\begin{lemma}[Folded exact-row-space count]\label{rsd:lem:folded-exact-count}
Let \(\mathfrak V\) be a folded local profile and \(U\le B^*\).  Then \(|\cM^{=}_{\mathfrak V}(U)|\le q^{a_{\mathfrak V}(U)}\).  Moreover, if
\[
   a_{\mathfrak V}(U)-a_{\mathfrak V}(U')\ge\epsilon
   \qquad\text{for every proper }U'<U,
\]
then
\[
   |\cM^{=}_{\mathfrak V}(U)|
      \ge q^{a_{\mathfrak V}(U)}\bigl(1-q^{-\epsilon+b^2}\bigr).
\]
\end{lemma}

\begin{proof}
The upper bound is~\eqref{eq:folded-mass-count}.  A map in \(\cM_{\mathfrak V}(U)\) whose row space is not \(U\) lies in \(\cM_{\mathfrak V}(U')\) for some proper \(U'<U\).  The number of subspaces of \(B^*\cong\F_q^b\) is at most \(q^{b^2}\), so
\[
   |\cM_{\mathfrak V}(U)\setminus\cM^{=}_{\mathfrak V}(U)|
      \le\sum_{U'<U}q^{a_{\mathfrak V}(U')}
      \le q^{b^2}\,q^{a_{\mathfrak V}(U)-\epsilon}.
\]
Subtracting from~\eqref{eq:folded-mass-count} proves the claim.
\end{proof}

\begin{lemma}[Folded fixed-map probability]\label{rsd:lem:folded-fixed-map}
Let \(M:B\to W^n\) have \(\row(M)=U\), \(\dim U=u\), and let \(T\le U\), \(\dim T=t\).  In the ensemble \(\PCSS^W(n,q;m_0,m_1)\),
\[
   \Prb\bigl[M(B)\subseteq N\text{ and }\Lambda_S(M)=T\bigr]
      =q^{-\kappa^W_{m_0,m_1}(U,T)}\,\iota(c,t),
\]
with \(\iota(c,t)\) as in Lemma~\ref{qlcl:lem:fixed-matrix-prob}.
\end{lemma}

\begin{proof}
Under the identification \(W^n\cong\F_q^{ns}\), the map \(M\) is a linear map \(B\to\F_q^{ns}\) with \(\ker M=U^\circ\), so \(\dim M(B)=u\).  The event \(M(B)\subseteq N\) is \(H_CM=0\), which constrains each of the \(a\) rows of \(H_C\) to lie in the codimension-\(u\) space \(M(B)^\perp\); it has probability \(q^{-au}\) and depends only on \(H_C\).

Condition on \(H_CM=0\).  Then \(K_S(M)=M^{-1}(S)=\ker(H_SM)\), and \(\Lambda_S(M)=T\) iff \(\ker(H_SM)=T^\circ\).  Since \(H_S\) is uniform and independent of \(H_C\), the map induced by \(H_SM\) on \(B/\ker M\) is a uniformly random linear map from a \(u\)-dimensional space to \(\F_q^c\).  As \(U^\circ=\ker M\subseteq T^\circ\), the condition \(\ker(H_SM)=T^\circ\) says that this random map vanishes exactly on the fixed \((u-t)\)-dimensional subspace \(T^\circ/U^\circ\): vanishing on it has probability \(q^{-c(u-t)}\), and after quotienting by it, injectivity of the induced map on the remaining \(t\)-dimensional space has probability \(\iota(c,t)\).  Multiplying the factors gives the claim, since \(au+c(u-t)=\kappa^W_{m_0,m_1}(U,T)\).
\end{proof}

Exactly as in the scalar case, a positive folded gap separates the mass of \(U\) from the mass of its proper subspaces: if \(\Gamma^W_{m_0,m_1}(\mathfrak V,U,T)\ge\epsilon\), then for every proper \(U'<U\), writing \(T'=T\cap U'\),
\begin{equation}\label{eq:folded-gamma-mass}
   a_{\mathfrak V}(U)-a_{\mathfrak V}(U')
   =\bigl(\Phi^W(U,T)-\Phi^W(U',T')\bigr)
    +\bigl(\kappa^W(U,T)-\kappa^W(U',T')\bigr)
   \ge\epsilon,
\end{equation}
because
\(\kappa^W(U,T)-\kappa^W(U',T')
   =(ns-m_0)\bigl((u-t)-(\dim U'-\dim T')\bigr)+(ns-m_1)(t-\dim T')\ge0\)
when \(T'=T\cap U'\), exactly as in Lemma~\ref{qlcl:lem:gamma-mass}.

\begin{corollary}[Folded one-flag expectation]\label{rsd:cor:folded-one-flag}
Let
\[
   X^W_{\mathfrak V,U,T}
      =\#\{M\in\cM^{=}_{\mathfrak V}(U):\ M(B)\subseteq N,\ \Lambda_S(M)=T\}.
\]
Then
\[
   \Exp X^W_{\mathfrak V,U,T}\le q^{\Phi^W_{m_0,m_1}(\mathfrak V,U,T)}.
\]
If \(\Gamma^W_{m_0,m_1}(\mathfrak V,U,T)\ge\epsilon\), \(t\le c\), and \(\epsilon\ge2b^2+1\), then
\[
   \Exp X^W_{\mathfrak V,U,T}\ge q^{\Phi^W_{m_0,m_1}(\mathfrak V,U,T)-2b^2}.
\]
\end{corollary}

\begin{proof}
Combine Lemmas~\ref{rsd:lem:folded-exact-count} and~\ref{rsd:lem:folded-fixed-map} with~\eqref{eq:folded-gamma-mass}, exactly as in Corollary~\ref{qlcl:cor:one-flag-expectation}: the upper bound uses \(|\cM^{=}_{\mathfrak V}(U)|\le q^{a_{\mathfrak V}(U)}\); the lower bound uses \(|\cM^{=}_{\mathfrak V}(U)|\ge q^{a_{\mathfrak V}(U)}(1-q^{-\epsilon+b^2})\) together with \(\iota(c,t)\ge q^{-t}\ge q^{-b}\) for \(t\le c\) and \(1-q^{-\epsilon+b^2}\ge q^{-1}\) under the numerical assumption.
\end{proof}

\begin{lemma}[Folded quotient-minor identity]\label{rsd:lem:folded-minor}
Let \(U_0<U\le B^*\), let \(\rho:U\to\overline U=U/U_0\) be the quotient map, and let \(\overline B\) be a coefficient space of dimension \(\dim\overline U\) with a fixed identification \(\overline B^*\cong\overline U\).  There is a linear map \(\beta:\overline B\to B\) with \(\beta^*|_U=\rho\).  Define the quotient folded profile
\[
   \overline{\mathfrak V}_i
      =\{f\circ\beta:\ f\in\mathfrak V_i\cap\Hom_U(B,W)\}
      \le\Hom(\overline B,W).
\]
Then:
\begin{enumerate}[label=\textup{(\roman*)}]
\item If \(M\) is a folded realization of \((\mathfrak V,U,T)\) in \(S\subseteq N\), then \(\overline M=M\circ\beta\) is a folded realization of \((\overline{\mathfrak V},\overline U,\rho(T))\) in the same nested pair, and \(\dim\rho(T)=\dim T-\dim(T\cap U_0)\).
\item \(a_{\overline{\mathfrak V}}(\overline U)=a_{\mathfrak V}(U)-a_{\mathfrak V}(U_0)\).
\item \(\Phi^W_{m_0,m_1}(\overline{\mathfrak V},\overline U,\rho(T))
        =\Phi^W_{m_0,m_1}(\mathfrak V,U,T)
        -\Phi^W_{m_0,m_1}(\mathfrak V,U_0,T\cap U_0)\).
\end{enumerate}
\end{lemma}

\begin{proof}
Existence of \(\beta\): extend \(\rho\), viewed through the identification \(\overline U\cong\overline B^*\) as a map \(U\to\overline B^*\), to a linear map \(B^*\to\overline B^*\), and let \(\beta\) be its dual.

Throughout we use the standard annihilator identity: for a linear map \(\beta:\overline B\to B\) and a subspace \(K\le B\),
\[
   (\beta^{-1}K)^\circ=\beta^*(K^\circ)\le\overline B^* .
\]
(The inclusion \(\supseteq\) is direct; equality follows from the dimension count
\(\dim(\beta^{-1}K)^\circ
   =\dim\im\beta-\dim(K\cap\im\beta)
   =\dim\beta^*(K^\circ)\).)

(i)  Since \(\row(M)=U\), we have \(\ker M=U^\circ\subseteq\ker M_i\), so \(M_i\in\mathfrak V_i\cap\Hom_U(B,W)\) and hence \(\overline M_i=M_i\circ\beta\in\overline{\mathfrak V}_i\) for every \(i\).  The image of \(\overline M\) is contained in \(M(B)\subseteq N\).  For the row space, \(\ker\overline M=\beta^{-1}(\ker M)=\beta^{-1}(U^\circ)\), so by the annihilator identity \(\row(\overline M)=\beta^*(U)=\rho(U)=\overline U\).  For the logical row space, \(K_S(\overline M)=\beta^{-1}(K_S(M))=\beta^{-1}(T^\circ)\), hence \(\Lambda_S(\overline M)=\beta^*(T)=\rho(T)\), using \(T\le U\) and \(\beta^*|_U=\rho\).  The kernel of \(\rho|_T\) is \(T\cap U_0\), which gives the dimension formula.

(ii)  Fix \(i\) and consider the surjective linear map \(f\mapsto f\circ\beta\) from \(\mathfrak V_i\cap\Hom_U(B,W)\) onto \(\overline{\mathfrak V}_i\).  Its kernel consists of those \(f\) that vanish on \(U^\circ+\im\beta\).  Now
\[
   (U^\circ+\im\beta)^\circ
      =U\cap(\im\beta)^\circ
      =U\cap\ker\beta^*
      =\ker(\beta^*|_U)
      =\ker\rho
      =U_0,
\]
so \(U^\circ+\im\beta=U_0^\circ\), and the kernel is exactly \(\mathfrak V_i\cap\Hom_{U_0}(B,W)\).  Rank--nullity gives
\[
   \dim\overline{\mathfrak V}_i
      =\dim\bigl(\mathfrak V_i\cap\Hom_U(B,W)\bigr)
      -\dim\bigl(\mathfrak V_i\cap\Hom_{U_0}(B,W)\bigr).
\]
Since \(\overline U=\overline B^*\) is the full dual space, \(\Hom_{\overline U}(\overline B,W)=\Hom(\overline B,W)\), so \(a_{\overline{\mathfrak V}}(\overline U)=\sum_i\dim\overline{\mathfrak V}_i\), and summing the display over \(i\) proves (ii).

(iii)  By (i), \(\dim\overline U=u-\dim U_0\) and \(\dim\rho(T)=t-\dim(T\cap U_0)\), so the second expression for \(\kappa^W\) in Definition~\ref{rsd:def:folded-ensembles} gives
\[
   \kappa^W_{m_0,m_1}(\overline U,\rho(T))
      =\kappa^W_{m_0,m_1}(U,T)-\kappa^W_{m_0,m_1}(U_0,T\cap U_0).
\]
Combining with (ii) proves (iii).
\end{proof}

\begin{proposition}[Folded subcritical direction]\label{rsd:prop:folded-subcritical}
If \(U\neq0\), \(T\le U\), and \(\Gamma^W_{m_0,m_1}(\mathfrak V,U,T)\le-\epsilon\), then, for \((S,N)\sim\PCSS^W(n,q;m_0,m_1)\),
\[
   \Prb\bigl[(S,N)\text{ admits a folded realization of }(\mathfrak V,U,T)\bigr]
      \le q^{-\epsilon}.
\]
\end{proposition}

\begin{proof}
Choose \(U_0<U\) with \(\Phi^W(\mathfrak V,U,T)-\Phi^W(\mathfrak V,U_0,T\cap U_0)\le-\epsilon\).  By Lemma~\ref{rsd:lem:folded-minor}(i), any realization of \((\mathfrak V,U,T)\) yields a realization of \((\overline{\mathfrak V},\overline U,\rho(T))\) in the same pair.  By Markov's inequality, the first-moment upper bound of Corollary~\ref{rsd:cor:folded-one-flag} applied to the quotient data, and Lemma~\ref{rsd:lem:folded-minor}(iii), this quotient realization exists with probability at most
\[
   q^{\Phi^W(\overline{\mathfrak V},\overline U,\rho(T))}
      =q^{\Phi^W(\mathfrak V,U,T)-\Phi^W(\mathfrak V,U_0,T\cap U_0)}
      \le q^{-\epsilon}. \qedhere
\]
\end{proof}

For the second moment, identify \(\Hom(B,W)\oplus\Hom(B,W)\) with \(\Hom(B\oplus B,W)\) via \((f,g)(x,y)=f(x)+g(y)\), and identify \((B\oplus B)^*\) with \(B^*\oplus B^*\).  For \(\Omega\le B^*\oplus B^*\) and a joint flag \(\Theta\le\Omega\), define the folded joint mass and joint potential
\[
   a^{(2)}_{\mathfrak V}(\Omega)
      =\sum_{i=1}^n\dim\bigl((\mathfrak V_i\oplus\mathfrak V_i)\cap\Hom_\Omega(B\oplus B,W)\bigr),
\]
\[
   \Phi^{(2),W}(\mathfrak V,\Omega,\Theta)
      =a^{(2)}_{\mathfrak V}(\Omega)-a\dim\Omega-c(\dim\Omega-\dim\Theta).
\]
This is the potential~\eqref{eq:folded-qlcl-potential} for the doubled coefficient space \(B\oplus B\).

\begin{lemma}[Folded pair submodularity]\label{rsd:lem:folded-pair-submodularity}
Let \(\Omega\le U\oplus U\) and suppose the first projection \(\pi_1:\Omega\to U\) is onto.  Identify \(U_0=\ker(\pi_1|_\Omega)\) with a subspace of the second copy of \(U\).  Let \(\Theta\le\Omega\) satisfy \(\pi_1(\Theta)=T\), and put \(T'=\ker(\pi_1|_\Theta)\), so that \(T'\le U_0\) because \(\Theta\le\Omega\).  Then
\[
   \Phi^{(2),W}(\mathfrak V,\Omega,\Theta)
      \le\Phi^W_{m_0,m_1}(\mathfrak V,U,T)
      +\Phi^W_{m_0,m_1}(\mathfrak V,U_0,T').
\]
\end{lemma}

\begin{proof}
We use the standard duality for a subspace \(L\le B\oplus B\): the annihilator of \(\pi_2(L)\) inside \(B^*\) equals \(\{\psi\in B^*:(0,\psi)\in L^\circ\}\), and dually \(\pi_1(L^\circ)=\bigl(L\cap(B\oplus0)\bigr)^\circ\), the annihilator taken inside the first factor.

Fix a coordinate \(i\) and let \(h=(f,g)\in(\mathfrak V_i\oplus\mathfrak V_i)\cap\Hom_\Omega(B\oplus B,W)\), so that \(\Omega^\circ\subseteq\ker h\) inside \(B\oplus B\).

First, \(f\in\mathfrak V_i\cap\Hom_U(B,W)\): for \(x\in U^\circ\) and every \((\phi,\psi)\in\Omega\le U\oplus U\) one has \(\phi(x)+\psi(0)=0\), so \((x,0)\in\Omega^\circ\) and hence \(f(x)=h(x,0)=0\).  Thus \(h\mapsto f\) is a linear map into \(\mathfrak V_i\cap\Hom_U(B,W)\).

Second, the kernel of this map consists of pairs \((0,g)\) with \(g\) vanishing on \(\pi_2(\Omega^\circ)\).  By the duality above applied to \(L=\Omega^\circ\), the annihilator of \(\pi_2(\Omega^\circ)\) is \(\{\psi:(0,\psi)\in\Omega\}=\ker(\pi_1|_\Omega)=U_0\), so \(\pi_2(\Omega^\circ)=U_0^\circ\) and \(g\in\mathfrak V_i\cap\Hom_{U_0}(B,W)\).  Rank--nullity gives
\[
   \dim\bigl((\mathfrak V_i\oplus\mathfrak V_i)\cap\Hom_\Omega(B\oplus B,W)\bigr)
      \le\dim\bigl(\mathfrak V_i\cap\Hom_U(B,W)\bigr)
      +\dim\bigl(\mathfrak V_i\cap\Hom_{U_0}(B,W)\bigr).
\]
Summing over \(i\) yields \(a^{(2)}_{\mathfrak V}(\Omega)\le a_{\mathfrak V}(U)+a_{\mathfrak V}(U_0)\).  Since \(\pi_1|_\Omega\) is onto with kernel \(U_0\) and \(\pi_1(\Theta)=T\) with kernel \(T'\),
\[
   \dim\Omega=\dim U+\dim U_0,
   \qquad
   \dim\Theta=\dim T+\dim T',
\]
and substituting into the joint potential proves the claim.
\end{proof}

\begin{lemma}[Folded dependent-pair bound]\label{rsd:lem:folded-dependent-pairs}
Assume \(U\neq0\), \(T\le U\), \(t\le c\), and \(\Gamma^W_{m_0,m_1}(\mathfrak V,U,T)\ge\epsilon\).  Let \(Y\) be the contribution to \(\Exp\bigl[(X^W_{\mathfrak V,U,T})^2\bigr]\) from ordered pairs \((M,M')\) for which the joint map \(J:B\oplus B\to W^n\), \(J(x,y)=M(x)+M'(y)\), has \(row(J)\neq U\oplus U\).  Then
\[
   Y\le q^{2\Phi^W_{m_0,m_1}(\mathfrak V,U,T)-\epsilon+8b^2}.
\]
\end{lemma}

\begin{proof}
If \(\row(M)=\row(M')=U\), then \(\ker J\supseteq\ker M\oplus\ker M'\) gives \(\Omega\defeq row(J)\le U\oplus U\), and
\(\pi_1(\Omega)^\circ=\{x:(x,0)\in\ker J\}=\ker M=U^\circ\),
so \(\pi_1(\Omega)=U\); similarly \(\pi_2(\Omega)=U\).  Fix such a joint row space \(\Omega\neq U\oplus U\).  The coordinates of \(J\) are \(J_i=(M_i,M'_i)\in(\mathfrak V_i\oplus\mathfrak V_i)\cap\Hom_\Omega(B\oplus B,W)\), so the number of ordered pairs with exact joint row space \(\Omega\) is at most \(q^{a^{(2)}_{\mathfrak V}(\Omega)}\).

If both \(M\) and \(M'\) are counted by \(X^W_{\mathfrak V,U,T}\), then the joint logical row space \(\Theta=\Lambda_S(J)\le\Omega\) has both projections equal to \(T\): by the duality of Lemma~\ref{rsd:lem:folded-pair-submodularity} applied to \(L=K_S(J)\), and since \(K_S(J)\cap(B\oplus0)=K_S(M)\oplus0\),
\[
   \pi_1(\Theta)
      =\pi_1\bigl(K_S(J)^\circ\bigr)
      =\bigl(K_S(J)\cap(B\oplus0)\bigr)^\circ
      =K_S(M)^\circ=T,
\]
and similarly \(\pi_2(\Theta)=T\).  For fixed \((\Omega,\Theta)\), Lemma~\ref{rsd:lem:folded-fixed-map} applied to \(J\) bounds the probability of the corresponding joint event by \(q^{-a\dim\Omega-c(\dim\Omega-\dim\Theta)}\), dropping the injectivity factor.  The number of choices of \(\Omega\) and of \(\Theta\) in \(B^*\oplus B^*\) is at most \(q^{4b^2}\) each.

Apply Lemma~\ref{rsd:lem:folded-pair-submodularity}.  Since \(\Omega\neq U\oplus U\) and \(\pi_1|_\Omega\) is onto, \(U_0=\ker(\pi_1|_\Omega)\) is a proper subspace of \(U\).  Moreover \(T'=\ker(\pi_1|_\Theta)\le T\cap U_0\), because the second projection of \(\Theta\) is \(T\).  The potential is monotone nondecreasing in the logical row space for fixed ordinary row space, since its coefficient on \(\dim T\) is \(c\ge0\); hence
\[
   \Phi^W(\mathfrak V,U_0,T')
      \le\Phi^W(\mathfrak V,U_0,T\cap U_0)
      \le\Phi^W(\mathfrak V,U,T)-\epsilon,
\]
using the gap assumption in the last step.  Therefore every joint type contributes at most
\[
   q^{\Phi^{(2),W}(\mathfrak V,\Omega,\Theta)}
      \le q^{2\Phi^W(\mathfrak V,U,T)-\epsilon},
\]
and multiplying by the at most \(q^{8b^2}\) joint types proves the claim.
\end{proof}

\begin{theorem}[Folded fixed-type threshold]\label{rsd:thm:folded-fixed-type}
Let \((S,N)\sim\PCSS^W(n,q;m_0,m_1)\).  Fix a folded local profile \(\mathfrak V\) and a flag \((U,T)\) with \(U\neq0\), \(T\le U\le B^*\).
\begin{enumerate}[label=(\alph*)]
\item If \(\Gamma^W_{m_0,m_1}(\mathfrak V,U,T)\le-\epsilon\), then
\[
   \Prb\bigl[(S,N)\text{ admits a folded realization of }(\mathfrak V,U,T)\bigr]\le q^{-\epsilon}.
\]
\item If \(\dim T\le m_1-m_0\), \(\Gamma^W_{m_0,m_1}(\mathfrak V,U,T)\ge\epsilon\), and \(\epsilon\ge2b^2+1\), then
\[
   \Prb\bigl[(S,N)\text{ admits a folded realization of }(\mathfrak V,U,T)\bigr]
      \ge1-q^{-\epsilon+12b^2}.
\]
\end{enumerate}
In particular, Theorem~\ref{qlcl:thm:fixed-type} holds verbatim in the folded setting, with \(\Phi\) and \(\Gamma\) replaced by \(\Phi^W\) and \(\Gamma^W\), scalar witness matrices replaced by folded realizations \(M:B\to N\), and the same constants.
\end{theorem}

\begin{proof}
Part (a) is Proposition~\ref{rsd:prop:folded-subcritical}.

For (b), let \(X=X^W_{\mathfrak V,U,T}\).  Since \(0<U\) is one of the quotient minors, the gap assumption gives \(\Phi^W(\mathfrak V,U,T)\ge\epsilon\), and Corollary~\ref{rsd:cor:folded-one-flag} gives \(\Exp X\ge q^{\Phi^W(\mathfrak V,U,T)-2b^2}\).

For an ordered pair \((M,M')\) with \(row(J)=U\oplus U\), we have \(\ker J=U^\circ\oplus U^\circ\), so \(\dim J(B\oplus B)=2u\) and hence \(M(B)\cap M'(B)=0\): the two images are independent subspaces of \(\F_q^{ns}\).  The restrictions of the independent uniform rows of \(H_C\) and \(H_S\) to independent subspaces are independent, so the two indicator variables are independent and their covariance is zero.  For all remaining ordered pairs the covariance is at most the joint probability, and Lemma~\ref{rsd:lem:folded-dependent-pairs} bounds the sum of these joint probabilities.  Therefore
\[
   \operatorname{Var}(X)\le q^{2\Phi^W(\mathfrak V,U,T)-\epsilon+8b^2},
\]
and Chebyshev's inequality yields
\[
   \Prb[X=0]\le\frac{\operatorname{Var}(X)}{(\Exp X)^2}\le q^{-\epsilon+12b^2}. \qedhere
\]
\end{proof}

\begin{corollary}[Folded threshold theorems and the uniform ensemble]\label{rsd:cor:folded-uniform}
Define folded qLCL properties by finite families of folded profiles and sets of allowed logical row spaces in \(\cL(B^*)\), exactly as in Definition~\ref{qlcl:def:qlcl-property}.  Then Theorem~\ref{qlcl:thm:qlcl-threshold} holds verbatim for \(\PCSS^W(n,q;m_0,m_1)\), and Theorem~\ref{qlcl:thm:uniform-threshold} holds verbatim for \(\Nest^W(n,q;m_0,m_1)\), with \(\Phi,\Gamma\) replaced by \(\Phi^W,\Gamma^W\), with the same constants, and with the same \(\gamma_q^{-1}\) conditioning loss.
\end{corollary}

\begin{proof}
The union bounds over profiles and over the at most \(q^{2b^2}\) flags in \(B^*\) are unchanged, so the folded analogue of Theorem~\ref{qlcl:thm:qlcl-threshold} follows from Theorem~\ref{rsd:thm:folded-fixed-type} exactly as in the scalar case.  For the uniform ensemble, Lemma~\ref{qlcl:lem:conditioning} concerns only the parity checks on the scalar unfolding \(\F_q^{ns}\) and applies verbatim: conditioned on full rank, \(\PCSS^W(n,q;m_0,m_1)\) is \(\Nest^W(n,q;m_0,m_1)\), and the full-rank event has probability at least \(\gamma_q\).  Finally, the feasibility restriction in the analogue of~\eqref{qlcl:eq:Gmax-unif} is justified by the folded analogue of Lemma~\ref{qlcl:lem:auto-feasibility}, whose proof applies verbatim to a folded realization \(M:B\to N\) in a fixed nested pair \(S\subseteq N\): the image of \(K_S(M)=T^\circ\) under \(M\) is \(M(B)\cap S\), with kernel \(\ker M=U^\circ\), so \(\dim\bigl(M(B)\cap S\bigr)=u-t\le\dim S\), and the induced image of \(M(B)\) in \(N/S\) has dimension \(t\le\dim(N/S)\).
\end{proof}

\begin{remark}[Consistency with the scalar theory]\label{rsd:rem:folded-consistency}
For \(s=1\) the fold is trivial: \(\Hom(B,\F_q)=B^*\), a folded profile is a scalar profile after the identification \(B\cong\F_q^b\), \(\Hom_U(B,\F_q)=U\), the folded mass~\eqref{eq:folded-mass} is the scalar mass~\eqref{qlcl:eq:profile-mass}, and Lemmas~\ref{rsd:lem:folded-exact-count}--\ref{rsd:lem:folded-dependent-pairs} together with Theorem~\ref{rsd:thm:folded-fixed-type} specialize to Lemmas~\ref{qlcl:lem:exact-row-count}, \ref{qlcl:lem:fixed-matrix-prob}, \ref{qlcl:lem:minor-potential}, \ref{qlcl:lem:pair-submodularity}, \ref{qlcl:lem:dependent-pairs}, and Theorem~\ref{qlcl:thm:fixed-type}.  The folded statements are therefore a strict generalization proved by the same argument, not a reduction to the scalar case.
\end{remark}

\subsection{Joint Sector Formulation}
In this short subsection, we state how folded (random) CSS codes interact with local profiles in the joint sector case. Instead of stating and proving the variants of the statements from Subsection~\ref{sub:joint-sec}, we will instead describe the changes required to make the proofs go through.
In place of working over subspaces of $\F_q^b$, we will be working over subspaces of $B^*$, where $B$ is a finite dimensional coefficient space, as described above. Naturally, local profiles are replaced with folded local profiles.
Thus, we can define generalizations of $\lambda_s(\cV, U ,T)$ and $R^{\rm flag}_{\rm s}(\cV; U, T)$ in the natural manner.
Following \cite{GGH26}'s formulation of the classical LCL framework, the classical rate threshold ($R_{\rm cl}(\cV)$) also can be defined in terms of subspaces of $B^*$ and folded local profiles.

Next, the definition of joint realization (Definition~\ref{def:join-realize}) is defined using the generalizations described above, and the same holds for $R^{\rm joint}(\cF)$. We only state the version of Corollary~\ref{cor:joint-form}.

Let $\cF$ denote a family of folded local profiles. For a random folded CCS code $C$, let $E(C, \cF)$ denote the event that $C$ realizes $(\mathfrak V, U_X+U_Z, T_X+T_Z)$ for some $\mathfrak V \in \cF$, $T_X \le U_X, T_Z \le U_Z$ satisfying $T_X+T_Z \in \LDist_b$.
\begin{corollary}[Probability of Joint Sector Satisfaction]\label{cor:joint-form-folded}
For the event $E$ defined above, we have:
\begin{equation}\label{eq:prob-join-sat-folded}
    Pr_C[E(C, \cF)] \le |\cF|\cdot q^{4b^2} \cdot q^{-\delta sn}.
\end{equation}
\end{corollary}
\begin{proof}
    The corollary follows by union bounding over all $\mathfrak V \in \cF$, $U_X,U_Z,T_X,T_Z$ satisfying the appropriate conditions, and then applying the folded version of Proposition~\ref{prop:fixed-type}.
\end{proof}

Regarding the balanced case, we will require the same condition as before: $\max\{r_X,r_Z\} +\delta \le R_{\rm cl}(\cF)$. This
reduces to
\[
    (1+R_Q)/2 +\delta \le R_{\rm cl}(\cF).
\]

\subsection{Kernel profiles}

We now design local profiles that when avoided, produce subspace design codes. A kernel profile prescribes kernels rather than arbitrary coordinate-map subspaces.  Let \(V\) be a finite-dimensional space and let \(V_i\le V\).  Put \(B=V^*\).  Define the associated kernel-profile LCL constraint by
\begin{equation}\label{eq:kernel-profile-lcl-constraint}
      \cZ_i(V_i)
      =\{f\in\Hom(V^*,W): V_i^\circ\subseteq\ker f\}
      \cong\Hom(V_i^*,W).
\end{equation}
Thus a coordinate map is allowed exactly when it vanishes on the coefficient directions annihilated by \(V_i\).

\begin{lemma}[Kernel-profile mass identity]\label{rsd:lem:kernel-mass}
For every \(U\le V\),
\[
      \dim\bigl(\cZ_i(V_i)\cap\Hom_U(V^*,W)\bigr)
      =s\dim(U\cap V_i),
\]
and hence
\[
      a_{\cZ(V)}(U)=s\sum_{i=1}^n\dim(U\cap V_i).
\]
\end{lemma}

\begin{proof}
A map \(f:V^*\to W\) lies in \(\cZ_i(V_i)\cap\Hom_U(V^*,W)\) precisely when it vanishes on both \(V_i^\circ\) and \(U^\circ\), hence on \(V_i^\circ+U^\circ\).  Since
\[
      (V_i^\circ+U^\circ)^\circ=V_i\cap U,
\]
the quotient \(V^*/(V_i^\circ+U^\circ)\) has dimension \(\dim(V_i\cap U)\).  The space of maps from this quotient to \(W\) has dimension \(s\dim(V_i\cap U)\).
\end{proof}

For \(\alpha\in\R\), define the normalized subspace-profile potential
\begin{equation}\label{eq:subspace-profile-potential}
   \Phi_V(U,(V_i)_{i=1}^n,\alpha)
      \defeq
      \alpha\dim U-
      \frac1n\sum_{i=1}^n
        \bigl(\dim U-
              \dim(U\cap V_i)\bigr).
\end{equation}
Equivalently,
\[
   \Phi_V(U,(V_i),\alpha)
      =\frac1n\sum_i\dim(U\cap V_i)-(1-\alpha)\dim U.
\]
Let \(\rho_N\defeq m_1/(ns)\) be the actual rate of the sampled normalizer \(N\).  For kernel profiles and injective-logical flags \(T=U\), equations~\eqref{eq:folded-qlcl-potential} and~\eqref{eq:subspace-profile-potential} give the exact potential dictionary
\begin{equation}\label{eq:potential-dictionary}
      \frac1{sn}\Phi^{W}_{m_0,m_1}(\cZ(V),U,U)
      =\Phi_V(U,(V_i),\rho_N).
\end{equation}
For an arbitrary design threshold \(\alpha\), the deterministic design potential is related to the qLCL potential at the actual rate by
\begin{equation}\label{eq:potential-slack}
      \Phi_V(U,(V_i),\alpha)
      =\frac1{sn}\Phi^{W}_{m_0,m_1}(\cZ(V),U,U)
        +(\alpha-\rho_N)\dim U.
\end{equation}
Thus the subspace-design potential is exactly the normalized qLCL potential only at the actual normalizer rate \(\rho_N\).  If \(\alpha>\rho_N\), then a negative \(\alpha\)-design potential gives negative qLCL potential at rate \(\rho_N\) with slack \((\alpha-\rho_N)\dim U\).  Quotient minors \(U' < U\) are exactly the tests of the same profile on all nonzero subspaces \(U'\le V\).

\begin{remark}[Classical specializations]\label{rsd:rem:classical-special-cases}
The framework specializes cleanly to both classical theories.
\begin{enumerate}[leftmargin=2em,label=\textup{(\roman*)}]
\item \emph{Classical scalar LCL.}  Set $S=0$ and take an ordinary linear code $C\le\F_q^n$ of dimension $m$.  Then $K_S(A)=\ker(A)$ for every witness matrix $A$, so the logical row space is automatically the ordinary row space: $T=U$.  For scalar local profiles one obtains
\[
   \Phi_{0,m}(\cV,U,U)=a_\cV(U)-(n-m)\dim U,
\]
and for every proper $W<U$,
\[
   \Phi_{0,m}(\cV,U,U)-\Phi_{0,m}(\cV,W,W)
   =a_\cV(U)-a_\cV(W)-(n-m)(\dim U-\dim W).
\]
This is the LCL potential increment for random linear codes in the sense of Levi--Mosheiff--Shagrithaya~\cite[Section~4]{LMS25}.  Thus the quotient threshold theorem reduces, after the harmless relabeling $T=U$, to the classical fixed-width LCL threshold.

\item \emph{Ordinary folded subspace designs.}  Again set $S=0$, but allow a folded alphabet $W$ and restrict local constraints to kernel profiles $\cZ_i(V_i)$.  The injective-logical condition is automatic for every nonzero represented subspace, because no nonzero vector is a stabilizer.  If $\dim N=\alpha ns$, then
\[
   \frac{1}{sn}\Phi^W_{0,\alpha ns}(\cZ(V),U,U)
   =\frac1n\sum_i\dim(U\cap V_i)-(1-\alpha)\dim U
   =\Phi_V(U,(V_i),\alpha).
\]
The relative CSS definition below is obtained by keeping the same folded kernel-profile tests but allowing $S\ne0$ and explicitly requiring the tested representative subspace to intersect $S$ trivially.
\end{enumerate}
\end{remark}

\section{Relative subspace designs for CSS codes}\label{sec:relative-sd}

The dictionary identifies the qLCL witnesses that should be forbidden.  We now state the resulting property directly for nested folded spaces.  The definition has one quotient feature, namely $A\cap S=0$: pure stabilizer subspaces are discarded, but coordinate zeros are still measured on the physical representatives in $A$.  This distinction is the main conceptual difference between ordinary subspace designs and their CSS relative version.

\subsection{Definitions}

\begin{definition}[Relative subspace design]\label{rsd:def:relative}
Let \(S\subseteq N\subseteq W^n\), let \(r\in\N\), and let \(\alpha\in[0,1]\).  The pair \((N,S)\) is an \emph{\(\alpha\)-relative subspace design up to dimension \(r\)} if for every subspace \(A\le N\) satisfying
\[
  1\le \dim A\le r,
  \qquad A\cap S=\{0\},
\]
one has
\begin{equation}\label{rsd:eq:relative-design}
   \frac1n\sum_{i=1}^n \dim(A\cap K_i)
      \le \alpha\dim A .
\end{equation}
Equivalently,
\[
   \frac1n\sum_{i=1}^n \dim\pi_i(A)
      \ge (1-\alpha)\dim A .
\]
\end{definition}

The condition \(A\cap S=0\) says that \(A\) injects into the logical quotient \(N/S\).  Thus the test ignores pure stabilizer directions but measures zeros in the actual physical representatives.

\begin{definition}[CSS subspace-design code]\label{rsd:def:css-sd}
Let \(Q=Q(C_1,C_2)\) be the CSS code defined by \(C_2\subseteq C_1\subseteq W^n\).  We say that \(Q\) is an \emph{\((\alpha_X,\alpha_Z)\)-CSS subspace design up to dimension \(r\)} if
\[
   (C_1,C_2)\quad\text{is an }\alpha_X\text{-relative subspace design up to dimension }r,
\]
and
\[
   (C_2^\perp,C_1^\perp)\quad\text{is an }\alpha_Z\text{-relative subspace design up to dimension }r.
\]
\end{definition}

One may also use a dimension-dependent function \(\tau(t)\) by replacing \(\alpha\dim A\) with \(\tau(\dim A)\dim A\).  The AEL lifting theorem in this paper is stated for a uniform \(\alpha\), which is the parameter regime needed for near-optimal designs.

\begin{remark}[Why the relative condition is not an ordinary design on the quotient]\label{rsd:rem:not-quotient-design}
The quotient $N/S$ is the logical space, but its coordinates are not obtained by quotienting each folded coordinate independently.  A logical coset can have many physical representatives with different zero patterns.  Therefore the design condition must test an actual representative subspace $A\le N$ and only then require $A\cap S=0$.  This is the same order of operations used throughout quotient LCL: local constraints are imposed before quotienting, while independence is measured after quotienting.
\end{remark}

\subsection{Basic properties}

\begin{proposition}[Section formulation]\label{rsd:prop:section}
Let \(q_N:N\to N/S\) be the quotient map.  The pair \((N,S)\) is an \(\alpha\)-relative subspace design up to dimension \(r\) if and only if for every \(B\le N/S\) with \(1\le \dim B\le r\) and every linear section \(\sigma:B\to N\) of \(q_N\),
\[
   \frac1n\sum_i \dim(\sigma(B)\cap K_i)
      \le \alpha\dim B .
\]
\end{proposition}

\begin{proof}
If \(A\le N\) and \(A\cap S=0\), then \(q_N|_A:A\to q_N(A)\) is an isomorphism, so \(A\) is the image of the section \((q_N|_A)^{-1}\).  Conversely, the image of any section intersects \(S\) trivially.
\end{proof}

\begin{proposition}[Inheritance from ordinary designs]\label{rsd:prop:inherit}
If \(N\le W^n\) is an ordinary \(\alpha\)-subspace-design code up to dimension \(r\), then \((N,S)\) is an \(\alpha\)-relative subspace design up to dimension \(r\) for every \(S\subseteq N\).
\end{proposition}

\begin{proof}
The ordinary condition is imposed on every low-dimensional \(A\le N\).  The relative condition is imposed only on those \(A\) satisfying \(A\cap S=0\).
\end{proof}

\begin{proposition}[Monotonicity]\label{rsd:prop:monotone}
Let \(S\subseteq N\).  If \((N,S)\) is an \(\alpha\)-relative design up to dimension \(r\), then \((N',S\cap N')\) is an \(\alpha\)-relative design up to dimension \(r\) for every \(N'\subseteq N\).  Also, if \(S\subseteq S'\subseteq N\), then \((N,S')\) is an \(\alpha\)-relative design up to dimension \(r\).
\end{proposition}

\begin{proof}
Both statements reduce the collection of subspaces \(A\) that must be tested.
\end{proof}

\begin{proposition}[Logical distance]\label{rsd:prop:distance}
If \((N,S)\) is an \(\alpha\)-relative subspace design up to dimension \(1\), then every \(x\in N\setminus S\) satisfies
\[
   \frac{\wt(x)}n\ge 1-\alpha .
\]
Consequently, if \(Q(C_1,C_2)\) is an \((\alpha_X,\alpha_Z)\)-CSS subspace design up to dimension \(1\), then
\[
   \frac{d_X}{n}\ge 1-\alpha_X,
   \qquad
   \frac{d_Z}{n}\ge 1-\alpha_Z.
\]
\end{proposition}

\begin{proof}
Let \(x\in N\setminus S\) and set \(A=\Span\{x\}\).  Then \(A\cap S=0\).  Moreover \(\dim(A\cap K_i)=1\) exactly when \(x_i=0\), and is zero otherwise.  Thus
\[
  1-\frac{\wt(x)}n
   =\frac1n\sum_i\dim(A\cap K_i)
   \le \alpha.
\]
Apply this to \((C_1,C_2)\) and \((C_2^\perp,C_1^\perp)\).
\end{proof}

\begin{remark}[Quantum Singleton scale]
For a balanced CSS code of quantum rate \(R\), the two normalizer rates are naturally around \((1+R)/2\).  Thus a near-optimal CSS subspace design has \(\alpha_X,\alpha_Z\approx(1+R)/2\), and Proposition~\ref{rsd:prop:distance} yields relative distance close to \((1-R)/2\), the quantum Singleton scale.  Folded CSS codes with Singleton-scale relative distance are already known via quantum AEL distance amplification~\cite{BGG24,BJMST24}; the design condition strengthens the distance bound from dimension one to all dimensions up to \(r\).
\end{remark}

\subsection{Contained profiles and deterministic equivalence}

\begin{definition}[Contained relative profile]\label{rsd:def:profile}
Let \(S\subseteq N\subseteq W^n\).  A tuple \((V_1,\ldots,V_n)\in\calL(V)^n\) is a \emph{relative \(V\)-profile contained in \((N,S)\)} if there are a subspace \(A\le N\) and an isomorphism \(\phi:V\to A^*\) such that
\[
      A\cap S=0
\]
and, for every coordinate \(i\),
\begin{equation}\label{rsd:eq:profile-containment}
      \phi(V_i)^\circ\subseteq A\cap K_i .
\end{equation}
\end{definition}

The containment condition says that the coordinate row space of \(A\to W\) is contained in \(\phi(V_i)\).  This is exactly a kernel-profile local constraint.

\begin{proposition}[Kernel-profile quotient LCL equals contained profiles]\label{rsd:prop:lcl-contained-profile}
Let \(S\subseteq N\subseteq W^n\), let \(V\) be finite-dimensional, and put \(B=V^*\).  A tuple \((V_i)_{i=1}^n\in\calL(V)^n\) is a relative \(V\)-profile contained in \((N,S)\) if and only if \((S,N)\) has a folded qLCL realization \(M:B\to N\) with local constraints \(\cZ_i(V_i)\), ordinary row space \(R(M)=V\), and logical row space \(\Lambda_S(M)=V\), after identifying \(B^*=V\).
\end{proposition}

\begin{proof}
Suppose first that the profile is contained, witnessed by \(A\le N\) and \(\phi:V\to A^*\).  Dualizing \(\phi\) gives an isomorphism \(\phi^*:A\to V^*=B\), hence its inverse gives a map \(M:B\to A\subseteq N\).  Because \(M\) is injective, \(R(M)=V\).  Because \(A\cap S=0\), we have \(K_S(M)=\ker M=0\), hence \(\Lambda_S(M)=V\).  For coordinate \(i\), the condition \(\phi(V_i)^\circ\subseteq A\cap K_i\) is equivalent, after transporting through \(M\), to \(V_i^\circ\subseteq\ker(\pi_iM)\).  Therefore \(M_i\in\cZ_i(V_i)\).

Conversely, let \(M:B\to N\) be such a folded qLCL realization and set \(A=M(B)\).  Since \(R(M)=B^*=V\), the map \(M\) is injective.  Since \(\Lambda_S(M)=R(M)\), we have \(M^{-1}(S)=\ker M=0\), so \(A\cap S=0\).  The inverse map \(M^{-1}:A\to B=V^*\) dualizes to an isomorphism \(\phi:V\to A^*\).  The local constraint \(V_i^\circ\subseteq\ker(\pi_iM)\) says exactly that \(\phi(V_i)^\circ\subseteq A\cap K_i\).  Hence the profile is contained.
\end{proof}

\begin{theorem}[Profile characterization]\label{rsd:thm:profile-characterization}
Let \(S\subseteq N\subseteq W^n\).  The pair \((N,S)\) is an \(\alpha\)-relative subspace design up to dimension \(r\) if and only if for every vector space \(V\) with \(\dim V\le r\) and every relative \(V\)-profile \((V_1,\ldots,V_n)\) contained in \((N,S)\),
\[
      \Phi_V(V,(V_i),\alpha)\ge0.
\]
Equivalently, \((N,S)\) has no injective-logical kernel-profile witness whose associated design profile has negative normalized design potential \(\Phi_V(V,(V_i),\alpha)\).  When \(\alpha=\dim N/(ns)\), this is exactly absence of negative normalized folded qLCL potential; for general \(\alpha\) the two potentials are related by~\eqref{eq:potential-slack}.
\end{theorem}

\begin{proof}
Assume first that \((N,S)\) is a relative design, and let the profile be witnessed by \(A\le N\) and \(\phi:V\to A^*\).  From \(\phi(V_i)^\circ\subseteq A\cap K_i\),
\[
   \dim(A\cap K_i)
      \ge \dim\phi(V_i)^\circ
      =\dim V-\dim V_i.
\]
Since \(A\cap S=0\) and \(\dim A=\dim V\le r\),
\[
  \frac1n\sum_i(\dim V-\dim V_i)
      \le \frac1n\sum_i\dim(A\cap K_i)
      \le \alpha\dim V.
\]
This is exactly \(\Phi_V(V,(V_i),\alpha)\ge0\).

Conversely, suppose the relative-design condition fails.  Then there is \(A\le N\), \(A\cap S=0\), \(1\le t=\dim A\le r\), such that
\[
   \frac1n\sum_i\dim(A\cap K_i)>\alpha t.
\]
Set \(V=A^*\), let \(\phi\) be the identity map \(V\to A^*\), and define
\[
   V_i\defeq(A\cap K_i)^\circ\le A^*=V.
\]
Then \(\phi(V_i)^\circ=A\cap K_i\), so the profile is contained, and
\[
   \Phi_V(V,(V_i),\alpha)
     =\alpha t-\frac1n\sum_i\dim(A\cap K_i)<0.
\]
The final equivalence follows from Proposition~\ref{rsd:prop:lcl-contained-profile} and the definition of the design potential.  If \(\alpha=\dim N/(ns)\), this is the exact normalized qLCL potential by~\eqref{eq:potential-dictionary}; otherwise the comparison is the slack identity~\eqref{eq:potential-slack}.
\end{proof}

\begin{lemma}[Minimal bad profile]\label{rsd:lem:minimal-bad}
If \((N,S)\) is not an \(\alpha\)-relative subspace design up to dimension \(r\), then there exist a vector space \(V\) with \(1\le\dim V\le r\) and a contained relative \(V\)-profile \((V_i)\) such that
\[
      \Phi_V(U,(V_i),\alpha)<0
      \qquad\text{for every nonzero }U\le V.
\]
\end{lemma}

\begin{proof}
By Theorem~\ref{rsd:thm:profile-characterization}, there is a contained profile \((V_i)\) with \(\Phi_V(V,(V_i),\alpha)<0\).  Among all subspaces \(W\le V\) maximizing \(\Phi_V(W,(V_i),\alpha)\), choose one of maximum dimension.  Since \(\Phi_V(0,(V_i),\alpha)=0\) and \(\Phi_V(V,(V_i),\alpha)<0\), this maximizer is a proper subspace of \(V\).

Let \(\bar V=V/W\) and \(\bar V_i=(V_i+W)/W\).  We claim that \((\bar V_i)\) is contained.  Suppose \((V_i)\) is witnessed by \(A\le N\) and \(\phi:V\to A^*\).  Let
\[
      A'\defeq \phi(W)^\circ\le A.
\]
Then \(A'\cap S=0\), and restriction of functionals gives a natural isomorphism \(\bar\phi:\bar V\to(A')^*\).  Moreover
\[
      \bar\phi(\bar V_i)^\circ
      =A'\cap\phi(V_i)^\circ
      \subseteq A'\cap K_i,
\]
so the quotient profile is contained in \((N,S)\).

For any nonzero \(\bar U\le\bar V\), write \(\bar U=(U+W)/W\) with \(U\not\subseteq W\).  A direct dimension calculation gives
\[
      \Phi_{\bar V}(\bar U,(\bar V_i),\alpha)
      =\Phi_V(U+W,(V_i),\alpha)-\Phi_V(W,(V_i),\alpha).
\]
The maximality of \(W\), together with the choice of maximum dimension among maximizers, implies that the right-hand side is negative whenever \(U+W\supsetneq W\).  Hence every nonzero subspace of the quotient has negative potential.
\end{proof}

\section{Robust inner LCL}\label{sec:robust-inner}

The AEL construction examines a left vertex through a linear map from a global
bad witness into the inner normalizer.  The kernel of its logical quotient may
consist of local stabilizers, not literal zero codewords.  Thus the local inner
condition must be robust under maps that have stabilizer kernel.  We formulate
this robustness first for folded qLCL profiles.  Relative
subspace designs are then recovered by restricting the profile class to
kernel profiles.

Let \(S\subseteq N\subseteq W^d\), where \(W\cong\F_q^s\), and let
\(q_N:N\to N/S\) be the quotient map.  For \(\ell\in[d]\), write
\(\pi_\ell:W^d\to W\) for coordinate projection.  If \(B\) is a coefficient
space and \(\mathfrak V=(\mathfrak V_1,\ldots,\mathfrak V_d)\) is a folded
profile with \(\mathfrak V_\ell\le\Hom(B,W)\), set
\begin{equation}\label{rsd:eq:lcl-potential}
   \widehat\Phi^W_\alpha(U,\mathfrak V)
      \defeq
      \frac1{sd}\sum_{\ell=1}^d
        \dim\bigl(\mathfrak V_\ell\cap\Hom_U(B,W)\bigr)
      -(1-\alpha)\dim U,
      \qquad U\le B^* .
\end{equation}
This is the normalized injective-logical folded LCL potential at design
threshold \(\alpha\).  It is the same normalization as
\eqref{eq:potential-dictionary}; the only difference is that the coordinate
constraints are now arbitrary folded LCL subspaces rather than kernel-profile
subspaces.

\begin{definition}[Bounded-entropy folded profile class]\label{rsd:def:bounded-profile-class}
Fix \(r\).  A folded profile class
\(\mathscr F_{\le r}=\{\mathscr F_t\}_{1\le t\le r}\) assigns, for each
\(t\), a collection \(\mathscr F_t\) of length-\(d\) folded profiles on a fixed
\(t\)-dimensional coefficient space.  We say that it has entropy exponent
\(h\) if
\[
       |\mathscr F_t|\le q^{h d}
       \qquad\text{for every }1\le t\le r .
\]
Profiles on another \(t\)-dimensional coefficient space are identified with
members of \(\mathscr F_t\) after choosing a linear isomorphism.  This convention
only changes the estimates by the usual \(q^{O(r^2)}\) number of bases, which is
absorbed in the entropy term below.
\end{definition}

\begin{definition}[AEL-robust folded LCL]\label{rsd:def:ael-robust}
Let \(\mathscr F_{\le r}\) be a folded profile class.  The pair
\((N,S)\) is \emph{AEL-robust \((\alpha,\mathscr F_{\le r})\) folded
LCL up to dimension \(r\)} if the following holds.  For every coefficient
space \(B\) with \(1\le\dim B\le r\), every profile
\(\mathfrak V\in\mathscr F_{\dim B}\), and every linear map
\(T:B\to N\) such that
\[
       \pi_\ell T\in \mathfrak V_\ell
       \qquad(\ell=1,\ldots,d),
\]
either
\[
       q_N(T(B))=0,
\]
or there is a nonzero subspace \(U\le B^*\) such that
\[
       \widehat\Phi^W_\alpha(U,\mathfrak V)\ge0 .
\]
Equivalently, if all nonzero subprofiles have negative potential, then every
realization of that profile inside \(N\) is logically trivial modulo \(S\).
\end{definition}

\begin{definition}[AEL-robust relative subspace design]\label{rsd:def:robust}
The pair \((N,S)\) is an \emph{AEL-robust \(\alpha\)-relative subspace design
up to dimension \(r\)} if the following holds.  For every vector space \(V\)
with \(1\le t=\dim V\le r\), every vector space \(A\) with \(\dim A=t\), every
isomorphism \(\phi:V\to A^*\), every linear map \(T:A\to N\), and every tuple
\((V_1,\ldots,V_d)\in\calL(V)^d\) satisfying
\begin{equation}\label{rsd:eq:robust-containment}
   \phi(V_\ell)^\circ\subseteq \ker(\pi_\ell\circ T)
   \qquad(\ell=1,\ldots,d),
\end{equation}
either
\[
   q_N(T(A))=0,
\]
or there is a nonzero \(U\le V\) such that
\[
   \Phi_V(U,(V_\ell)_{\ell=1}^d,\alpha)\ge0.
\]
Here \(\Phi_V\) is the subspace-profile potential~\eqref{eq:subspace-profile-potential}, with the tuple length \(d\) in place of \(n\) in the normalization.
\end{definition}

\begin{proposition}[Kernel profiles are an LCL specialization]\label{rsd:prop:kernel-specialization}
Let \(\mathscr K_{\le r}\) be the profile class consisting, for each
\(t\le r\), of all kernel profiles
\[
      \bigl(\cZ_1(V_1),\ldots,\cZ_d(V_d)\bigr),
      \qquad V_\ell\le V,
      \qquad \dim V=t,
\]
with \(B=V^*\).  Then \(\mathscr K_{\le r}\) has entropy exponent at most
\(r^2\).  Moreover, for kernel profiles the potential
\eqref{rsd:eq:lcl-potential} agrees with the subspace-profile potential:
\[
      \widehat\Phi^W_\alpha(U,\cZ_1(V_1),\ldots,\cZ_d(V_d))
       = \Phi_V(U,(V_\ell)_{\ell=1}^d,\alpha).
\]
Consequently AEL-robust \((\alpha,\mathscr K_{\le r})\) folded LCL is
exactly AEL-robust \(\alpha\)-relative subspace design in the sense of
Definition~\ref{rsd:def:robust}.
\end{proposition}

\begin{proof}
The number of subspaces of a \(t\)-dimensional vector space is at most
\(q^{t^2}\), so the number of length-\(d\) kernel profiles is at most
\(q^{dt^2}\le q^{dr^2}\).  The potential identity is Lemma~\ref{rsd:lem:kernel-mass},
divided by \(sd\).  Transporting the coefficient space through the isomorphism
\(\phi:V\to A^*\) identifies the containment condition
\eqref{rsd:eq:robust-containment} with the requirement
\(\pi_\ell T\in\cZ_\ell(V_\ell)\).  Under the same identification the nonnegative
subprofile condition is exactly the displayed potential condition.  This proves
the equivalence of the folded-LCL and kernel-profile formulations.
\end{proof}

\begin{proposition}[Robust implies relative]\label{rsd:prop:robust-implies-relative}
If \(S\subseteq N\subseteq W^d\) and \((N,S)\) is AEL-robust
\(\alpha\)-relative up to dimension \(r\), then it is an \(\alpha\)-relative
subspace design up to dimension \(r\).
\end{proposition}

\begin{proof}
If \((N,S)\) were not relative, Lemma~\ref{rsd:lem:minimal-bad} would give a
contained profile \((V_\ell)\) witnessed by \(A\le N\), \(A\cap S=0\), such that
all nonzero subspaces have negative potential.  Apply robustness to the
inclusion map \(T:A\hookrightarrow N\).  Since \(A\cap S=0\), we have
\(q_N(T(A))\ne0\).  Robustness then gives a nonzero subspace of nonnegative
potential, a contradiction.
\end{proof}

\subsection{Random normalizers are robust for LCL}

We now prove the inner-code existence theorem at the LCL level.  The proof
uses only the randomness of the normalizer \(N\); the stabilizer
\(S\subseteq N\) may be chosen adversarially after \(N\) is sampled.  This is
useful for nested CSS pairs, where the two normalizers are \(C_1\) and
\(C_2^\perp\).

We need two elementary counting estimates.  The first is the standard containment
probability for a random subspace.

\begin{lemma}[Random-subspace containment]\label{rsd:lem:containment}
Let \(E\) be an \(m\)-dimensional vector space over \(\F_q\), let \(N\le E\) be
uniformly random of dimension \(\rho m\), and let \(B\le E\) be fixed of
dimension \(b\), where all displayed dimensions are integral.  Then there is an
absolute constant \(\Gamma_*<\infty\) such that
\[
   \Pr[B\subseteq N]
      \le \Gamma_* q^{-(1-\rho)mb}.
\]
In particular, \(\Pr[B\subseteq N]\le q^{-(1-\rho)mb+O(1)}\).
\end{lemma}

\begin{proof}
If \(b>\rho m\), then the probability is zero.  Otherwise the probability is
\[
  \frac{\genfrac{[}{]}{0pt}{}{m-b}{\rho m-b}_q}
       {\genfrac{[}{]}{0pt}{}{m}{\rho m}_q},
\]
where \(\genfrac{[}{]}{0pt}{}{a}{b}_q\) is a Gaussian binomial coefficient.  The
standard estimate
\[
   q^{k(a-k)}\le \genfrac{[}{]}{0pt}{}{a}{k}_q
      \le \Gamma_* q^{k(a-k)},
   \qquad
   \Gamma_*\defeq\prod_{j=1}^{\infty}(1-2^{-j})^{-1},
\]
gives
\[
   \Pr[B\subseteq N]
     \le \Gamma_* q^{(\rho m-b)(m-\rho m)-(\rho m)(m-\rho m)}
     =\Gamma_* q^{-(1-\rho)mb}.
\]
\end{proof}

\begin{lemma}[Map count for bad profiles]\label{rsd:lem:map-count}
Let \(B\) be \(t\)-dimensional, let
\(\mathfrak V=(\mathfrak V_1,\ldots,\mathfrak V_d)\) be a folded profile on
\(B\), and suppose that
\begin{equation}\label{rsd:eq:all-bad-alpha}
   \widehat\Phi^W_\alpha(U,\mathfrak V)<0
   \qquad\text{for every nonzero }U\le B^* .
\end{equation}
For each \(b\in\{1,\ldots,t\}\), the number of linear maps \(T:B\to W^d\) of
rank \(b\) satisfying
\[
   \pi_\ell T\in\mathfrak V_\ell
   \qquad(\ell=1,\ldots,d)
\]
is at most
\[
   q^{t^2+sd(1-\alpha)b}.
\]
\end{lemma}

\begin{proof}
For such a map \(T\), let \(U=R(T)=(\ker T)^\circ\le B^*\).  Then
\(\dim U=b\).  For a fixed row space \(U\), the number of choices for the
\(\ell\)-th coordinate map is at most
\[
      q^{\dim(\mathfrak V_\ell\cap\Hom_U(B,W))}.
\]
Multiplying over \(\ell\) gives at most
\[
      q^{\sum_\ell\dim(\mathfrak V_\ell\cap\Hom_U(B,W))}.
\]
By~\eqref{rsd:eq:all-bad-alpha}, this exponent is less than
\(sd(1-\alpha)b\).  The number of possible \(b\)-dimensional row spaces
\(U\le B^*\) is at most \(q^{t^2}\).  This proves the claim.
\end{proof}

\begin{theorem}[Random normalizers are LCL robust]\label{rsd:thm:random-robust-normalizer}
There is an absolute constant \(C\) such that the following holds.  Fix
\(q\), \(r\in\N\), \(\rho\in(0,1)\), and \(\eps\in(0,1-\rho)\).  Let
\(\mathscr F_{\le r}\) be a folded profile class of entropy exponent \(h\).  Let
\(W=\F_q^s\), let \(m=sd\), assume
\[
       s\ge C\frac{h+r^2+1}{\eps},
\]
and assume \(\rho m\) is an integer.  Let \(N\le W^d\) be a uniformly random
\(\rho m\)-dimensional subspace.  Then, with probability at least
\[
   1-q^{-\Omega(\eps s d)},
\]
the following simultaneous statement holds: for every subspace \(S\subseteq N\),
the pair \((N,S)\) is AEL-robust
\((\rho+\eps,\mathscr F_{\le r})\) folded LCL up to dimension \(r\).
\end{theorem}

\begin{proof}
Let \(\alpha=\rho+\eps\).  If robustness fails for some \(S\subseteq N\), then
there are a coefficient space \(B\) of dimension \(t\le r\), a profile
\(\mathfrak V\in\mathscr F_t\), and a map \(T:B\to N\) whose image is nonzero
modulo \(S\), while every nonzero subprofile has negative potential at threshold
\(\alpha\).  Dropping the requirement that the image be nonzero modulo \(S\), it
is enough to rule out a nonzero map \(T:B\to W^d\) with \(T(B)\subseteq N\) and
with such a bad profile.

Fix \(t\le r\), a model coefficient space \(B=\F_q^t\), and a bad profile
\(\mathfrak V\in\mathscr F_t\).  For rank \(b\in\{1,\ldots,t\}\),
Lemma~\ref{rsd:lem:map-count} bounds the number of consistent rank-\(b\) maps by
\[
   q^{t^2+sd(1-\alpha)b}.
\]
For each fixed rank-\(b\) map, its image is a fixed \(b\)-dimensional subspace of
\(W^d\), and Lemma~\ref{rsd:lem:containment} gives
\[
   \Pr[T(B)\subseteq N]
      \le q^{-(1-\rho)sdb+O(1)}.
\]
Thus the contribution of rank \(b\) maps for the fixed profile is at most
\[
   q^{t^2+sd(1-\alpha)b-(1-\rho)sdb+O(1)}
   =q^{t^2+O(1)-\eps sdb}.
\]
There are at most \(q^{hd}\) profiles in \(\mathscr F_t\), and at most
\(q^{O(t^2)}\) changes of basis for transporting profiles from the model
coefficient space.  Union-bounding over profiles, bases, \(t\le r\), and
\(b\le t\), the failure probability is at most
\[
   \sum_{t=1}^r\sum_{b=1}^t
      q^{hd+O(r^2)-\eps sdb}
   \le q^{-\Omega(\eps s d)}
\]
provided the absolute constant \(C\) is large enough.  The estimate is uniform
over all \(S\subseteq N\), because \(S\) was only used in the discarded condition
\(q_N(T(B))\ne0\).
\end{proof}

\begin{corollary}[Random normalizers are AEL-robust relative designs]\label{rsd:cor:random-robust-relative-normalizer}
There is an absolute constant \(C'\) such that the following holds.  Fix
\(q\), \(r\in\N\), \(\rho\in(0,1)\), and \(\eps\in(0,1-\rho)\).  Let
\(W=\F_q^s\), let \(m=sd\), assume \(s\ge C'r^2/\eps\), and assume
\(\rho m\) is an integer.  Let \(N\le W^d\) be a uniformly random
\(\rho m\)-dimensional subspace.  Then, with probability at least
\[
   1-q^{-\Omega(\eps s d)},
\]
for every subspace \(S\subseteq N\), the pair \((N,S)\) is AEL-robust
\((\rho+\eps)\)-relative up to dimension \(r\).
\end{corollary}

\begin{proof}
Apply Theorem~\ref{rsd:thm:random-robust-normalizer} to the kernel-profile class
\(\mathscr K_{\le r}\).  Its entropy exponent is at most \(r^2\) by
Proposition~\ref{rsd:prop:kernel-specialization}, and the same proposition
identifies robustness for \(\mathscr K_{\le r}\) with AEL-robust relative
subspace design.
\end{proof}

\subsection{Random nested CSS inner gadgets}

Let \(m=sd\).  Fix two desired normalizer rates \(\rho_X,\rho_Z\in(0,1)\) with
\[
   R_{\rm in}\defeq \rho_X+\rho_Z-1>0.
\]
A CSS inner pair \(C_{2,\rm in}\subseteq C_{1,\rm in}\subseteq W^d\) with
\[
   \dim C_{1,\rm in}=\rho_X m,
   \qquad
   \dim C_{2,\rm in}=(1-\rho_Z)m
\]
has inner quantum rate \(R_{\rm in}\).  Its two relative pairs are
\[
   (N_X,S_X)=(C_{1,\rm in},C_{2,\rm in}),
   \qquad
   (N_Z,S_Z)=(C_{2,\rm in}^\perp,C_{1,\rm in}^\perp).
\]

\begin{theorem}[Existence of robust inner CSS gadgets]\label{rsd:thm:random-inner-css}
Fix \(q,r,\rho_X,\rho_Z\) with \(\rho_X+\rho_Z>1\), and fix
\[
   0<\eps<\min\{1-\rho_X,1-\rho_Z\}.
\]
Let \(\mathscr F^X_{\le r}\) and \(\mathscr F^Z_{\le r}\) be folded profile
classes of length-\(d\) profiles (Definition~\ref{rsd:def:bounded-profile-class})
of entropy exponent at most \(h\).  Let \(W=\F_q^s\), \(m=sd\), assume
\(s\ge C(h+r^2+1)/\eps\), where \(C\) is the constant from
Theorem~\ref{rsd:thm:random-robust-normalizer}, and
assume \(\rho_Xm\) and \((1-\rho_Z)m\) are integers.  Choose a random nested pair
\[
   C_{2,\rm in}\subseteq C_{1,\rm in}\subseteq W^d
\]
with
\[
   \dim C_{1,\rm in}=\rho_Xm,
   \qquad
   \dim C_{2,\rm in}=(1-\rho_Z)m.
\]
Then with probability at least \(1-q^{-\Omega(\eps sd)}\), the pair
\((C_{1,\rm in},C_{2,\rm in})\) is AEL-robust
\((\rho_X+\eps,\mathscr F^X_{\le r})\) folded LCL up to dimension \(r\),
and the pair \((C_{2,\rm in}^\perp,C_{1,\rm in}^\perp)\) is AEL-robust
\((\rho_Z+\eps,\mathscr F^Z_{\le r})\) folded LCL up to dimension \(r\).
In particular, such inner CSS gadgets exist.  For fixed
\((q,r,\eps,\rho_X,\rho_Z,d,s)\) and explicitly listed classes
\(\mathscr F^X_{\le r},\mathscr F^Z_{\le r}\), they can be found by
exhaustive search.
\end{theorem}

\begin{proof}
Choose \(C_{1,\rm in}\) uniformly among \(\rho_Xm\)-dimensional subspaces and
then \(C_{2,\rm in}\) uniformly among \((1-\rho_Z)m\)-dimensional subspaces of
\(C_{1,\rm in}\).  The marginal distribution of \(C_{1,\rm in}\) is uniform, so
Theorem~\ref{rsd:thm:random-robust-normalizer}, applied with the class
\(\mathscr F^X_{\le r}\), implies that, with probability
\(1-q^{-\Omega(\eps sd)}\), the pair \((C_{1,\rm in},S)\) is AEL-robust
\((\rho_X+\eps,\mathscr F^X_{\le r})\) folded LCL for every
\(S\subseteq C_{1,\rm in}\), hence for \(S=C_{2,\rm in}\).

The marginal distribution of \(C_{2,\rm in}\) is uniform among all subspaces of
dimension \((1-\rho_Z)m\): every such subspace is contained in the same number of
\(\rho_Xm\)-dimensional spaces.  Therefore \(C_{2,\rm in}^\perp\) is uniformly
distributed among \(\rho_Zm\)-dimensional subspaces of \(W^d\).  Applying
Theorem~\ref{rsd:thm:random-robust-normalizer} again, now with the class
\(\mathscr F^Z_{\le r}\), the pair \((C_{2,\rm in}^\perp,S)\) is AEL-robust
\((\rho_Z+\eps,\mathscr F^Z_{\le r})\) folded LCL for every
\(S\subseteq C_{2,\rm in}^\perp\), hence for \(S=C_{1,\rm in}^\perp\).  A union
bound proves the simultaneous statement.

For fixed parameters, the ambient space \(W^d\) has constant dimension.
Exhaustively enumerate nested pairs and check
Definition~\ref{rsd:def:ael-robust} by finite enumeration of coefficient
spaces, basis identifications, the listed profiles of the given classes, and
maps.  The probabilistic argument guarantees that the search succeeds.
\end{proof}

\begin{remark}[Kernel-profile specialization of the gadget theorem]\label{rsd:rem:gadget-kernel}
Taking \(\mathscr F^X_{\le r}=\mathscr F^Z_{\le r}=\mathscr K_{\le r}\), of
entropy exponent at most \(r^2\),
Proposition~\ref{rsd:prop:kernel-specialization} identifies the
conclusion with the statement that \((C_{1,\rm in},C_{2,\rm in})\) and
\((C_{2,\rm in}^\perp,C_{1,\rm in}^\perp)\) are AEL-robust relative subspace
designs up to dimension \(r\) with parameters \(\rho_X+\eps\) and
\(\rho_Z+\eps\); the hypothesis then reads \(s\ge C(2r^2+1)/\eps=O(r^2/\eps)\).
This is the inner hypothesis of Theorem~\ref{rsd:thm:ael-lift}.
\end{remark}

\section{The CSS-AEL construction}\label{sec:css-ael}

We now pass from constant-size robust inner gadgets to long CSS codes.  The construction is an Alon--Edmonds--Luby style edge-expander transform: the inner code controls each left vertex, the outer code controls the sequence of inner logical symbols, and the graph permutation regroups edge symbols into right folded coordinates. Our presentation follows the CSS-AEL template first introduced in~\cite{BGG24}. The new point here is that the transform turns inner LCL robustness and outer logical distance into a global LCL guarantee, of which the relative subspace-design property is the kernel-profile case.

\subsection{Explicitness convention}
A family of linear or CSS codes is called \emph{explicit} if there is a deterministic algorithm that, on input the block-length parameter, outputs generator or parity-check matrices in time polynomial in the block length.  All constants hidden in the notation may depend on the fixed design dimension \(r\), the target rate, and the accuracy parameter \(\eps\), but not on the block length.  When a constant-size object is shown to exist by the probabilistic method, exhaustive search over that constant universe is counted as explicit in this standard asymptotic sense.

\subsection{Expanders}

Let \(G=(V_L,V_R,E_G)\) be a \(d\)-regular bipartite graph with \(|V_L|=|V_R|=n\).  We allow a regular bipartite multigraph, which is the natural output of several lift-based Ramanujan constructions; the proof only uses the normalized biadjacency operator and a fixed ordering of incident edges.  The graph is a \(\lambda\)-spectral expander if the second singular value of its normalized biadjacency matrix is at most \(\lambda\).  For \(j\in V_L\), write
\[
  N_1(j),\ldots,N_d(j)\in V_R
\]
for its right neighbors in left-edge order, with multiplicity if the graph is a multigraph.

A \(d\)-regular bipartite Ramanujan graph has every nontrivial unnormalized singular value at most \(2\sqrt{d-1}\), so its normalized second singular value is at most \(2\sqrt{d-1}/d=O(d^{-1/2})\).  The original explicit algebraic constructions are due to Lubotzky--Phillips--Sarnak and Margulis, with Morgenstern extending the construction to \((q+1)\)-regular graphs for every prime power \(q\). In the parameter theorems we only use the consequence that explicit \(d\)-regular bipartite expanders with \(\lambda=O(d^{-1/2})\) are available.

\begin{lemma}[Expander averaging]\label{rsd:lem:expander-avg}
Let \(G\) be a \(d\)-regular bipartite \(\lambda\)-expander.  For \(x\in\R_{\ge0}^{V_R}\), set
\[
  \mu=\frac1n\sum_{i\in V_R}x_i,
  \qquad
  y_j=\frac1d\sum_{\ell=1}^d x_{N_\ell(j)}.
\]
Then
\[
  \frac1n\sum_{j\in V_L}(y_j-\mu)^2
     \le \lambda^2\|x\|_\infty^2.
\]
\end{lemma}

\begin{proof}
Let \(A_G\) be the normalized biadjacency matrix.  Then \(y=A_Gx\) and \(A_G\mathbf 1=\mathbf 1\).  Since \(x-\mu\mathbf 1\) is orthogonal to \(\mathbf 1\),
\[
   \|y-\mu\mathbf 1\|_2
     =\|A_G(x-\mu\mathbf 1)\|_2
     \le \lambda\|x-\mu\mathbf 1\|_2
     \le \lambda\sqrt n\|x\|_\infty.
\]
Dividing by \(n\) gives the claim.
\end{proof}

\subsection{Inner and outer CSS data}

Let the inner CSS pair be
\[
   C_{2,\rm in}\subseteq C_{1,\rm in}\subseteq W^d.
\]
Define
\[
  N_X^{\rm in}=C_{1,\rm in},\quad
  S_X^{\rm in}=C_{2,\rm in},
\]
and
\[
  N_Z^{\rm in}=C_{2,\rm in}^\perp,
  \quad
  S_Z^{\rm in}=C_{1,\rm in}^\perp.
\]
The inner logical alphabets are
\[
  L_X^{\rm in}=N_X^{\rm in}/S_X^{\rm in},
  \qquad
  L_Z^{\rm in}=N_Z^{\rm in}/S_Z^{\rm in},
\]
which are dual through the induced pairing.  Let
\[
   \ell\defeq \dim L_X^{\rm in}=\dim C_{1,\rm in}-\dim C_{2,\rm in}.
\]

The outer CSS pair is
\[
   D_2\subseteq D_1\subseteq (L_X^{\rm in})^n.
\]
Its dual sector is
\[
   D_2^\perp/D_1^\perp\subseteq (L_Z^{\rm in})^n,
\]
where orthogonality uses the induced coordinatewise pairing between \(L_X^{\rm in}\) and \(L_Z^{\rm in}\).  Define the outer logical distances
\[
   \delta_X^{\rm out}
      =\frac1n\min_{x\in D_1\setminus D_2}\wt(x),
   \qquad
   \delta_Z^{\rm out}
      =\frac1n\min_{z\in D_2^\perp\setminus D_1^\perp}\wt(z).
\]

\subsection{The construction}

Let \(\Pi_G:(W^d)^{V_L}\to (W^d)^{V_R}\) be the edge permutation that regroups edge labels from left vertices to right vertices.  Let
\[
   q_X:N_X^{\rm in}\to L_X^{\rm in},
   \qquad
   q_Z:N_Z^{\rm in}\to L_Z^{\rm in}
\]
be the quotient maps, extended component-wise to \(V_L\).

\begin{construction}[CSS-AEL pair]\label{rsd:constr:css-ael}
Define the global \(X\)-normalizer and \(X\)-stabilizer spaces by
\[
   C_{1,\rm AEL}
      \defeq
      \Pi_G\Bigl(\{y\in (N_X^{\rm in})^{V_L}:
          (q_X(y_j))_{j\in V_L}\in D_1\}\Bigr),
\]
and
\[
   C_{2,\rm AEL}
      \defeq
      \Pi_G\Bigl(\{y\in (N_X^{\rm in})^{V_L}:
          (q_X(y_j))_{j\in V_L}\in D_2\}\Bigr).
\]
The CSS-AEL code is the CSS code defined by the nested pair
\[
   C_{2,\rm AEL}\subseteq C_{1,\rm AEL}
      \subseteq (W^d)^{V_R}.
\]
Thus the final block length is \(n\) and the final folded alphabet is \(W^d\cong\F_q^{sd}\).
\end{construction}

The corresponding \(Z\)-normalizer and \(Z\)-stabilizer have the following explicit form.

\begin{proposition}[Exact CSS orthogonality]\label{rsd:prop:ael-orthogonality}
With notation as above,
\[
  C_{2,\rm AEL}^\perp
      =\Pi_G\Bigl(\{z\in (N_Z^{\rm in})^{V_L}:
          (q_Z(z_j))_{j\in V_L}\in D_2^\perp\}\Bigr),
\]
and
\[
  C_{1,\rm AEL}^\perp
      =\Pi_G\Bigl(\{z\in (N_Z^{\rm in})^{V_L}:
          (q_Z(z_j))_{j\in V_L}\in D_1^\perp\}\Bigr).
\]
In particular, \(C_{2,\rm AEL}\subseteq C_{1,\rm AEL}\) defines a valid CSS code.
\end{proposition}

\begin{proof}
The permutation \(\Pi_G\) is an isometry, so work in the left grouping.  Let
\[
  \widetilde C_2
  =\{y\in (N_X^{\rm in})^{V_L}:q_X(y)\in D_2\}.
\]
If \(z\in \widetilde C_2^\perp\), then \(z\) is orthogonal to \((S_X^{\rm in})^{V_L}\subseteq \widetilde C_2\).  Hence each left block \(z_j\) lies in
\[
   (S_X^{\rm in})^\perp=(C_{2,\rm in})^\perp=N_Z^{\rm in},
\]
so \(q_Z(z_j)\) is defined.  For \(y_j\in N_X^{\rm in}\) and \(z_j\in N_Z^{\rm in}\), the inner product descends to the logical pairing:
\[
   \ip{y_j}{z_j}_{W^d}
      =\ip{q_X(y_j)}{q_Z(z_j)}_{L_X,L_Z}.
\]
Therefore \(z\) is orthogonal to every \(y\) with \(q_X(y)\in D_2\) if and only if \(q_Z(z)\in D_2^\perp\).  This proves the first identity.  The second is identical with \(D_1\) in place of \(D_2\).  Since \(D_2\subseteq D_1\), the global pair is nested.
\end{proof}

\begin{proposition}[Rate]\label{rsd:prop:ael-rate}
Let
\[
  R_{\rm in}=\frac{\dim C_{1,\rm in}-\dim C_{2,\rm in}}{sd},
  \qquad
  R_{\rm out}=\frac{\dim D_1-
          \dim D_2}{n\ell},
\]
where \(\ell=\dim L_X^{\rm in}\).  Then the CSS-AEL code has quantum rate
\[
   R_{\rm AEL}=R_{\rm in}R_{\rm out}.
\]
\end{proposition}

\begin{proof}
The preimage of \(D_a\) under \((q_X)^{V_L}\) has dimension
\[
   n\dim S_X^{\rm in}+\dim D_a
\]
for \(a=1,2\).  Hence
\[
   \dim C_{1,\rm AEL}-\dim C_{2,\rm AEL}
      =\dim D_1-
       \dim D_2
      =n\ell R_{\rm out}.
\]
Since \(\ell=sd\,R_{\rm in}\), division by the physical dimension \(nsd\) gives \(R_{\rm AEL}=R_{\rm in}R_{\rm out}\).
\end{proof}

\subsection{The lifting theorem}

We state the AEL local-to-global step for folded LCL
profiles; this is the form used by the explicit instantiation.  The
subspace-design statement is its kernel-profile corollary.

\begin{definition}[Global folded LCL profile]\label{rsd:def:global-profile}
Let the final AEL alphabet be \(W^d\).  For a coefficient space \(B\), a
\emph{right profile} is a tuple
\[
      \mathfrak V=(\mathfrak V_i)_{i\in V_R},
      \qquad \mathfrak V_i\le\Hom(B,W).
\]
A map \(M:B\to (W^d)^{V_R}\) obeys this profile if, for every right vertex
\(i\in V_R\) and every edge coordinate \(h\in[d]\) in the right block, the edge
projection \(M_{i,h}:B\to W\) lies in \(\mathfrak V_i\).  For \(U\le B^*\), set
\begin{equation}\label{rsd:eq:global-edge-potential}
   \widehat\Phi^{W,\mathrm{edge}}_\alpha(U,\mathfrak V)
      =\frac1{sn}\sum_{i\in V_R}
        \dim\bigl(\mathfrak V_i\cap\Hom_U(B,W)\bigr)
       -(1-\alpha)\dim U .
\end{equation}
The normalization is by the final scalar length \(snd\): each right constraint is
applied to all \(d\) edge symbols in the block, so the factor \(d\) cancels.
\end{definition}

\begin{theorem}[AEL lift for folded LCL]\label{rsd:thm:ael-LCL-lift}
Fix \(r\in\N\) and \(\epsilon>0\), and let the inner edge alphabet be \(W\cong\F_q^s\).  Suppose the inner CSS pair satisfies:
\begin{enumerate}[label=(\roman*),leftmargin=2em]
\item \((N_X^{\rm in},S_X^{\rm in})\) is AEL-robust
\((\alpha_X,\mathscr F^X_{\le r})\) folded LCL up to dimension \(r\);
\item \((N_Z^{\rm in},S_Z^{\rm in})\) is AEL-robust
\((\alpha_Z,\mathscr F^Z_{\le r})\) folded LCL up to dimension \(r\).
\end{enumerate}
Suppose the outer CSS pair has logical distances \(\delta_X^{\rm out},\delta_Z^{\rm out}>0\), and put
\[
   \delta_{\rm out}=\min\{\delta_X^{\rm out},\delta_Z^{\rm out}\}.
\]
If \(G\) is a \(d\)-regular bipartite \(\lambda\)-expander with
\begin{equation}\label{rsd:eq:lambda-condition}
   q^{r^2}\frac{\lambda^2}{\epsilon^2}<\delta_{\rm out},
\end{equation}
then the following holds in the \(X\)-sector.  There is no folded LCL
witness \(M:B\to C_{1,\rm AEL}\), with \(1\le\dim B\le r\), such that
\(M\) is injective, \(M(B)\cap C_{2,\rm AEL}=0\), and \(M\) obeys a right profile
\(\mathfrak V=(\mathfrak V_i)_{i\in V_R}\) (Definition~\ref{rsd:def:global-profile}) satisfying
\begin{enumerate}[label=(\alph*),leftmargin=2em]
\item for every left vertex \(j\), the left-star profile
\[
      (\mathfrak V_{N_1(j)},\ldots,\mathfrak V_{N_d(j)})
\]
belongs to \(\mathscr F^X_{\dim B}\), and
\item every nonzero \(U\le B^*\) has
\[
      \widehat\Phi^{W,\mathrm{edge}}_{\alpha_X+\epsilon}(U,\mathfrak V)<0 .
\]
\end{enumerate}
The analogous statement holds in the \(Z\)-sector with
\((C_{2,\rm AEL}^\perp,C_{1,\rm AEL}^\perp)\), \(\alpha_Z\), and
\(\mathscr F^Z_{\le r}\).
\end{theorem}

\begin{proof}
We prove the \(X\)-sector statement; the \(Z\)-sector is identical after using the
explicit dual description in Proposition~\ref{rsd:prop:ael-orthogonality}, with
outer pair \(D_1^Z=D_2^\perp\) and \(D_2^Z=D_1^\perp\).

Assume that such a witness \(M:B\to C_{1,\rm AEL}\) exists.  Pull it back through
\(\Pi_G\) to the left grouping.  For each left vertex \(j\in V_L\), let
\[
   T_j:B\to N_X^{\rm in}
\]
be the linear map that returns the left block at \(j\), and define
\(\bar T_j=q_X\circ T_j\).  For every \(b\in B\), the word
\((\bar T_j(b))_{j\in V_L}\) lies in \(D_1\).  The induced map
\[
   \bar T:B\to D_1,
   \qquad
   b\mapsto (\bar T_j(b))_{j\in V_L},
\]
is injective modulo \(D_2\): if \(\bar T(b)\in D_2\), then the corresponding
global word lies in \(C_{2,\rm AEL}\), so \(M(b)\in M(B)\cap C_{2,\rm AEL}=\{0\}\).  Since \(M\) is injective,
this gives \(b=0\).  Thus \(\bar T(B)\cap D_2=0\),
and \(\bar T(B)\) is a nonzero outer logical subspace.

Fix a nonzero \(U\le B^*\), and put
\[
   x_i(U)=\frac1s\dim\bigl(\mathfrak V_i\cap\Hom_U(B,W)\bigr),
   \qquad i\in V_R .
\]
Then \(0\le x_i(U)\le\dim U\).  Let
\[
   \mu(U)=\frac1n\sum_{i\in V_R}x_i(U).
\]
The assumed negativity of the global potential gives
\[
   \mu(U)<(1-\alpha_X-\epsilon)\dim U .
\]
For a left vertex \(j\), define
\[
   y_j(U)=\frac1d\sum_{\ell=1}^d x_{N_\ell(j)}(U).
\]
By Lemma~\ref{rsd:lem:expander-avg} and Markov's inequality,
\[
   \Pr_{j\in V_L}\bigl[y_j(U)\ge (1-\alpha_X)\dim U\bigr]
      \le \frac{\lambda^2}{\epsilon^2}.
\]
There are at most \(q^{(\dim B)^2}\le q^{r^2}\) nonzero subspaces
\(U\le B^*\).  Hence, by a union bound and~\eqref{rsd:eq:lambda-condition}, for
more than \(1-\delta_{\rm out}\) of the left vertices \(j\), the inequality
\begin{equation}\label{rsd:eq:left-all-bad}
   \frac1d\sum_{\ell=1}^d
      \frac1s\dim\bigl(\mathfrak V_{N_\ell(j)}\cap\Hom_U(B,W)\bigr)
       <(1-\alpha_X)\dim U
\end{equation}
holds for every nonzero \(U\le B^*\).

Fix such a good left vertex \(j\).  Because the global right profile is obeyed, the 
local map \(T_j:B\to N_X^{\rm in}\) obeys the left-star profile
\((\mathfrak V_{N_1(j)},\ldots,\mathfrak V_{N_d(j)})\).  Inequality
\eqref{rsd:eq:left-all-bad} says exactly that every nonzero subspace has
negative potential at threshold \(\alpha_X\) for this left-star profile.
AEL-robustness of the inner \(X\)-pair therefore forces
\[
   q_X(T_j(B))=0.
\]
Equivalently, the \(j\)-th coordinate of the outer subspace
\(\bar T(B)\le D_1\) is zero.

Thus \(\bar T(B)\le D_1\), \(\bar T(B)\cap D_2=0\), and more than
\((1-\delta_{\rm out})n\) of its coordinates are identically zero.  Any nonzero
vector in \(\bar T(B)\) is therefore a word in \(D_1\setminus D_2\) of weight
less than \(\delta_{\rm out}n\), contradicting
\(\delta_X^{\rm out}\ge\delta_{\rm out}\).  This proves the \(X\)-sector claim.
\end{proof}

It is convenient to name the property established by the lifting theorem.

\begin{definition}[Folded LCL pair]\label{rsd:def:global}
Let \(G\) be the \(d\)-regular bipartite graph of
Construction~\ref{rsd:constr:css-ael}, let \(\mathscr F_{\le r}\) be a folded
profile class of length-\(d\) profiles over the edge alphabet \(W\), and let
\(S\subseteq N\subseteq (W^d)^{V_R}\) be a nested pair.  The pair \((N,S)\) is
\emph{\((\alpha,\mathscr F_{\le r})\) folded LCL up to dimension
\(r\) with respect to \(G\)} if there is no linear map \(M:B\to N\) with
\(1\le\dim B\le r\) such that \(M\) is injective, \(M(B)\cap S=0\), and \(M\)
obeys a right profile \(\mathfrak V=(\mathfrak V_i)_{i\in V_R}\)
(Definition~\ref{rsd:def:global-profile}) with the two properties that every
left-star profile \((\mathfrak V_{N_1(j)},\ldots,\mathfrak V_{N_d(j)})\),
\(j\in V_L\), belongs to \(\mathscr F_{\dim B}\), and that
\(\widehat\Phi^{W,\mathrm{edge}}_{\alpha}(U,\mathfrak V)<0\) for every nonzero
\(U\le B^*\).  A CSS pair \(C_2\subseteq C_1\subseteq (W^d)^{V_R}\) is
\emph{\((\alpha_X,\mathscr F^X_{\le r};\alpha_Z,\mathscr F^Z_{\le r})\) 
 folded LCL up to dimension \(r\)} if \((C_1,C_2)\) is
\((\alpha_X,\mathscr F^X_{\le r})\) folded LCL and
\((C_2^\perp,C_1^\perp)\) is \((\alpha_Z,\mathscr F^Z_{\le r})\) 
folded LCL, both up to dimension \(r\) and with respect to the same graph.
\end{definition}

In this language, the conclusion of Theorem~\ref{rsd:thm:ael-LCL-lift} is
precisely that the CSS-AEL code is
\((\alpha_X+\epsilon,\mathscr F^X_{\le r};\alpha_Z+\epsilon,\mathscr F^Z_{\le r})\) 
 folded LCL up to dimension \(r\) with respect to \(G\).  The explicit
instantiation in the final section is stated in this form.

We now state the lifting theorem
\begin{theorem}[CSS-AEL lifting theorem for relative subspace designs]\label{rsd:thm:ael-lift}
Fix \(r\in\N\) and \(\epsilon>0\).  Suppose the inner CSS pair satisfies:
\begin{enumerate}[label=(\roman*),leftmargin=2em]
\item \((N_X^{\rm in},S_X^{\rm in})\) is AEL-robust \(\alpha_X\)-relative subspace design up to dimension \(r\);
\item \((N_Z^{\rm in},S_Z^{\rm in})\) is AEL-robust \(\alpha_Z\)-relative subspace design up to dimension \(r\).
\end{enumerate}
Suppose the outer CSS pair has logical distances \(\delta_X^{\rm out},\delta_Z^{\rm out}>0\).  Let
\[
   \delta_{\rm out}=\min\{\delta_X^{\rm out},\delta_Z^{\rm out}\}.
\]
If \(G\) is a \(d\)-regular bipartite \(\lambda\)-expander satisfying
\eqref{rsd:eq:lambda-condition} with the present \(q,r,\epsilon,\delta_{\rm out}\), then the CSS-AEL code is an
\((\alpha_X+\epsilon,\alpha_Z+\epsilon)\)-CSS subspace design up to dimension \(r\).
\end{theorem}

\begin{proof}
We prove the \(X\)-sector statement.  Suppose, toward a contradiction, that
\((C_{1,\rm AEL},C_{2,\rm AEL})\) is not an \((\alpha_X+\epsilon)\)-relative design
up to dimension \(r\).  Lemma~\ref{rsd:lem:minimal-bad} gives a contained
relative profile \((V_i)_{i\in V_R}\) on a vector space \(V\), with
\(1\le\dim V\le r\), whose potential is negative on every nonzero subspace at
threshold \(\alpha_X+\epsilon\).  By
Proposition~\ref{rsd:prop:lcl-contained-profile}, this contained profile is an
injective-logical folded qLCL witness.  The coordinate-zero condition for
the right block is equivalent to requiring each of its \(d\) edge projections to
lie in the kernel profile \(\cZ_i(V_i)\).  Hence it is an LCL
witness for the kernel-profile class.

By Proposition~\ref{rsd:prop:kernel-specialization}, inner relative
robustness is the same as folded-LCL robustness for kernel profiles, and the 
potential is exactly the subspace-profile potential.  Theorem~\ref{rsd:thm:ael-LCL-lift}
therefore rules out the minimal bad witness.  This proves the \(X\)-sector.  The
\(Z\)-sector is identical using Proposition~\ref{rsd:prop:ael-orthogonality} and
the robust \(Z\)-inner pair.  Thus both relative design conditions hold.
\end{proof}

\section{Leverrier--Z\'emor quantum Tanner codes as outer codes}\label{rsd:sec:lz-outer}

The AEL lifting theorem uses the outer code only through three properties: CSS nesting, quantum rate, and positive block logical distance in both sectors.  It does not use LDPC structure, decoding, or any subspace-design property of the outer code.  Therefore any asymptotically good CSS family over the required outer alphabet can serve as the outer layer.  This separation is useful: all qLCL work is done by the constant-size inner gadget, while the outer family supplies only global logical distance.

For the end-to-end instantiation there is an alphabet mismatch that must be handled explicitly.  Leverrier--Z\'emor quantum Tanner codes are scalar binary CSS codes, whereas AEL asks for an outer code whose coordinate alphabet is the inner logical vector space.  We bridge this by scalar extension: tensor the binary outer code with the constant-dimensional binary space $L_X^{\rm in}$, and use the dual tensor extension on the $Z$ side.  The final code has folded physical alphabet $(\F_2^s)^d\cong\F_2^{sd}$; it is not claiming that the AEL outer symbols are single bits.

\subsection{Scalar extension of binary CSS outer codes}

Let \(L_X\) and \(L_Z\) be finite-dimensional dual vector spaces over \(\F_2\), equipped with a nondegenerate pairing \(L_X\times L_Z\to\F_2\).  We identify
\[
  \F_2^n\otimes L_X\cong L_X^n,
  \qquad
  \F_2^n\otimes L_Z\cong L_Z^n.
\]

\begin{lemma}[Vector-alphabet extension]\label{rsd:lem:scalar-extension}
Let \(E_2\subseteq E_1\subseteq \F_2^n\) be a binary nested CSS pair.  Define
\[
   D_1=E_1\otimes L_X\subseteq L_X^n,
   \qquad
   D_2=E_2\otimes L_X\subseteq L_X^n.
\]
Then \(D_2\subseteq D_1\) is a CSS pair over the vector alphabet \(L_X\), and
\[
   D_2^\perp=E_2^\perp\otimes L_Z,
   \qquad
   D_1^\perp=E_1^\perp\otimes L_Z .
\]
Its quantum rate equals that of \(E_2\subseteq E_1\):
\[
  \frac{\dim D_1-\dim D_2}{n\dim L_X}
    =\frac{\dim E_1-\dim E_2}{n}.
\]
Moreover its block logical distances satisfy
\[
  d_X(D_1,D_2)\ge d_X(E_1,E_2),
  \qquad
  d_Z(D_1,D_2)\ge d_Z(E_1,E_2).
\]
\end{lemma}

\begin{proof}
The formula for orthogonal complements follows from nondegeneracy of the coordinatewise pairing:
\((E\otimes L_X)^\perp=E^\perp\otimes L_Z\) for every \(E\le\F_2^n\).  The rate identity is immediate from \(\dim(E_a\otimes L_X)=\dim(E_a)\dim L_X\).

For distance, let \(x\in D_1\setminus D_2\).  The image of \(x\) in
\((E_1/E_2)\otimes L_X\) is nonzero.  Hence there is a linear functional \(\lambda\in L_X^*\) such that the coordinatewise contraction \((\operatorname{id}\otimes\lambda)(x)\) is nonzero in \(E_1/E_2\).  Thus \((\operatorname{id}\otimes\lambda)(x)\in E_1\setminus E_2\), while its Hamming support is contained in the block support of \(x\).  Therefore \(\wt_{\rm block}(x)\ge d_X(E_1,E_2)\).  The \(Z\)-distance proof is identical, using \(L_Z^*\) and the displayed formulas for duals.
\end{proof}

\subsection{The Leverrier--Z\'emor outer family}

We record the consequence of the Leverrier--Z\'emor quantum Tanner-code construction in the exact form needed here.

\begin{theorem}[LZ outer codes over the inner logical alphabet~\cite{LZ22}]\label{rsd:cor:lz-vector-outer}
Let \(L_X,L_Z\) be any dual pair of binary vector spaces.  For every \(\xi\in(0,1)\) and an infinite sequence of lengths \(n\to\infty\), the Leverrier--Z\'emor family yields an explicit outer CSS family
\[
   D_{2,n}\subseteq D_{1,n}\subseteq L_X^n
\]
of quantum rate at least \(1-\xi\) and block logical distances at least \(\delta_{\rm LZ}(\xi)n\) in both sectors, where $\delta_{\rm LZ}(\xi) = \poly(\xi)$.
\end{theorem}

\section{Explicit parameter instantiation}

We now choose parameters.  The preceding sections separated the proof into a constant-size inner gadget, expander mixing, and outer logical distance.  Following the classical LCL derandomization of Jeronimo--Shagrithaya~\cite{JS26}, the final theorems are stated once, at the level of arbitrary bounded-entropy folded profile classes; no separate subspace-design instantiation is stated, since the design and distance guarantees are the kernel-profile specialization recorded in Remark~\ref{rsd:rem:explicit-design}.

The classes must be supplied uniformly in the inner parameters, which the construction chooses last.

\begin{definition}[Bounded-entropy profile ensemble]\label{rsd:def:profile-ensemble}
Fix \(q\) and \(r\in\N\).  A \emph{folded profile ensemble \(\mathscr F\) of entropy exponent at most \(h\)} assigns to every fold alphabet \(W\cong\F_q^s\) and every inner length \(d\) a folded profile class \(\mathscr F_{\le r}[W,d]\) as in Definition~\ref{rsd:def:bounded-profile-class}, whose entropy exponent is at most \(h\), with \(h\) independent of \((s,d)\).  The ensemble is \emph{explicit} if there is an algorithm that, on input \((s,d)\), outputs the finite lists of profiles in \(\mathscr F_t[W,d]\) for all \(t\le r\).  All uses below have constant \((s,d)\), so these lists are constant-size objects in the sense of the explicitness convention.  The kernel-profile ensemble \(\mathscr K\), assigning to every \((W,d)\) the class \(\mathscr K_{\le r}\) of Proposition~\ref{rsd:prop:kernel-specialization}, is explicit with entropy exponent \(r^2\).
\end{definition}

The construction is explicit once the following constant-size and outer ingredients are fixed.

\begin{enumerate}[leftmargin=2em]
\item \textbf{Degree and expander.}  First choose the AEL degree \(d\) large enough that an explicit \(d\)-regular bipartite expander has second singular value \(\lambda\) satisfying~\eqref{rsd:eq:lambda-condition}.  Near-Ramanujan constructions give \(\lambda=O(d^{-1/2})\), so it suffices to take
\[
  d=O\left(\frac{q^{r^2}}{\epsilon^2\delta_{\rm out}}\right).
\]
\item \textbf{Inner gadget.}  After \(d\) is fixed, choose a nested CSS pair \(C_{2,\rm in}\subseteq C_{1,\rm in}\subseteq (\F_q^s)^d\) whose \(X\)- and \(Z\)-relative pairs are AEL-robust folded LCL up to dimension \(r\) for the two supplied classes.  Theorem~\ref{rsd:thm:random-inner-css} proves such a pair exists with \(s=O((h+r^2)/\eps)\).  Because \((s,d)\) are constants and the ensembles are explicit, exhaustive search produces a deterministic inner gadget.
\item \textbf{Outer CSS family.}  An explicit CSS family \(D_2\subseteq D_1\subseteq (L_X^{\rm in})^n\) of quantum rate close to one and constant block \(X\)- and \(Z\)-logical distances over the logical alphabet produced by the inner gadget.
\end{enumerate}

The same symbol \(d\) denotes the inner length and the graph degree; this is intrinsic to the AEL transform, since one inner coordinate is placed on each incident edge of a left vertex.

\begin{theorem}[Explicit LCL CSS family]\label{rsd:thm:explicit}
Fix \(q\), \(r,h\in\N\), explicit folded profile ensembles \(\mathscr F^X,\mathscr F^Z\) of entropy exponent at most \(h\) (Definition~\ref{rsd:def:profile-ensemble}), a target quantum rate \(R\in(0,1)\), and \(\eps\in(0,1/10)\).  Suppose that, for every constant logical alphabet produced by the inner gadgets below, there is an explicit outer CSS family of quantum rate at least \(1-\eps\) and block logical distances at least \(\delta_{\rm out} n\) in both sectors, where \(\delta_{\rm out}=\Omega(\eps)\) is a known constant independent of \(n\).  Then there are a fold alphabet \(W\cong\F_q^s\) and an inner length \(d\) with fold size
\[
   S_{\rm fold}=sd
      =\poly(r,h,1/\eps)\,q^{O(r^2)},
\]
together with an explicit family of folded \(q\)-ary CSS codes with folded block lengths \(n\to\infty\), folded alphabet \(\F_q^{S_{\rm fold}}\), and quantum rate at least \(R-O(\eps)\), such that every code in the family, with respect to its AEL graph, is
\[
   \bigl(\alpha_X,\mathscr F^X_{\le r}[W,d];\;\alpha_Z,\mathscr F^Z_{\le r}[W,d]\bigr)\text{ folded LCL up to dimension } r
\]
in the sense of Definition~\ref{rsd:def:global}, for thresholds
\[
   \alpha_X,\alpha_Z
      \le \frac{1+R}{2}+O(\eps).
\]
\end{theorem}

\begin{proof}
Set \(\epsilon=\eps/8\) and \(\eps_0=\eps/8\).  Choose the inner target rate
\[
   R_{\rm in}=\begin{cases}
      R+\eps, & R\le 1-2\eps,\\
      1-\eps, & R>1-2\eps,
   \end{cases}
   \qquad
   \rho_X=\rho_Z=\frac{1+R_{\rm in}}2.
\]
Then \(R_{\rm in}=R+O(\eps)\), \(1-R_{\rm in}\ge\eps\), and therefore
\[
   \eps_0<1-\rho_X=1-\rho_Z=\frac{1-R_{\rm in}}2,
\]
so the hypotheses of Theorem~\ref{rsd:thm:random-inner-css} can be met.  Choose the AEL degree first, taking
\[
   d=O\left(\frac{q^{r^2}}{\eps^2\delta_{\rm out}}\right)
\]
large enough for an explicit expander to satisfy~\eqref{rsd:eq:lambda-condition} with the above \(\epsilon\).  Then choose \(s=O((h+r^2)/\eps)\), increasing the constant if necessary so that \(s\ge C(h+r^2+1)/\eps_0\) for the constant \(C\) of Theorem~\ref{rsd:thm:random-robust-normalizer}, and so that the dimensions \(\rho_Xsd\) and \((1-\rho_Z)sd\) are integral.

Apply Theorem~\ref{rsd:thm:random-inner-css} with error parameter \(\eps_0\) and with the classes \(\mathscr F^X_{\le r}[W,d]\) and \(\mathscr F^Z_{\le r}[W,d]\).  This gives a constant-size inner CSS pair over \((\F_q^s)^d\) that is AEL-robust folded LCL for these classes with parameters
\[
   \alpha_X^{\rm in},\alpha_Z^{\rm in}
      \le \frac{1+R_{\rm in}}2+\eps_0.
\]
Because \((s,d)\) are independent of \(n\) and the ensembles are explicit, exhaustive search makes this inner pair explicit.  Now take the assumed explicit outer CSS family over the resulting inner logical alphabet, with rate \(R_{\rm out}\ge1-\eps\) and block logical distances at least \(\delta_{\rm out}n\) in both sectors.

Theorem~\ref{rsd:thm:ael-LCL-lift} shows that each code, with respect to its AEL graph, is folded LCL for the two classes up to dimension \(r\), at thresholds
\[
   \alpha_X,
   \alpha_Z
      \le \frac{1+R_{\rm in}}2+\eps_0+\epsilon
      \le \frac{1+R}{2}+O(\eps),
\]
where the last inequality also covers the case \(R>1-2\eps\).  By Proposition~\ref{rsd:prop:ael-rate}, the final quantum rate is
\[
   R_{\rm AEL}=R_{\rm in}R_{\rm out}
      \ge R-O(\eps).
\]
The fold size is \(sd\), and since \(\delta_{\rm out}=\Omega(\eps)\), the displayed bounds on \(s\) and \(d\) give \(S_{\rm fold}=\poly(r,h,1/\eps)q^{O(r^2)}\).
\end{proof}

\section{Consequences of Derandomization}
\cref{rsd:thm:explicit} shows that for folded profile ensembles \(\mathscr F^X,\mathscr F^Z\), we get explicit constructions of CSS codes that are $\bigl(\alpha_X,\mathscr F^X_{\le r}[W,d];\;\alpha_Z,\mathscr F^Z_{\le r}[W,d]\bigr)\text{ folded LCL up to dimension } r$ as long as the individual sector rates satisfy
\[
   \alpha_X,\alpha_Z
      \le \frac{1+R}{2}+O(\eps),
\]
where $R-O(\eps)$ is the quantum rate.
This immediately implies explicit constructions of quantum list decodable and list recoverable codes with optimal list sizes. We give details below.

\begin{theorem}\label{thm:exp-list-dec}
    Fix \(q\), \(r,h\in\N\), a rate $R_Q \in (0, 1)$ and explicit folded profile ensembles \(\mathscr F^X,\mathscr F^Z\) that contain the folded local profiles for $(\rho, L)$-list decodability, where $\rho$ satisfies:
    \[
        \rho < \frac{L}{L+1}\cdot \frac{(1-R_Q-\eps)}{2}
    \]
    for some constant $\eps>0$.
    Then, there exists an explicit family of folded q-ary CSS codes having quantum rate $R_Q-O(\eps)$ that are $(\rho, L)$-quantum list decodable, and over alphabet $\F_q^{S_{\rm fold}}$, where $S_{\rm fold}=\poly(L,1/\eps)\,q^{O(L^2)}$
\end{theorem}
\begin{proof}
    We take $r=L+1$ because every bad witness for $(\rho, L)$-quantum list decodability has locality at most $L+1$. Adding some slack to the rate $R_Q$ in \eqref{eq:list-dec-joint}, we get that robust inner CSS gadgets exist at rate $R_Q$ for \(\mathscr F^X,\mathscr F^Z\). Thus, we can invoke Theorem~\ref{rsd:thm:explicit} and obtain explicit $(\rho, L)$-quantum list decodable codes.
\end{proof}

\begin{theorem}\label{thm:exp-list-rec}
    Fix \(q\), \(r,h\in\N\), a rate $R_Q \in (0, 1)$ and explicit folded profile ensembles \(\mathscr F^X,\mathscr F^Z\) that contain the folded local profiles for $(\rho, \ell, L)$-list recoverability, where $\rho$ satisfies:
    \[
        \rho < \frac{(1-R_Q-\eps)}{2}
    \]
    for some constant $\eps>0$, and also
    \[
        L \le \left( \frac{\ell}{R+\eps}\right)^{(R+\eps)/\eps}.
    \]
    Then, there exists an explicit family of folded q-ary CSS codes having quantum rate $R_Q-O(\eps)$ that are $(\rho, \ell, L)$-quantum list decodable, and over alphabet $\F_q^{S_{\rm fold}}$, where $S_{\rm fold}=\poly(\ell, L,1/\eps)\,q^{O(L^2)}$.
\end{theorem}
\begin{proof}
    We take $r=L+1$ because every bad witness for $(\rho, \ell, L)$-quantum list recoverability has locality at most $L+1$. Taking the bound on rate $R_Q$ in \eqref{eq:list-rec-joint}, we get that robust inner CSS gadgets exist at rate $R_Q$ for \(\mathscr F^X,\mathscr F^Z\). Thus, we can invoke Theorem~\ref{rsd:thm:explicit} and obtain explicit $(\rho, \ell, L)$-quantum list recoverable codes.
\end{proof}

\begin{remark}[Kernel-profile specialization: designs and distance]\label{rsd:rem:explicit-design}
The subspace-design and distance guarantees are the kernel-profile case of the theorems above and are not stated separately.  Instantiate \(\mathscr F^X=\mathscr F^Z=\mathscr K\), the kernel-profile ensemble of entropy exponent \(r^2\).  If a code in the resulting family failed the \(\alpha_X\)-relative-design condition in the \(X\)-sector up to dimension \(r\), then, exactly as in the proof of Theorem~\ref{rsd:thm:ael-lift}, Lemma~\ref{rsd:lem:minimal-bad} and Proposition~\ref{rsd:prop:lcl-contained-profile} would produce an injective-logical map \(M\) obeying a kernel profile whose left-star profiles lie in the kernel-profile class and whose potential---equal to the subspace-profile potential by Lemma~\ref{rsd:lem:kernel-mass} and Proposition~\ref{rsd:prop:kernel-specialization}---is negative on every nonzero subspace at threshold \(\alpha_X\).  This is forbidden by the LCL property, and the \(Z\)-sector is identical.  Hence the families of Theorems~\ref{rsd:thm:explicit}, instantiated with \(\mathscr K\), are \((\alpha_X,\alpha_Z)\)-CSS subspace designs up to dimension \(r\) with \(\alpha_X,\alpha_Z\le(1+R)/2+O(\eps)\), and Proposition~\ref{rsd:prop:distance} gives folded block logical distances
\[
   \frac{d_X}{n},\frac{d_Z}{n}\ge\frac{1-R}{2}-O(\eps),
\]
measured in the folded coordinates of alphabet \(\F_q^{S_{\rm fold}}\).
\end{remark}

\section*{AI Disclosure}

The research direction and conceptual aspects of this work are all human.
The authors used versions of ChatGPT from 5.3 to 5.6 Pro for research exploration, 
proof writing, and editorial assistance with the manuscript. The authors take responsibility 
for the mathematical claims, proofs, and citations.

\bibliographystyle{alpha}
\bibliography{references}

\newcommand{\etalchar}[1]{$^{#1}$}
\begin{thebibliography}{BCDZ26b}

\bibitem[AEL95]{AEL95}
N.~Alon, J.~Edmonds, and M.~Luby.
\newblock Linear time erasure codes with nearly optimal recovery.
\newblock In {\em Proceedings of IEEE 36th Annual Foundations of Computer
  Science}, pages 512--519, 1995.

\bibitem[AGL24]{AGL24}
Omar Alrabiah, Venkatesan Guruswami, and Ray Li.
\newblock Randomly punctured reed-solomon codes achieve list-decoding capacity
  over linear-sized fields.
\newblock In Bojan Mohar, Igor Shinkar, and Ryan O'Donnell, editors, {\em
  Proceedings of the 56th Annual {ACM} Symposium on Theory of Computing, {STOC}
  2024, Vancouver, BC, Canada, June 24-28, 2024}, pages 1458--1469. {ACM},
  2024.

\bibitem[BCDZ26a]{BCDZ26b}
Joshua Brakensiek, Yeyuan Chen, Manik Dhar, and Zihan Zhang.
\newblock Combinatorial bounds for list recovery via discrete brascamp-lieb
  inequalities.
\newblock In Aditya Bhaskara and Artur Czumaj, editors, {\em Proceedings of the
  58th Annual {ACM} Symposium on Theory of Computing, {STOC} 2026, Salt Lake
  City, UT, USA, June 22-26, 2026}, pages 365--376. {ACM}, 2026.

\bibitem[BCDZ26b]{BCDZ26a}
Joshua Brakensiek, Yeyuan Chen, Manik Dhar, and Zihan Zhang.
\newblock From random to explicit via subspace designs with applications to
  local properties and matroids.
\newblock In Aditya Bhaskara and Artur Czumaj, editors, {\em Proceedings of the
  58th Annual {ACM} Symposium on Theory of Computing, {STOC} 2026, Salt Lake
  City, UT, USA, June 22-26, 2026}, pages 619--630. {ACM}, 2026.

\bibitem[BE21]{BreuckmannEberhardt2021}
Nikolas~P. Breuckmann and Jens~N. Eberhardt.
\newblock Balanced product quantum codes.
\newblock {\em IEEE Transactions on Information Theory}, 67(10):6653--6674,
  2021.

\bibitem[BGG24]{BGG24}
Thiago Bergamaschi, Louis Golowich, and Sam Gunn.
\newblock Approaching the quantum singleton bound with approximate error
  correction.
\newblock In Bojan Mohar, Igor Shinkar, and Ryan O'Donnell, editors, {\em
  Proceedings of the 56th Annual {ACM} Symposium on Theory of Computing, {STOC}
  2024, Vancouver, BC, Canada, June 24-28, 2024}, pages 1507--1516. {ACM},
  2024.

\bibitem[BJM{\etalchar{+}}24]{BJMST24}
Thiago Bergamaschi, Fernando~Granha Jeronimo, Tushant Mittal, Shashank
  Srivastava, and Madhur Tulsiani.
\newblock List decodable quantum {LDPC} codes.
\newblock {\em CoRR}, abs/2411.04306, 2024.

\bibitem[CS96]{CSS1}
A~Robert Calderbank and Peter~W Shor.
\newblock Good quantum error-correcting codes exist.
\newblock {\em Physical Review A}, 54(2):1098, 1996.

\bibitem[DHLV23]{DHLV23}
Irit Dinur, Min{-}Hsiu Hsieh, Ting{-}Chun Lin, and Thomas Vidick.
\newblock Good quantum {LDPC} codes with linear time decoders.
\newblock In Barna Saha and Rocco~A. Servedio, editors, {\em Proceedings of the
  55th Annual {ACM} Symposium on Theory of Computing, {STOC} 2023, Orlando, FL,
  USA, June 20-23, 2023}, pages 905--918. {ACM}, 2023.

\bibitem[EKZ20]{EvraKaufmanZemor2020}
Shai Evra, Tali Kaufman, and Gilles Z{\'e}mor.
\newblock Decodable quantum {LDPC} codes beyond the $\sqrt{n}$ distance barrier
  using high-dimensional expanders.
\newblock In {\em 61st IEEE Annual Symposium on Foundations of Computer Science
  (FOCS)}, pages 218--227, 2020.

\bibitem[GG26]{GG26}
Rohan Goyal and Venkatesan Guruswami.
\newblock Optimal proximity gaps for subspace-design codes and (random)
  reed-solomon codes.
\newblock In Aditya Bhaskara and Artur Czumaj, editors, {\em Proceedings of the
  58th Annual {ACM} Symposium on Theory of Computing, {STOC} 2026, Salt Lake
  City, UT, USA, June 22-26, 2026}, pages 1558--1568. {ACM}, 2026.

\bibitem[GGH26]{GGH26}
Rohan Goyal, Venkatesan Guruswami, and Jun-Ting Hsieh.
\newblock Explicit constant-alphabet subspace design codes, 2026.

\bibitem[GI02]{GI02}
Venkatesan Guruswami and Piotr Indyk.
\newblock Near-optimal linear-time codes for unique decoding and new
  list-decodable codes over small alphabets.
\newblock pages 812--821, 2002.

\bibitem[GJS27]{GayJeronimoShukul27}
William Gay, Fernando~Granha Jeronimo, and Abhi Shukul.
\newblock Explicit capacity-achieving quantum {LDPC} codes list decodable in
  near-linear time.
\newblock In {\em Proceedings of the 2027 Annual ACM-SIAM Symposium on Discrete
  Algorithms (SODA)}, 2027.
\newblock To appear.

\bibitem[GK13]{GK13}
Venkatesan Guruswami and Swastik Kopparty.
\newblock Explicit subspace designs.
\newblock In {\em 54th Annual {IEEE} Symposium on Foundations of Computer
  Science, {FOCS} 2013, Berkeley, CA, USA, October, 26-29, 2013}, pages
  608--617. {IEEE} Computer Society, 2013.

\bibitem[Got14]{Gottesman2014}
Daniel Gottesman.
\newblock Fault-tolerant quantum computation with constant overhead.
\newblock {\em Quantum Information and Computation}, 14(15--16):1338--1371,
  2014.

\bibitem[GPT23]{GuPattisonTang2023}
Shouzhen Gu, Christopher~A. Pattison, and Eugene Tang.
\newblock An efficient decoder for a linear distance quantum {LDPC} code.
\newblock In {\em Proceedings of the 55th Annual ACM Symposium on Theory of
  Computing (STOC)}, pages 919--932. ACM, 2023.

\bibitem[GX13]{GX13}
Venkatesan Guruswami and Chaoping Xing.
\newblock List decoding reed-solomon, algebraic-geometric, and gabidulin
  subcodes up to the singleton bound.
\newblock In Dan Boneh, Tim Roughgarden, and Joan Feigenbaum, editors, {\em
  Symposium on Theory of Computing Conference, STOC'13, Palo Alto, CA, USA,
  June 1-4, 2013}, pages 843--852. {ACM}, 2013.

\bibitem[HHO21]{HastingsHaahODonnell2021}
Matthew~B. Hastings, Jeongwan Haah, and Ryan O'Donnell.
\newblock Fiber bundle codes: Breaking the {$N^{1/2}\operatorname{polylog}(N)$}
  barrier for quantum {LDPC} codes.
\newblock In {\em Proceedings of the 53rd Annual ACM SIGACT Symposium on Theory
  of Computing (STOC)}, pages 1276--1288. ACM, 2021.

\bibitem[JMST25]{JMST25}
Fernando~Granha Jeronimo, Tushant Mittal, Shashank Srivastava, and Madhur
  Tulsiani.
\newblock Explicit codes approaching generalized singleton bound using
  expanders.
\newblock 2025.

\bibitem[JS26]{JS26}
Fernando~Granha Jeronimo and Nikhil Shagrithaya.
\newblock Probabilistic guarantees to explicit constructions: Local properties
  of linear codes.
\newblock In Aditya Bhaskara and Artur Czumaj, editors, {\em Proceedings of the
  58th Annual {ACM} Symposium on Theory of Computing, {STOC} 2026, Salt Lake
  City, UT, USA, June 22-26, 2026}, pages 806--813. {ACM}, 2026.

\bibitem[Kit03]{Kitaev2003}
A.~Yu. Kitaev.
\newblock Fault-tolerant quantum computation by anyons.
\newblock {\em Annals of Physics}, 303(1):2--30, 2003.

\bibitem[KMRS17]{KMRS17}
Swastik Kopparty, Or~Meir, Noga Ron{-}Zewi, and Shubhangi Saraf.
\newblock High-rate locally correctable and locally testable codes with
  sub-polynomial query complexity.
\newblock {\em J. {ACM}}, 64(2):11:1--11:42, 2017.

\bibitem[KRSW23]{KRZSW23}
Swastik Kopparty, Noga Ron{-}Zewi, Shubhangi Saraf, and Mary Wootters.
\newblock Improved list decoding of folded reed-solomon and multiplicity codes.
\newblock {\em {SIAM} J. Comput.}, 2023.

\bibitem[KT21]{KaufmanTessler2021}
Tali Kaufman and Ran~J. Tessler.
\newblock New cosystolic expanders from tensors imply explicit quantum {LDPC}
  codes with {$\Omega(\sqrt{n}\log^k n)$} distance.
\newblock In {\em Proceedings of the 53rd Annual ACM SIGACT Symposium on Theory
  of Computing (STOC)}, pages 1317--1329. ACM, 2021.

\bibitem[LMS25]{LMS25}
Matan Levi, Jonathan Mosheiff, and Nikhil Shagrithaya.
\newblock Random reed-solomon codes and random linear codes are locally
  equivalent.
\newblock In {\em 66th {IEEE} Annual Symposium on Foundations of Computer
  Science, {FOCS} 2025, Sydney, Australia, December 14-17, 2025}, pages
  2097--2131. {IEEE}, 2025.

\bibitem[LS08]{LS08}
Debbie Leung and Graeme Smith.
\newblock Communicating over adversarial quantum channels using quantum list
  codes.
\newblock {\em IEEE Transactions on Information Theory}, 54(2):883--887, 2008.

\bibitem[LS25]{LS25}
Ray Li and Nikhil Shagrithaya.
\newblock Near-optimal list-recovery of linear code families.
\newblock {\em CoRR}, abs/2502.13877, 2025.

\bibitem[LZ22]{LZ22}
Anthony Leverrier and Gilles Z{\'{e}}mor.
\newblock Quantum tanner codes.
\newblock In {\em 63rd {IEEE} Annual Symposium on Foundations of Computer
  Science, {FOCS} 2022, Denver, CO, USA, October 31 - November 3, 2022}, pages
  872--883. {IEEE}, 2022.

\bibitem[PK22a]{PK22}
Pavel Panteleev and Gleb Kalachev.
\newblock Asymptotically good quantum and locally testable classical {LDPC}
  codes.
\newblock In Stefano Leonardi and Anupam Gupta, editors, {\em {STOC} '22: 54th
  Annual {ACM} {SIGACT} Symposium on Theory of Computing, Rome, Italy, June 20
  - 24, 2022}, pages 375--388. {ACM}, 2022.

\bibitem[PK22b]{PanteleevKalachev2022}
Pavel Panteleev and Gleb Kalachev.
\newblock Quantum {LDPC} codes with almost linear minimum distance.
\newblock {\em IEEE Transactions on Information Theory}, 68(1):213--229, 2022.

\bibitem[Ste96]{CSS2}
Andrew~M Steane.
\newblock Error correcting codes in quantum theory.
\newblock {\em Physical Review Letters}, 77(5):793, 1996.

\bibitem[TZ14]{TillichZemor2014}
Jean-Pierre Tillich and Gilles Z{\'e}mor.
\newblock Quantum {LDPC} codes with positive rate and minimum distance
  proportional to the square root of the blocklength.
\newblock {\em IEEE Transactions on Information Theory}, 60(2):1193--1202,
  2014.

\end{thebibliography}
\end{document}